\documentclass[11pt]{article}
\usepackage{amsmath,amsfonts,amsthm,amssymb,color}
\usepackage{fancybox}
\usepackage{framed}
\usepackage{comment}
\usepackage{mathtools}
\usepackage{tikz}
\usepackage{mathrsfs}
\usetikzlibrary{arrows}
\usepackage{graphicx}
\usepackage{subcaption}

\usepackage{thmtools}

\usepackage[margin=1in]{geometry}

\newtheorem{theorem}{Theorem}[section]
\newtheorem{corollary}[theorem]{Corollary}
\newtheorem{lemma}[theorem]{Lemma}

\newtheorem{fact}[theorem]{Fact}

\theoremstyle{definition}
\newtheorem{definition}[theorem]{Definition}
\newtheorem{remark}[theorem]{Remark}

\newenvironment{fminipage}%
  {\begin{Sbox}\begin{minipage}}%
  {\end{minipage}\end{Sbox}\fbox{\TheSbox}}

\def\dim#1{\mathrm{dim} (#1)}

\newcommand\Exp{E}

\newcommand\zero{\mathrm{zero}}

\newcommand\trans{\mathrm{Trans}}

\newcommand\matrixspaceA{{\mathfrak{A}}}
\newcommand\matrixspaceB{{\mathfrak{B}}}
\newcommand\matrixspaceC{{\mathfrak{C}}}

\newcommand\matrixspaceG{{\mathfrak{G}}}
\newcommand\matrixspaceH{{\mathfrak{H}}}
\newcommand\matrixspaceX{{\mathfrak{X}}}
\newcommand\matrixspaceY{{\mathfrak{Y}}}
\newcommand\matrixspaceZ{{\mathfrak{Z}}}

\renewcommand\SS{{\mathrm{SS}}}
\newcommand\GL{{\mathrm{GL}}}

\renewcommand\AA{{\mathcal{A}}}
\newcommand\BB{{\mathcal{B}}}
\newcommand\FF{{\mathcal{F}}}
\newcommand\GG{{\mathcal{G}}}
\newcommand\HH{{\mathcal{H}}}

\newcommand\semic{\mathbf{SC}}

\newcommand{\rowspan}{\mathrm{rowspan}}
\newcommand{\rank}[1]{\mathrm{rank} \left(#1\right)}

\title{Faster Isomorphism for $p$-Groups of Class 2 and Exponent $p$}
\author{Xiaorui Sun\thanks{\texttt{xiaorui@uic.edu}. 
University of Illinois at Chicago.}} 
\date{}
\begin{document}
{
\clearpage
\maketitle 
\thispagestyle{empty}

\setcounter{page}{0}

\begin{abstract}
The group isomorphism problem determines whether two groups, given by their Cayley tables, are isomorphic. For groups with order $n$, an algorithm with $n^{(\log n + O(1))}$ running time, attributed to Tarjan,  was proposed in the 1970s [Mil78]. Despite the extensive study over the past decades, the current best group isomorphism algorithm has an  $n^{(1 / 4 + o(1))\log n}$ running time [Ros13].

The isomorphism testing for $p$-groups of (nilpotent) class 2 and exponent $p$ has been identified as a major barrier to obtaining an $n^{o(\log n)}$ time algorithm for the group isomorphism problem. Although the $p$-groups of class 2 and exponent $p$ have much simpler algebraic structures than general groups, the best-known isomorphism testing algorithm for this group class also has an $n^{O(\log n)}$ running time. 

In this paper, we present an isomorphism testing algorithm for $p$-groups of class 2 and exponent $p$ with running time $n^{O((\log n)^{5/6})}$ for any prime $p > 2$. 
Our result is based on a novel reduction to the skew-symmetric matrix tuple isometry problem [IQ19]. To obtain the reduction, we develop several tools for matrix space analysis, including a matrix space individualization-refinement method and a characterization of the low rank matrix spaces.

\end{abstract}
}

\newpage 

\section{Introduction}
The group isomorphism problem is to determine whether two groups, given by their Cayley (multiplication) tables, are isomorphic. 
The problem is among a few classes of problems in NP that are not known to be solvable in polynomial time or NP-Complete~\cite{garey1979computers}. 
The group isomorphism problem and its variants have close connections to cryptography, computational group theory, and algebraic complexity theory~\cite{brooksbank2019incorporating}.
Furthermore, following Babai's breakthrough on the quasi-polynomial time algorithm for graph isomorphism~\cite{babai-quasipolynomial, babai2019canonical}, 
group isomorphism has become a bottleneck for the $n^{o(\log n)}$ time algorithm of graph isomorphism because
group isomorphism reduces to graph isomorphism.

The group isomorphism problem has been extensively studied since the 1970s~\cite{felsch1970programme, miller-latin, savage1980n2, o1994isomorphism, vikas1996ano, kavitha2007linear, gall2008efficient, wilson2009decomposing, wilson2009finding,
babai2011code, babai2012polynomial, babai2012polynomial_1, brooksbank2012computing, lewis2012isomorphism, qiao2012isomorphism, rosenbaum2013bidirectional, brooksbank2015fast,
luks2015group, rosenbaum2015beating, li2017linear, brooksbank2020improved, grochow2021p, dietrich2022group, grochow2021complexity}. 
A simple algorithm for group isomorphism, attributed to Tarjan,
picks a generating set in one of the groups and checks for all possible images of the generating set in the other group, whether the partial correspondence extends to an isomorphism~\cite{miller-latin}. 
%
%
Since every group of order $n$ has a generating set of size at most $\log_2 n$, this algorithm results in an $n^{\log_2 n+O(1) }$ running time. 
The current best-known algorithm for the group isomorphism problem has an $n^{(1/4 + o(1))\log_2 n}$ running time~\cite{rosenbaum2013bidirectional}.




It is long believed that the isomorphism testing of $p$-groups of class 2 and exponent $p$ is a major bottleneck for the group isomorphism problem~\cite{lewis2012isomorphism, babai2011code,brooksbank2012computing, rosenbaum2013bidirectional, brooksbank2015fast, li2017linear, brooksbank2019incorporating}. 
A group $G$ is a $p$-group of (nilpotent) class 2 and exponent $p$ for some prime number $p$ if 
every element except the identity has an order of $p$, and 
$G$ is not abelian but $[G, [G, G]]$ only contains the identity element, where $[G, H]$ denotes the 
group generated by 
$xyx^{-1}y^{-1}$ for all $x\in G, y \in H$.


The best-known algorithm for the isomorphism testing of $p$-groups of class 2 and exponent $p$ does not have a major advantage in the running time, being $n^{O(\log_2 n)}$~\cite{rosenbaum2013bidirectional},  over the general groups, 
even though 
the structure of $p$-groups of class 2 and exponent $p$ was well understood~\cite{baer1938groups, webb1983rank, wilson2009decomposing, wilson2009finding},
and the isomorphism testing of this group class has been studied in depth~\cite{lewis2012isomorphism, brooksbank2012computing, rosenbaum2013bidirectional, brooksbank2015fast, li2017linear, brooksbank2019incorporating, schrock2019complexity}. 
Hence, to develop a better algorithm for isomorphism testing of general groups, 
it is necessary to provide a faster algorithm for $p$-groups of class 2 and exponent $p$.


\subsection{Our result}

In this paper, 
we present an isomorphism testing algorithm for $p$-groups of class 2 and exponent $p$ with $n^{o(\log n)}$ running time for any odd prime $p$.


\begin{theorem}\label{thm:main_group}
Let $G$ and $H$ be two groups of order $n$. 
If both $G$ and $H$ are $p$-groups of class 2 and exponent $p$ for some prime number $p > 2$,
then   
given the Cayley tables of $G$ and $H$, there is an algorithm with running time $n^{O((\log n)^{5/6})}$ to determine whether $G$ and $H$ are isomorphic. 
\end{theorem}

Theorem~\ref{thm:main_group} utilizes the Baer's correspondence~\cite{baer1938groups},
which reduces  
the group isomorphism problem for $p$-groups of class 2 and exponent $p$ to the isometry testing problem of skew-symmetric matrix spaces. 

A square matrix $A$ is a skew-symmetric matrix if $A^T = -A$.
In the isometry testing problem for skew-symmetric matrix spaces,
the input consists of the linear bases of two skew-symmetric matrix spaces $\matrixspaceA$ and $\matrixspaceB$.
The problem is to
decide whether there is an isometry $S$ from $\matrixspaceA$ to $\matrixspaceB$, 
i.e., an invertible matrix  $S$ such that 
$S \matrixspaceA S^T = \matrixspaceB$, where $S\matrixspaceA S^T$ is the linear span of the matrices $S A S^T$ for all the matrices $A \in \matrixspaceA$.
We prove the following result for the isometry testing problem of skew-symmetric matrix spaces. 

\begin{theorem}\label{thm:main_matrix_space}
Let $\matrixspaceA$ and $\matrixspaceB$ be two linear matrix spaces, both of dimension $m$,
such that every matrix in $\matrixspaceA$ or $\matrixspaceB$ is an $n\times n$ skew-symmetric matrix over $\mathbb{F}_p$ for some prime number $p > 2$ and positive integers $m, n$.
There is an algorithm with running time $p^{O((n+m)^{1.8} \cdot \log (p))}$ to determine whether there is an invertible $n\times n$ matrix $S$ over $\mathbb{F}_p$ such that 
$S \matrixspaceA S^T = \matrixspaceB$.
\end{theorem}

We obtain 
Theorem~\ref{thm:main_matrix_space} by combining several new tools to analyze matrix spaces, 
including an individualization-refinement method for matrix spaces, a characterization of low rank matrix spaces, and a reduction from the isometry testing of skew-symmetric matrix spaces to the isometry testing of skew-symmetric matrix tuples~\cite{ivanyos2019algorithms}.



To obtain Theorem~\ref{thm:main_group}, let $k$ denote $\log_p (n)$.
We apply Theorem~\ref{thm:main_matrix_space} for the case of $k > (\log_2( p))^5$ by constructing the skew-symmetric matrix spaces for both input groups according to the Baer's correspondence~\cite{baer1938groups}. Theorem~\ref{thm:main_matrix_space} implies the running time for this case is $n^{O((\log n)^{5/6})}$.
For the case of $k \leq (\log_2( p))^5$, we run the aforementioned generating set enumeration algorithm~\cite{miller-latin}.
Because every $p$ group of order $p^k$ has a generating set of size at most $k$, the running time of the algorithm for this case is $p^{O(k^2)}$, 
which is also $n^{O((\log n)^{5/6})}$. 



\subsection{Related work}
The group isomorphism problem has been studied for variant group classes.
Polynomial time algorithms have been developed for abelian groups~\cite{kavitha2007linear,savage1980n2,vikas1996ano}, 
groups formed by semidirect products of an abelian group and a cyclic group~\cite{gall2008efficient, wilson2009decomposing, wilson2009finding},
groups with normal Hall subgroups~\cite{qiao2012isomorphism},  
groups with abelian Sylow towers~\cite{babai2012polynomial_1}, 
and groups with no abelian normal subgroups~\cite{babai2012polynomial}.
Dietrich and Wilson recently showed that the group isomorphism problem can be solved in nearly linear time for most orders~\cite{dietrich2022group}.

For $p$-groups of class 2 and exponent $p$, 
algorithms for some nontrivial subclasses of this group class have been proposed~\cite{lewis2012isomorphism, brooksbank2012computing, brooksbank2015fast}. 
Li and Qiao showed that 
if the $p$-groups of class 2 and exponent $p$ are generated randomly, then the isomorphism testing problem can be solved in polynomial time in the average case~\cite{li2017linear}. 
In \cite{brooksbank2019incorporating}, the average case running time was further improved to linear. 
In this work, we focus on the isomorphism testing for $p$-groups of class 2 and exponent $p$ in the worst case. 

The refinement methods, such as the naive refinement~\cite{babai1980random} and Weisfeiler-Leman refinement~\cite{weisfeiler-lehman}, have been powerful tools for the graph isomorphism problem. 
The refinement methods have been successfully used for graph isomorphism testing algorithms~\cite{bab-srg,
babai1980random, babai-pcc, babai1983canonical, zemlyachenko1985graph, spielman,
datta2009planar, bcstw, bw-stoc13, cst,sw-pcc, lokshtanov2017fixed, grohe2019canonisation, kiefer2019weisfeiler,
wiebking2019graph, grohe2020faster,
grohe2020isomorphism,neuen2022isomorphism}, including the celebrated quasi-polynomial time algorithm for graph isomorphism~\cite{babai-quasipolynomial, babai2019canonical}.


The refinement approach does not extend to groups in a naive way. 
Several representations of groups that allow refinement have been proposed recently.
In~\cite{brooksbank2019incorporating}, the authors defined a hypergraph using recursively refinable filters and proposed applying the Weisfeiler-Leman refinement on the hypergraph.
Brachter and Schweitzer proposed defining colors of group element tuples by
group operation patterns of the elements involved in the tuple and applying the  Weisfeiler-Leman refinement  to refine the colors of element tuples~\cite{brachter2020weisfeiler}.
Both approaches can distinguish between several non-isomorphic constructions of $p$-groups of class 2 and exponent $p$. 
However, it was unclear how these refinement methods could be used to develop faster worst case isomorphism testing algorithms. 

The isometry testing of skew-symmetric matrix spaces was studied in~\cite{li2017linear, brooksbank2020improved, grochow2021p, grochow2021complexity}. 
Its applications 
in cryptography were investigated in~\cite{brooksbank2019incorporating, ji2019general, tang2022practical}. 




\subsection{Technique overview}
We provide an overview of the algorithm for the isometry testing of skew-symmetric matrix spaces (Theorem~\ref{thm:main_matrix_space}).

\paragraph{Isometry Testing for Skew-Symmetric Matrix Tuples}
We start by introducing the skew-symmetric matrix tuple isometry problem, which is related to our problem.
A skew-symmetric matrix tuple $\AA = (A_1, \dots, A_k)$ of length $k$ is a sequence of $k$ skew-symmetric matrices of the same dimensions. 
 For matrices $P$ and $Q$, we use $P\AA Q$ to denote the matrix tuple $(P A_1 Q, P A_2 Q, \dots, P A_k Q)$.

Similar to the isometry between two skew-symmetric matrix spaces, we also define the isometry between two skew-symmetric matrix tuples. 
For two skew-symmetric matrix tuples $\AA$ and $\BB$, a matrix $S$ is an isometry from $\AA$ to $\BB$ if $S\AA S^T = \BB$, i.e., $S A_i S^T = B_i$ for all the $A_i \in \AA$.
The isometry problem of skew-symmetric matrix tuples determines whether there is an isometry between two input skew-symmetric matrix tuples. 
The difference between the skew-symmetric matrix tuple isometry problem and the skew-symmetric matrix space isometry problem is that 
the correspondence between matrices from two matrix tuples is fixed by the indices of the matrices, but 
for matrix spaces, no such correspondence is given. 
Ivanyos and Qiao presented a polynomial time algorithm for the isometry testing of skew-symmetric matrix tuples~\cite{ivanyos2019algorithms}. 

\begin{theorem}[Theorem 1.7 of \cite{ivanyos2019algorithms}]\label{thm:skew_tuple_isometry_intro}
Let $\AA = (A_1, \dots, A_k)$ and $\BB = (B_1, \dots, B_k)$ be two skew-symmetric matrix tuples of length $k$ such that the matrices in $\AA$ and $\BB$ are of dimension $n \times n$ over $\mathbb{F}_p$ for some prime $p > 2$. 
There is an algorithm with running time $\mathrm{poly}(n, k, p)$ to determine whether there is an isometry from $\AA$ to $\BB$. 
If yes, the algorithm also returns an isometry from $\AA$ to $\BB$. 
\end{theorem}

Our approach for the isometry testing of skew-symmetric matrix spaces is obtained by providing a $p^{O((n+m)^{1.8} \cdot \log p)}$ time reduction to the  skew-symmetric matrix tuple isometry problem, 
where $m$ is the dimension of the matrix space, and $n$ is the number of rows or columns for each square matrix in the matrix space.

\paragraph{Individualization-refinement for matrix spaces}
One powerful technique for graph isomorphism is the individualization-refinement method~\cite{bab-srg, babai1980random, babai-pcc, zemlyachenko1985graph, spielman, bcstw, bw-stoc13, cst, sw-pcc, babai-quasipolynomial}. 
For graphs, 
the individualization-refinement method first chooses a set of a small number of vertices and assigns each chosen vertex 
a distinct vertex color,
and then it refines the vertex colors by assigning distinguished vertices different colors in a canonical way until vertices of the same color cannot be further distinguished. 


A natural question for the group isomorphism problem is whether it is possible to define  individualization-refinement operations for group isomorphism. 
Based on the connection between group isomorphism for $p$-groups of class 2 and exponent $p$ and the skew-symmetric matrix space isometry problem~\cite{baer1938groups}, 
Li and Qiao proposed a matrix space individualization-refinement method, which follows the individualization-refinement for random graphs~\cite{babai1980random}, and 
analyzed the isometry testing of skew-symmetric matrix spaces in the average case~\cite{li2017linear}.




In this work, we propose a different matrix space individualization-refinement to enable the analysis of the isometry of skew-symmetric matrix spaces in the worst case.
Consider an $m \times n$ matrix space $\matrixspaceA$. 
The individualization in our scenario is defined by a left individualization matrix $L$ 
and a right individualization matrix $R$, where $L$ is a matrix with $m$ columns and $R$ is a matrix with $n$ rows. 
In the refinement, we compute $LAR$ for each matrix $A\in \matrixspaceA$. If  $LA'R$ does not equal $LA''R$ for some $A', A'' \in \matrixspaceA$, then $A'$ and $A''$ are distinguished.  

Ideally, if $LA'R$ does not equal $LA''R$ for any two matrices $A', A'' \in \matrixspaceA$, then each matrix $A$ in the space can be uniquely identified by $LAR$, and thus all the matrices in $\matrixspaceA$ are distinguished. 
Consider two isometric skew-symmetric matrix spaces $\matrixspaceA$ and $\matrixspaceB$.
Let  $L_\matrixspaceA$ and  $R_\matrixspaceA$ be individualization matrices for $\matrixspaceA$ that distinguish all the matrices in $\matrixspaceA$.
Let   $L_\matrixspaceB$ and  $R_\matrixspaceB$  be individualization matrices for $\matrixspaceB$ 
such that 
 $L_\matrixspaceB$ equals $L_\matrixspaceA S^{-1}$, and $R_\matrixspaceB$ equals $(S^T)^{-1} R_\matrixspaceA$ for some isometry $S$ from $\matrixspaceA$ to $\matrixspaceB$. 
One can distinguish all the matrices in both spaces by their individualization matrices and then establish a bijection between the matrices in the two spaces. 
Thus the skew-symmetric matrix space isometry problem reduces to the skew-symmetric matrix tuple isometry problem, which can be efficiently solved by Theorem~\ref{thm:skew_tuple_isometry_intro}.
Furthermore, suppose $L_\matrixspaceA$ contains a small number of rows and $R_\matrixspaceA$ contains a small number of columns.
Then one can solve the skew-symmetric matrix space isometry problem efficiently by enumerating all the possible corresponding $L_\matrixspaceB$ and  $R_\matrixspaceB$.

We show that
the number of rows for the left individualization matrices  and the number of columns for the right individualization matrices are related to the rank of matrices in the matrix space.
More specifically, we show that for a matrix space of dimension $d$ and any parameter $k$, 
there exist left and right individualization matrices $L$ and $R$ with $O(\max\{d \log(p), k\} / \sqrt{k})$ rows and columns, respectively, such that 
for each matrix $A$ in the matrix space with rank at least $k$,
$LAR$ is a non-zero matrix (Lemma~\ref{lem:individualzation_main}). 
In other words, if every matrix (except the zero matrix) in a skew-symmetric matrix space is of high rank, then the skew-symmetric matrix space isometry problem reduces to the skew-symmetric matrix tuple isometry problem efficiently.




\paragraph{Low rank matrix space characterization}
The hard case for the matrix space individualization/refine method is that 
there are some matrices $A$ in the space such that $LAR$ are zero matrices. 
Because of the linearity, such matrices form a linear subspace of the original matrix space. 
To tackle this hard case, 
we characterize the structure of the matrix space in which every matrix is of low rank. 
Such a matrix space is called a low rank matrix space.

As our main technical result for the low rank matrix space characterization, 
we show that, for  a matrix space  $\matrixspaceA$ such that every matrix in the space is of rank at most $r$,
there are invertible matrices $P$ and $Q$, called left and right formatting matrices, such that 
for each $A \in \matrixspaceA$, $PAQ$ has non-zero entries only in the last $O(r^2)$ rows or 
columns (Lemma~\ref{lem:low_rank_main}). 
Furthermore, if $\matrixspaceA$ is a skew-symmetric matrix space, 
then $Q = P^T$. 

Together with matrix space individualization-refinement, we can represent a matrix space in a more structured way.
First, we construct a ``semi-canonical'' basis for the input matrix space. 
Suppose we apply left and right individualization matrices $L$ and $R$ to a matrix space $\matrixspaceA$ of dimension $d$ and compute a linear basis $(A_1, \dots, A_d)$ of $\matrixspaceA$
such that $(L A_1 R, L A_2 R, \dots, LA_d R)$ is lexically minimized among all the linear basis of $\matrixspaceA$. 
Because the zero matrix is lexically the smallest among all the matrices, 
the first few matrices in the semi-canonical basis correspond to a linear basis of $\matrixspaceC$, which is the linear span of all the matrices $A \in \matrixspaceA$ such that $LAR$ is a zero matrix.


We further apply formatting matrices $P$ and $Q$ for $\matrixspaceC$ to each matrix in the semi-canonical basis of $\matrixspaceA$ (every matrix $A$ in the semi-canonical basis becomes $PAQ$). 
Then by our low rank matrix space characterization, the matrices that form a linear basis of $\matrixspaceC$
have non-zero entries only in the last few rows or columns. See Figure~\ref{fig:ind_low_rank_one_dim_intro} for an illustration.

\begin{figure}[h]
\begin{center}
\includegraphics[width=0.5\textwidth]{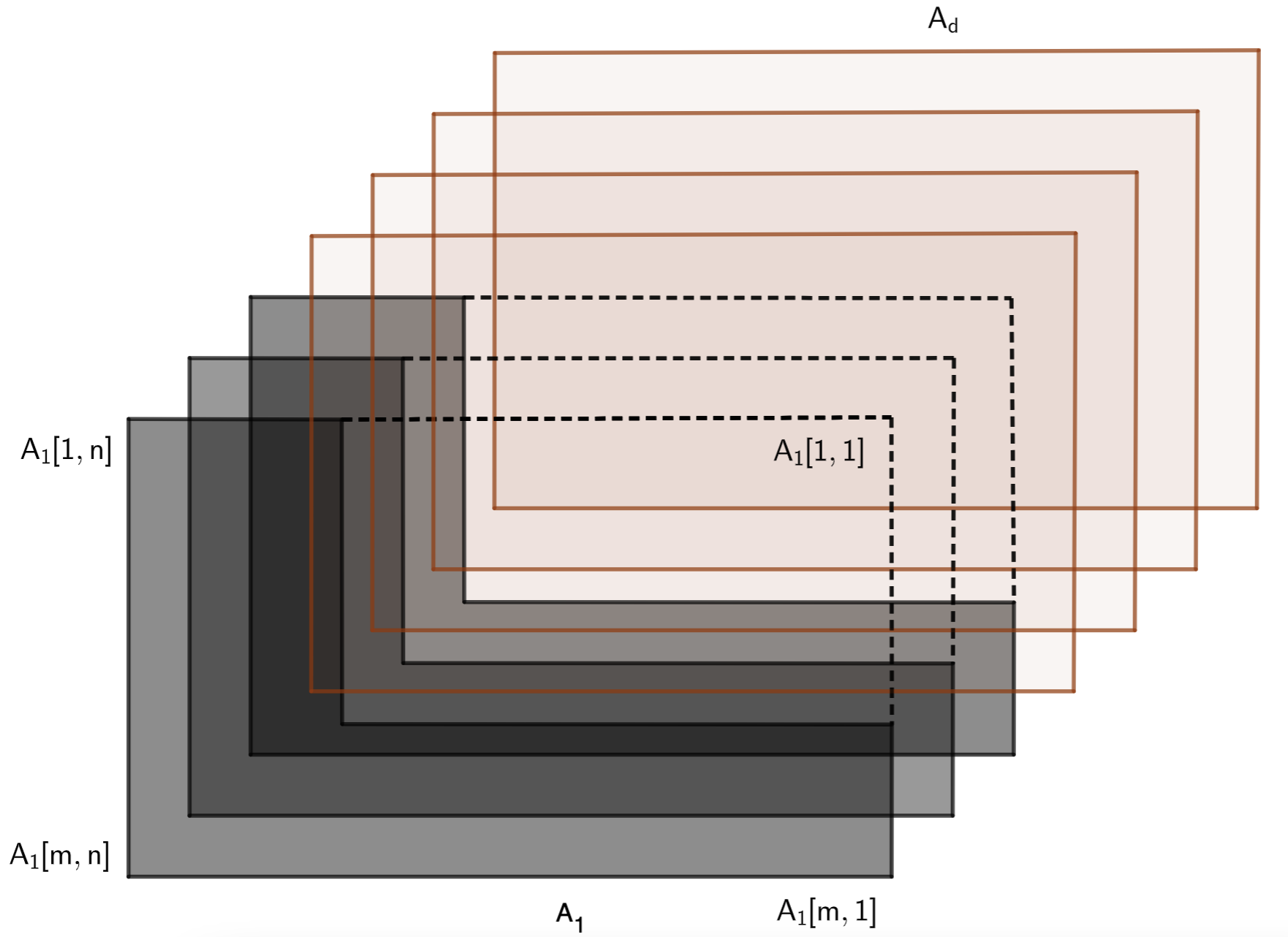}
\end{center}
\vspace{-0.6cm}\caption{The semi-canonical basis of a matrix space after applying matrix space individualization-refinement and the low rank matrix space characterization. The three black matrices in the front form a basis of the space spanned by all the matrices $A \in \matrixspaceA$ such that $LAR$ is a zero matrix. The transparent rectangles enclosed by the dashed black lines are zero matrices. The four light brown matrices in the back are the rest matrices in the basis.
}\label{fig:ind_low_rank_one_dim_intro}
\end{figure}

The semi-canonical basis is not canonical because, for fixed individualization matrices, there can be different semi-canonical bases. 
But the semi-canonical bases can provide a partial correspondence between two 
isometric skew-symmetric matrix spaces. 
Suppose two skew-symmetric matrix spaces $\matrixspaceA$ and $\matrixspaceB$ are isometric  and let $S$ be an isometry from $\matrixspaceA$ to $\matrixspaceB$. 
For individualization matrices $L$ and $R$ of $\matrixspaceA$, 
let $(A_1, \dots, A_d)$ be a semi-canonical basis  of $\matrixspaceA$ with $L$ and $R$ as individualization matrices, 
and $(B_1, \dots, B_d)$  be a semi-canonical basis of $\matrixspaceB$ with $L S^{-1}$ and $(S^T)^{-1}R$  as individualization matrices. 
Then 
for each $1 \leq i \leq d$, $S A_i S^T = B_i + B_i'$ for some $B_i'$ satisfying the condition that $L S^{-1} B_i' (S^T)^{-1}R$ is a zero matrix. 
The partial correspondence also holds for two equivalent matrix spaces. 
Two matrix spaces $\matrixspaceA$ and $\matrixspaceB$, in which matrices are not necessarily square matrices, are equivalent if there are invertible matrices $X$ and $Y$ such that $X \matrixspaceA Y = \matrixspaceB$, i.e., $\matrixspaceB$ equals the space spanned by $XAY$ for all the matrices $A \in \matrixspaceA$.


\paragraph{Tensor representation of skew-symmetric matrix spaces}
Next, we combine the matrix space individualization-refinement and the low rank matrix space characterization to analyze skew-symmetric matrix spaces. 
For convenience, let us define a three-tensor representation for skew-symmetric matrix spaces following~\cite{li2017linear}.
For a skew-symmetric matrix space $\matrixspaceA$ of dimension $m$ such that every matrix in the space is an $n \times n$ matrix, 
a three-tensor $\mathbf{G} \in \mathbb{F}_p^{m \times n \times n}$ is a skew-symmetric matrix space tensor of $\matrixspaceA$ if 
$\mathbf{G}[i, j, k] = A_i[j, k]$ for a linear basis $(A_1, \dots, A_m)$ of $\matrixspaceA$, 
where
$A_i[j, k]$ is the $(j, k)$-th entry of $A_i$, and $\mathbf{G}[i, j, k]$ is the $(i, j, k)$-th entry of $\mathbf{G}$.

Given a skew-symmetric matrix space tensor $\mathbf{G}$, 
we use $\matrixspaceX_{\mathbf{G}, i}$ to denote the $n \times n$ skew-symmetric matrix such that $\matrixspaceX_{\mathbf{G}, i}[j, k] = \mathbf{G}[i, j, k]$, 
use $\matrixspaceY_{\mathbf{G}, j}$ to denote the $m \times n$ matrix such that $\matrixspaceY_{\mathbf{G}, j}[i, k] = \mathbf{G}[i, j, k]$, and 
use $\matrixspaceZ_{\mathbf{G}, k}$ to denote the $m \times n$ matrix such that $\matrixspaceZ_{\mathbf{G}, k}[i, j] =  \mathbf{G}[i, j, k]$.
We also use $\matrixspaceX_\mathbf{G}$ to denote the matrix space $\langle \matrixspaceX_{\mathbf{G}, 1}, \dots \matrixspaceX_{\mathbf{G}, m}\rangle $,
use $\matrixspaceY_\mathbf{G}$ to denote the matrix space $\langle \matrixspaceY_{\mathbf{G}, 1}, \dots \matrixspaceY_{\mathbf{G}, n}\rangle $,
and use $\matrixspaceZ_\mathbf{G}$ to denote the matrix space $\langle \matrixspaceZ_{\mathbf{G}, 1}, \dots \matrixspaceZ_{\mathbf{G}, n}\rangle $, where $\langle \cdot \rangle$ is the linear span.
We remark that $\matrixspaceX_{\mathbf{G}}$ is a skew-symmetric matrix space, but $\matrixspaceY_{\mathbf{G}}$  and $\matrixspaceZ_{\mathbf{G}}$  are not. 

One can verify that 
two skew-symmetric matrix spaces are isometric if and only if their tensors (denoted as $\mathbf{G}$ and $\mathbf{H}$) are isometric,
i.e., 
there is an $n \times n$ invertible matrix $N$ and an $m \times m$ invertible matrix $M$ such that 
the transform of $\mathbf{G}$ by $N$ and $M$, denoted as $\trans_{N, M}(\mathbf{G})$, equals $\mathbf{H}$, where 
\[\matrixspaceX_{\trans_{N, M}(\mathbf{G}), i} = \sum_{i' = 1}^m  M[i, i'] \cdot\left( N  \cdot \matrixspaceX_{\mathbf{G}, i'} \cdot N^T \right).\]

\paragraph{Semi-canonical form of skew-symmetric matrix space tensors}
In this work, 
the purpose of the tensor representation of a skew-symmetric matrix space is to incorporate the matrix space  individualization-refinement and the low rank matrix space characterization techniques so the tensor is transformed into a more structured form, called the ``semi-canonical form'' of the tensor.

For a skew-symmetric matrix space tensor $\mathbf{G}$,
the semi-canonical form of $\mathbf{G}$, denoted as $\semic(\mathbf{G})$, is obtained by applying the two techniques to the three matrix spaces $\matrixspaceX_\mathbf{G}$, $\matrixspaceY_\mathbf{G}$, and $\matrixspaceZ_\mathbf{G}$
so matrices in each of the three matrix spaces have the structure shown in Figure~\ref{fig:ind_low_rank_one_dim_intro}.
To achieve this, 
we need to carefully choose the individualization and formatting matrices in a coordinated fashion. 
In particular, 
if the individualization and formatting matrices are chosen such that, 
for the left formatting matrix $P$ used for $\matrixspaceX_{\mathbf{G}}$, 
$P^T$ can also be used as the right formatting matrix for $\matrixspaceY_{\mathbf{G}}$ and $\matrixspaceZ_{\mathbf{G}}$, 
then 
the tensor semi-canonical form has the structure shown in Figure~\ref{fig:semi_tensor_intro}(a). 
The tensor values in the transparent region are all zero.
The union of the transparent region and the red cube is called the kernel of the tensor semi-canonical form.
The blue region is called the surface of the tensor semi-canonical form. 


\begin{figure}[h!]
\begin{center}
    \begin{subfigure}[b]{0.48\textwidth}
        \includegraphics[width=\textwidth]{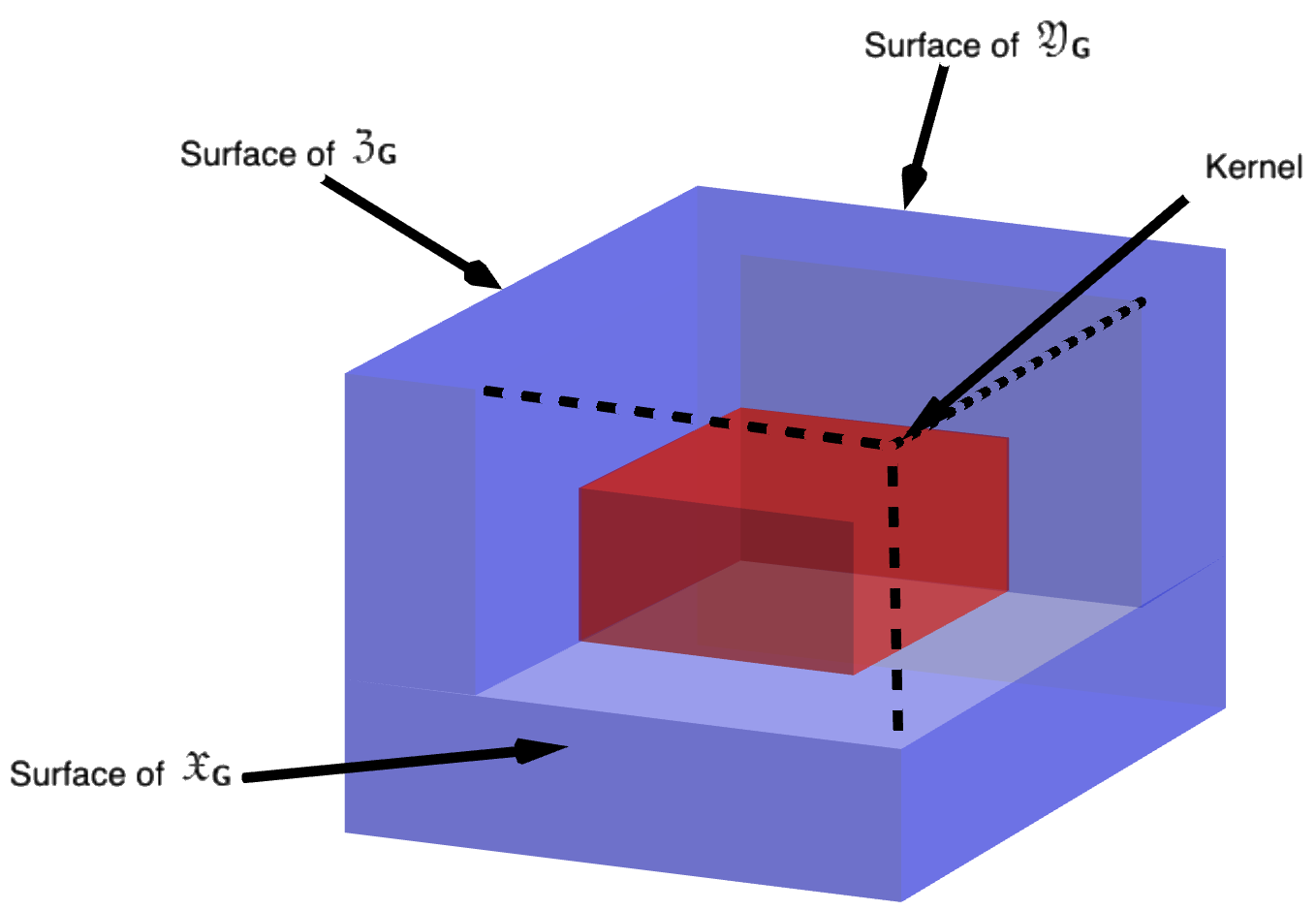}
        \caption{}
    \end{subfigure}
    ~ 
    \begin{subfigure}[b]{0.48\textwidth}
        \includegraphics[width=\textwidth]{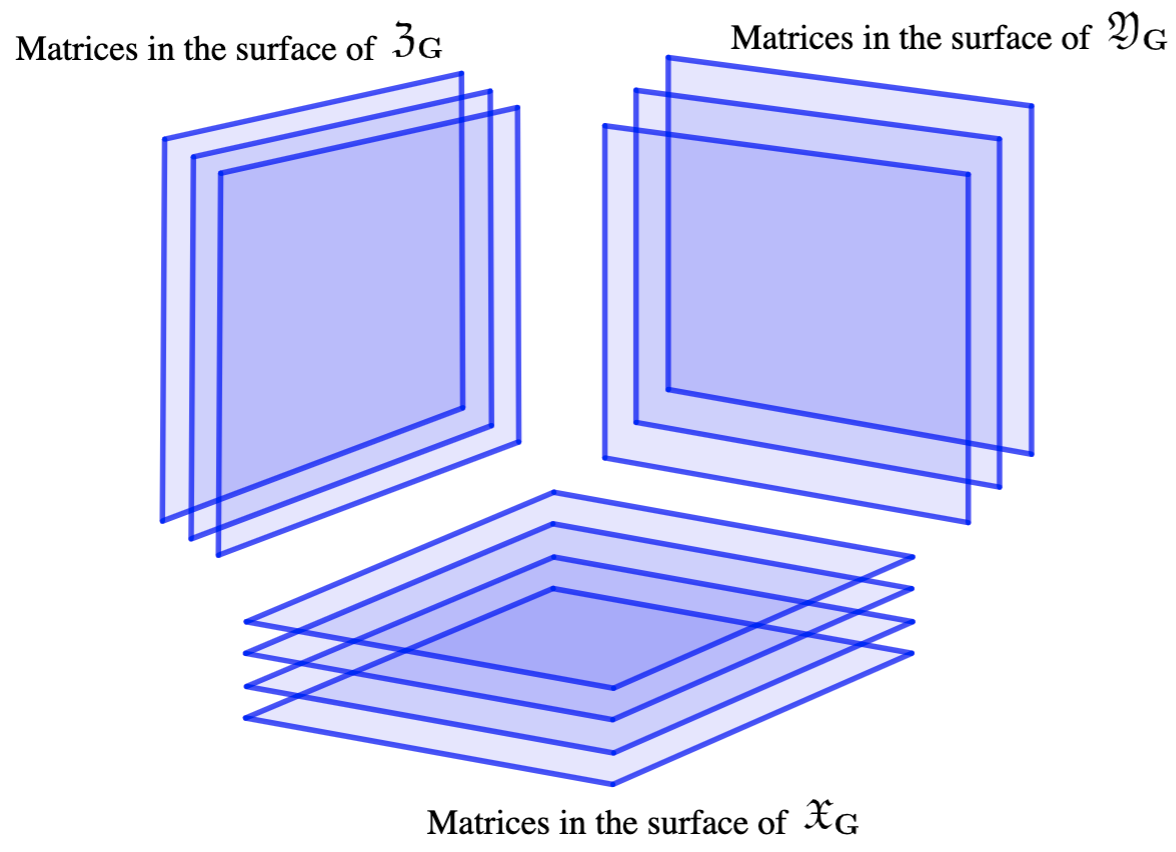}
        \caption{}
    \end{subfigure}

\end{center}
\vspace{-.2cm}\caption{(a). Semi-canonical form of a skew-symmetric matrix space tensor. (b) Matrices in the surfaces of $\matrixspaceX_\mathbf{G}, \matrixspaceY_\mathbf{G}$, and $\matrixspaceZ_\mathbf{G}$.} 
\label{fig:semi_tensor_intro}
\end{figure}




For fixed individualization and formatting matrices, the tensor semi-canonical form is also fixed. 
However, for an efficient tensor isometry testing algorithm, 
it is unacceptable to enumerate all possible formatting matrices, though it is affordable to enumerate all possible individualization matrices.
To address this issue, we show that if the individualization matrices are fixed, and the formatting matrices are partially fixed (i.e., only a few key rows are fixed, and all the other rows satisfy certain conditions), 
then the kernel is fixed (Lemma~\ref{lem:property_semi_canonical_form_tensor_iso}). 
This is also the reason for the term ``semi-canonical form'': the semi-canonical form is not unique for fixed individualization matrices and partially fixed formatting matrices, but the kernel is unique.

In other words, 
if two tensors are isometric, and one constructs the semi-canonical forms of the two tensors using individualization matrices and partially fixed formatting matrices that are the same up to some isometry, 
then the kernels of the two semi-canonical forms are identical. 
Therefore, to determine whether the two tensors are isometric, one only needs to check further if there are formatting matrices that make the surface identical for the two tensors while keeping the kernel unchanged.

In addition, based on the results from the matrix space individualization-refinement and the low rank matrix space characterization,
there are always semi-canonical forms such that 
the numbers of matrices in the surfaces of $\matrixspaceX_\mathbf{G}$, $\matrixspaceY_\mathbf{G}$, and $\matrixspaceZ_\mathbf{G}$ (Figure~\ref{fig:semi_tensor_intro}(b)) are small. 
Hence, in the partially fixed formatting matrices, 
we also fix the rows related to the surfaces of $\matrixspaceX_\mathbf{G}$, $\matrixspaceY_\mathbf{G}$, and $\matrixspaceZ_\mathbf{G}$. 
Then 
matrices in the surfaces from the three matrix spaces are fixed up to some formatting matrices satisfying the partially fixed constraint.   
%
%

Hence, the isometry testing of skew-symmetric matrix space tensors 
reduces to the isometry testing of their semi-canonical forms by enumerating individualization matrices and partially fixed formatting matrices for both tensors. 
Due to the fixed kernel for all the semi-canonical forms, 
the isometry testing of semi-canonical forms further reduces to deciding whether the surfaces are identical between semi-canonical forms 
up to some formatting matrices satisfying the partially fixed constraint.




\paragraph{Reduction to skew-symmetric matrix tuple isometry testing}
Finally, we reduce the isometry testing of semi-canonical forms of skew-symmetric matrix spaces to the aforementioned skew-symmetric matrix tuple isometry problem. 
The high-level idea is to construct a skew-symmetric matrix tuple to encode the surface of the tensor semi-canonical form. 
Because the matrices in the surfaces of $\matrixspaceX_\mathbf{G}$, $\matrixspaceY_\mathbf{G}$, and $\matrixspaceZ_\mathbf{G}$ are fixed, we can use different matrices in the matrix tuple to encode the matrices in the surface.

Suppose the kernel is of dimension $m' \times n' \times n'$ for some $1 \leq m' \leq m$ and $1 \leq n' \leq n$.
In our skew-symmetric matrix tuple of $\semic(\mathbf{G})$, denoted as $\FF_{\semic(\mathbf{G})}$, each matrix is of dimension  $(3 + n + m') \times (3 + n + m')$. 
The rows from the fourth to the $(3+n)$-th of matrices in $\FF_{\semic(\mathbf{G})}$ correspond to the rows of matrices in $\matrixspaceX_{\mathbf{G}}$.
The last $m'$ rows of matrices in $\FF_{\semic(\mathbf{G})}$ correspond to the first $m'$ rows of matrices in $\matrixspaceY_\mathbf{G}$ (or equivalently $\matrixspaceZ_\mathbf{G}$). 
The first three rows of matrices in $\FF_{\semic(\mathbf{G})}$ are auxiliary rows used to ensure that the other rows satisfy the constraints of the partially fixed formatting matrices. 
See Figure \ref{fig:fig_semi_canonical_iso_intro} for an illustration.

\begin{figure}[h]
\begin{center}
\includegraphics[width=0.6\textwidth]{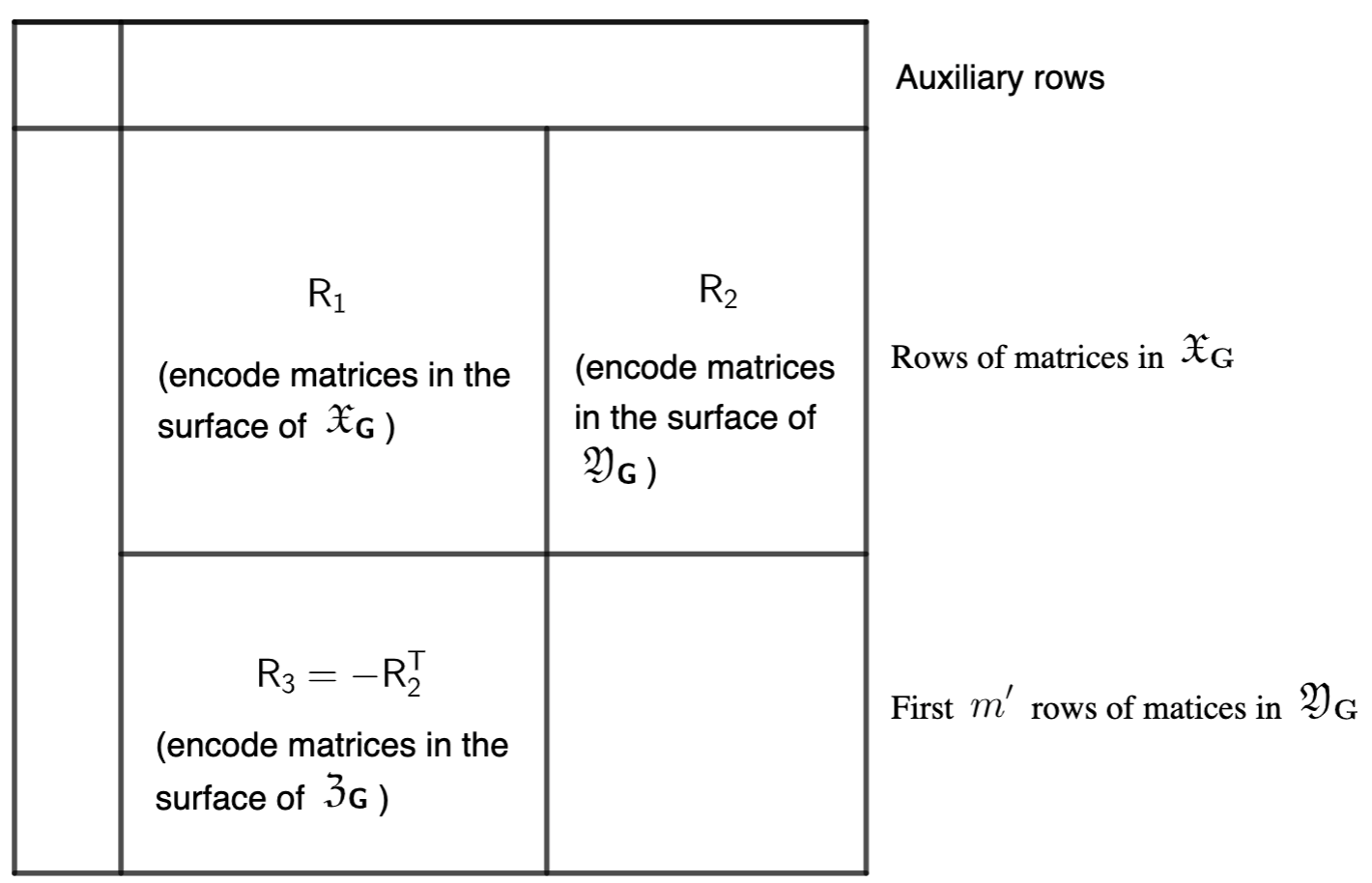}
\end{center}
\vspace{-.4cm}\caption{The matrices in $\FF_{\semic(\mathbf{G})}$.}\label{fig:fig_semi_canonical_iso_intro}
\end{figure}

We use the submatrices on $R_1$ (as Figure \ref{fig:fig_semi_canonical_iso_intro}) for all the matrices in $\FF_{\semic(\mathbf{G})}$ to 
encode the skew-symmetric matrices in the surface of $\matrixspaceX_{\mathbf{G}}$.
We also use the submatrices on $R_2$ for all the matrices in $\FF_{\semic(\mathbf{G})}$ to encode the matrices in the surface of $\matrixspaceY_{\mathbf{G}}$ 
(excluding the intersection with the surface of $\matrixspaceX_{\mathbf{G}}$). 
Consequently, 
the submatrices on $R_3$ for all the matrices in $\FF_{\semic(\mathbf{G})}$, which is the negative transpose of submatrices on $R_2$ by the skew-symmetric condition,  encode the matrices in the surface of $\matrixspaceZ_{\mathbf{G}}$ (excluding the intersection with the surface of $\matrixspaceX_{\mathbf{G}}$). 
We use the other submatrices to ensure constraints given by the partially fixed formatting matrices. 



By carefully designing matrix tuples constructed from tensor semi-canonical forms,
we show that the semi-canonical forms of two skew-symmetric matrix space tensors are isometric if and only if there is an isometry $S$ between the skew-symmetric matrix tuples 
such that $S$ is a block diagonal matrix
\[S = \left(\begin{array}{cc} Q & 0 \\ 0   & W\end{array}\right)\]
for some 
$(3 + n) \times (3+n)$ matrix $Q$ and $m' \times m'$ matrix $W$ (Lemma~\ref{lem:restricted_tuple_isometry}). 

Naturally, we want to determine the isometry of the two tensors by running the skew-symmetric matrix tuple isometry algorithm (Theorem~\ref{thm:skew_tuple_isometry_intro}) on the matrix tuples constructed from the semi-canonical forms. 
However, the requirement of $S$ being block diagonal makes things more complex. 

Suppose we run the algorithm for skew-symmetric matrix tuple isometry on the matrix tuples constructed. 
If the algorithm returns no, then the two semi-canonical forms are not isometric. 
If the algorithm returns yes and an isometry that is block diagonal, then the two semi-canonical forms are isometric. 
The difficult case is when 
the algorithm returns yes and an isometry that is not block diagonal. 
For this case, we can neither certify that the two semi-canonical forms are isometric,
nor show 
that the two semi-canonical forms are not isometric.

Let us consider an easier scenario: Suppose
for each non-zero row vector $v \in \mathbb{F}_p^n$, there is a matrix $X$ in the surface of $\matrixspaceX_\mathbf{G}$ such that $v X$ is a non-zero vector. 
With this condition, together with our construction of matrix tuples, 
we can show that the isometry returned is of the form 
\[\left( \begin{array}{cc} X & Y \\ 0 & Z \end{array}\right).\]
After carefully analyzing the matrix tuples constructed, we show that 
\[\left( \begin{array}{cc} X & 0 \\ 0 & Z \end{array}\right)\]
is also an isometry, and thus the two semi-canonical forms are isometric. 

The general case is more complex because the left bottom submatrix of the isometry returned can be non-zero. 
However, we show that either we can certify that there exists another block diagonal isometry for the skew-symmetric matrix tuples,
or we can reduce the problem to a matrix tuple equivalence problem, i.e., the problem of determining whether two matrix tuples $\AA$ and $\BB$ have invertible matrices $P$ and $Q$ such that 
$P \AA Q = \BB$. 
According to~\cite{ivanyos2019algorithms}, the matrix tuple equivalence problem can be solved efficiently.

\paragraph{Paper organization} 
In Section~\ref{sec:prelim}, we define the notations and provide the preliminaries. 
In Section~\ref{sec:individualization}, we present our results on matrix space individualization-refinement.
In Section~\ref{sec:low_rank}, we present our results on low rank matrix space characterization.
Section~\ref{sec:semi_canonical_form_tensor} defines skew-symmetric matrix space tensor and its semi-canonical form. %
In Section~\ref{sec:reduction}, we present the reduction to the skew-symmetric matrix tuple isometry problem.   
Section~\ref{sec:final} proves Theorem~\ref{thm:main_group} and Theorem~\ref{thm:main_matrix_space}.


\section{Notations and preliminaries}\label{sec:prelim}

Throughout the paper,  the vectors and  matrices are over $\mathbb{F}_p$ for a prime number $p > 2$.
We use $\langle \cdot \rangle$ to denote the linear span.
The base of the logarithm is two unless specified.
Let $\mathbb{F}_p^n$ be the linear space of row vectors of length $n$ over $\mathbb{F}_p$. 
Unless specified,
the vectors are row vectors. 
For a vector $v \in \mathbb{F}_p^n$, we use $v[i]$ to denote the $i$-th entry of $v$ for any $1 \leq i \leq n$.

\paragraph{Matrices} 
Let $M(n, \mathbb{F}_p)$ (and respectively $M(m, n, \mathbb{F}_p)$) be the linear space of $n \times n$ (and respectively $m \times n$) matrices over $\mathbb{F}_p$.
 Let $\mathrm{GL}(n, \mathbb{F}_p)$ be the group of $n \times n$ invertible matrices over $\mathbb{F}_p$.

For a matrix $A\in M(m, n, \mathbb{F}_p)$, let $\rank{A}$ be the rank of $A$, and $A^T$ be the transpose of $A$.
A square matrix $A \in M(n, \mathbb{F}_p)$ is a skew-symmetric matrix if and only if $A = -A^T$. 
For any $1 \leq i \leq m, 1 \leq j \leq n$, let
 $A[i, j]$ be  the entry of $A$ in the $i$-th row and $j$-th column.
For $1 \leq i \leq i' \leq m, 1 \leq j \leq j' \leq n$, 
let $A[i, i'; j, j']$  be the submatrix of $A$ on the rows between $i$ and $i'$ and the columns between $j$ and $j'$. 

For two matrices $A, B \in M(m \times n, \mathbb{F}_p)$, 
$A$ is lexically smaller than $B$, denoted as $A \prec B$, if there exist $1 \leq q \leq m$ and $1 \leq r \leq n$
such that the following conditions hold:
\begin{itemize}
\item $A[i, j] = B[i, j]$ for any $1 \leq i \leq q -1, 1 \leq j \leq n$ or any $i = q, 1 \leq j  < r$;
\item $A[q, r] < B[q, r]$.
\end{itemize}
We denote $A \preceq B$ if $A \prec B$ or $A = B$.


We use $I_{n}$ to denote the $n \times n$ identity matrix. 




\paragraph{Matrix tuples and matrix spaces}
 An $m \times n$ matrix tuple $\AA$ of length $k$, denoted as $\AA = (A_1, \dots, A_k)$, is an element in $M(m, n, \mathbb{F}_p)^k$.
 For 
any $P \in M(\alpha, m, \mathbb{F}_p)$ and $Q\in M(n, \beta, \mathbb{F}_p)$  with some positive integers $\alpha$ and $\beta$, let $P\AA Q$ be the matrix tuple $(P A_1 Q, P A_2, Q, \dots, P A_k Q)$. 



 An $m \times n$ matrix space $\matrixspaceA$ is a linear subspace of $M(m, n, \mathbb{F}_p)$.
 For any $P \in M(\alpha, m, \mathbb{F}_p)$ and $Q\in M(n, \beta, \mathbb{F}_p)$ with some positive integers $\alpha$ and $\beta$, let $P\matrixspaceA Q$ be the linear space spanned by $P A Q$ for all the $A \in \matrixspaceA$. 
 For any row vector $ v \in \mathbb{F}_p^m $, we use $\langle v \matrixspaceA\rangle$ to denote the row vector space spanned by $v A$ for all the $A \in \matrixspaceA$. 
  For two matrix spaces $\matrixspaceA$ and $\matrixspaceB$,
we denote $\matrixspaceA \leq \matrixspaceB$ if $\matrixspaceA$ is a subspace of $\matrixspaceB$.

Since any linear combination of skew-symmetric matrices of the same dimension 
is also a skew-symmetric matrix, 
we use $\mathrm{SS}(n, \mathbb{F}_p)$ to denote the linear space of all the $n \times n$ skew-symmetric matrices.

\paragraph{Isometry and equivalence for matrix tuples and spaces} 

We define equivalence relations for matrix tuples.

\begin{definition}[Matrix tuple equivalence]
Let $\AA = (A_1, \dots, A_k), \BB = (B_1, \dots, B_k)$ be two matrix tuples in $M(m, n, \mathbb{F}_p)^k$.
$\AA$ and $\BB$ are equivalent if there exist two matrices $P \in \GL(m, \mathbb{F}_p)$ and $Q \in \GL(n, \mathbb{F}_p)$ such that 
$P \AA Q = \BB$. 
\end{definition}

\begin{definition}[Skew-symmetric matrix tuple isometry]
Let $\AA = (A_1, \dots, A_k)$ and  $\BB = (B_1, \dots, B_k)$ be two skew-symmetric matrix tuples in $\SS(n, \mathbb{F}_p)^k$.
$\AA$ and $\BB$ are isometric if there exists a matrix $P \in \GL(n, \mathbb{F}_p)$ such that 
$P \AA P^T = \BB$. 
$P$ is called an isometry from $\AA$ to $\BB$ if $P$ exists.
\end{definition}

In this paper, we use the algorithm for the isometry testing of two skew-symmetric matrix tuples (Theorem~\ref{thm:skew_tuple_isometry_intro})
and the algorithm for the equivalence testing of two matrix tuples (Theorem~\ref{thm:tuple_isometry}), both proposed by Ivanyos and Qiao in~\cite{ivanyos2019algorithms}.

\begin{theorem}[Proposition 3.2 of \cite{ivanyos2019algorithms}]\label{thm:tuple_isometry}
Given two matrix tuples $\AA = (A_1, \dots, A_k)$ and $\BB = (B_1, \dots, B_k)$ in $M(m, n, \mathbb{F}_p)^k$ for some prime $p > 2$ and positive integers $k, m$ and $n$, 
there is an algorithm with running time $\mathrm{poly}(k, n, m, p)$ to determine whether $\AA$ and $\BB$ are equivalent.
\end{theorem}

Following the definitions for matrix tuples, we also define the equivalence of matrix spaces and 
the isometry of  skew-symmetric matrix spaces. 

\begin{definition}[Matrix space equivalence]
Let $\matrixspaceA,\matrixspaceB \leq M(m, n, \mathbb{F}_p)$ be two matrix spaces for some positive integers $m$ and $n$.
$\matrixspaceA$ and $\matrixspaceB$ are equivalent if there exist two matrices $P \in \GL(m, \mathbb{F}_p)$  and $Q \in \GL(n, \mathbb{F}_p)$ such that 
$P \matrixspaceA Q = \matrixspaceB$. 

\end{definition}

\begin{definition}[Skew-symmetric matrix space isometry]\label{def:skew_space_isometry}
Let $\matrixspaceA,\matrixspaceB \leq \SS(n, \mathbb{F}_p)$ be two skew-symmetric matrix spaces.
$\matrixspaceA$ and $\matrixspaceB$ are isometric if there exists a matrix $P \in \GL(n, \mathbb{F}_p)$ such that 
$P \matrixspaceA P^T = \matrixspaceB$. 
$P$ is called an isometry from $\matrixspaceA$ to $\matrixspaceB$ if $P$ exists.

\end{definition}

\paragraph{Baer's correspondence} 

For a $p$-group of nilpotent class 2 and exponent $p$, let $p^k$ denote the order of the group.  
Because of the class two and exponent $p$ condition, $G / Z(G)$ is isomorphic to $\mathbb{Z}_p^n$, and $[G, G]$ is isomorphic to $\mathbb{Z}_p^m$ for some positive integers $n$ and $m$ such that $m + n \leq k$, 
where $Z(G)$ denotes the center of $G$ and 
$[G, G]$ denotes the 
group generated by 
$xyx^{-1}y^{-1}$ for all $x, y\in G$.
Taking an arbitrary basis
of $G/Z(G)$, an arbitrary basis of $[G, G]$, and taking the commutator bracket, we obtain a skew-symmetric bilinear
map $b_G$ : $\mathbb{F}_p^n \times \mathbb{F}_p^n \rightarrow \mathbb{F}_p^m$, which can be represented by a skew-symmetric matrix tuple $\GG = (G_1, \dots, G_m)$ such that every $G_i$ is a matrix in $\mathrm{SS}(n, \mathbb{F}_p)$.
Such a skew-symmetric matrix tuple is called a skew-symmetric matrix tuple of $G$.

For two $p$-groups $G$ and $H$ of nilpotent class 2 and exponent $p$,
it is necessary for $H$ to be isomorphic to $G$ that 
$\mathrm{dim}_{\mathbb{Z}_p}(G / Z(G)) = \mathrm{dim}_{\mathbb{Z}_p}(H / Z(H))$
and $\mathrm{dim}_{\mathbb{Z}_p}([G, G]) = \mathrm{dim}_{\mathbb{Z}_p}([H, H])$.
The following theorem, also called Baer's correspondence, was proved by Baer in~\cite{baer1938groups}. 


\begin{theorem}[Baer's correspondence~\cite{baer1938groups}, rephrased]\label{thm:baer_correspondance}
Let $G$ and $H$ be two $p$-groups of class two and exponent $p$  for some prime number $p$ with the same order. 
Let $\GG$ and $\HH$ be the skew-symmetric matrix tuples of $G$ and $H$, respectively. 
If both $\GG$ and $\HH$ are $n\times n$ skew-symmetric matrix tuples of length $m$,
then $G$ and $H$ are isomorphic if and only if 
there are matrices $P \in \GL(n, \mathbb{F}_p)$ and $Q \in \GL(m, \mathbb{F}_p)$
such that $G_i = \sum_{j = 1}^m  Q[i, j](P\cdot   H_j \cdot P^T)$.
\end{theorem}

Furthermore, we can also represent skew-symmetric matrix tuples of groups by skew-symmetric matrix spaces. 
Given an arbitrary skew-symmetric matrix tuple $\GG$ of group $G$, 
the skew-symmetric matrix space $\matrixspaceG$ of $G$ is the linear matrix space spanned by matrices in $\GG$.
Hence,  Baer's correspondence can be rephrased as follows. 

\begin{corollary}
Let $G$ and $H$ be two $p$-groups of class two and exponent $p$  for some prime number $p$ with the same order. 
Let $\matrixspaceG$ and $\matrixspaceH$ be the skew-symmetric matrix spaces of  $G$ and $H$, respectively. 
$G$ and $H$ are isomorphic if and only if $\matrixspaceG$ and $\matrixspaceH$ are isometric.

\end{corollary}


In this paper, we will use the following fact from Baer's correspondence. 

\begin{fact}\label{fact:baer_correspondance_tensor}
Let $\matrixspaceG \leq \SS(n, \mathbb{F}_p)$  be a skew-symmetric matrix space of a $p$-group of class two and exponent $p$ for some prime number $p$. 
Then for any non-zero row vector $v \in \mathbb{F}_p^n$, 
there is a matrix $A \in \matrixspaceG$ such that $v A$ is a non-zero vector. 
\end{fact}

\section{Matrix space individualization-refinement}\label{sec:individualization}
In this paper, the individualization-refinement to a matrix space $\matrixspaceA \leq M(m, n, \mathbb{F}_p)$
is defined by a left individualization matrix $L \in M(\alpha, m, \mathbb{F}_p)$
and a right individualization matrix $R \in M(n, \beta, \mathbb{F}_p)$ for some positive integers $\alpha$ and $\beta$.
We aim to distinguish matrices in $\matrixspaceA$ by comparing $LAR$ 
for matrices $A\in \matrixspaceA$.

Ideally, if for any two different matrices $A, A' \in \matrixspaceA$, $LAR \neq LA'R$, then each matrix $A \in \matrixspaceA$ has its unique $LAR$. 
And thus, all the matrices in $\matrixspaceA$ are distinguished. 
But if there is a matrix $A \in \matrixspaceA$ such that $LAR$ is a zero matrix, then for each $A'  \in \matrixspaceA$, $LA' R = L (A' + A) R$. 
Let $\zero_{L, R}(\matrixspaceA)$ be the space spanned by matrices $A \in \matrixspaceA$ such that $LAR$ is a zero matrix. 

We show that 
in order to distinguish matrices in a matrix space,
the dimensions of the left and right individualization matrices are related to the rank of the matrices in the matrix space (Lemma~\ref{lem:individualzation_main}).



\begin{lemma}\label{lem:individualization_rank}
Let $A$ be a matrix in $M(m, n, \mathbb{F}_p)$ of rank at least $k$ for a prime $p$ and some positive integers $m, n, k$. 
Given an integer $1 \leq k' \leq k / 2$ and a parameter $0 < \delta < 1$ satisfying $\log(1 / \delta) \geq k$,
let $Q$ be a matrix in \[M\left(\left\lceil \frac{8k' \log(1 / \delta) }{ k}\right\rceil, m, \mathbb{F}_p\right)\] such that the entries of $Q$ are independently and uniformly sampled from $\{0, 1, \dots, p-1\}$. 
With probability $1 - \delta$, $QA$ is a matrix of rank at least $k'$. 
\end{lemma}
\begin{proof}
Since the entries of $Q$ are independently and uniformly sampled, without loss of generality, we assume the matrix $Q$ is sampled sequentially by rows $q_1, \dots, q_\alpha$, where $q_i$ is the $i$-th row of $Q$.

Suppose that after sampling the first $\alpha$ rows for some positive integer $\alpha$, $\langle q_1 A, \dots, q_\alpha A \rangle$ is a $\beta$ dimensional space for some $\beta < k'$. 
The probability that $q_{\alpha+1} A$ is a vector in $\langle q_1 A, \dots, q_\alpha A \rangle$ is $(1 / p)^{k - \beta}$. 
Thus, by sampling $\lceil \log(k' / \delta) / (k - k')\rceil $ new rows of $Q$, the rank of $Q$ increases by at least one with probability 
\[1 - \left(\frac{1}{p}\right)^{(k - \beta)\cdot \lceil \log(k' / \delta) / (k - k')\rceil} \geq 1 - \left(\frac{1}{p} \right )^{ \log(k' / \delta) } \geq 1 - \left(\frac{1}{2}\right)^{\log(k' / \delta)} = 1 - \frac{\delta }{k'}.\]

On the other hand, 
observe that 
\[k' \left\lceil \frac{\log(k' / \delta) }{ ( k - k')}\right\rceil \leq \frac{2 k'  \log(k' / \delta)}{k - k'} \leq  \frac{4 k'  \log(1/ \delta)}{k - k'}
\leq \frac{8 k' \log(1/\delta) }{  k},\]
where the first inequality uses the condition that $\log(k' / \delta) \geq \log(1/\delta)\geq k \geq k - k'$, the second inequality uses the condition that $1/\delta \geq \log(1 / \delta) \geq  k > k'$,
and the third inequality uses the condition that $k - k'\geq k / 2$.

By union bound, with probability at least $1 - \delta$, if $Q$ contains at least $\lceil 8 k' \log(1/\delta) /  k\rceil$ rows, 
then 
the rank of $Q$ is at least $k'$. 
%
\end{proof}

\begin{lemma}\label{lem:individualzation_main}
Let $\matrixspaceA$ be a $d$-dimensional matrix subspace of $M(m, n, \mathbb{F}_p)$ for a prime $p$ and some positive integers $d, m, n$. 
For any $k \geq 4$, 
denote 
\[t \coloneqq \left\lceil 32\max\{d\log(p), k \} / \sqrt{k}\right\rceil.\]
There is a left individualization matrix $L \in M(t, m, \mathbb{F}_p)$
and a  right individualization matrix $R \in M(n, t, \mathbb{F}_p)$ such that for any $A \in \matrixspaceA$ of rank at least $k$, 
$LAR$ is a non-zero matrix. 
\end{lemma}
\begin{proof}
Let $A$ be an arbitrary matrix in $\matrixspaceA$ of rank at least $k$.
By Lemma~\ref{lem:individualization_rank} with 
$ k' = \lfloor \sqrt{k}\rfloor$ and $\delta = \min\{1 / (4p^d), 1 / 2^k\}$, if every entry of $L$ is independently and uniformly sampled from $\{0, \dots, p-1\}$,
then with probability at least $1 - 1 / (4p^{d})$, $LA$ is of rank at least $\lfloor \sqrt{k}\rfloor$.
If this case happens, then by Lemma~\ref{lem:individualization_rank} with  $k' = 1$ and $\delta = \min\{1 / (4p^d), 1 / 2^k\}$, if every entry of $R$ is independently and uniformly sampled from $\mathbb{F}_p$,
then with probability at least $1 - 1 / (4p^{d})$, $LA R$ is of rank at least $1$.
By union bound, for random $Q$ and $R$, with probability at least $1 - 1 / (2p^d)$, $LAR$ is a non-zero matrix. 

By union bound, with constant probability, for random $Q$ and $R$, $QAR$ is a non-zero matrix for all the $A\in \matrixspaceA$ of rank at least $k$. 
\end{proof}

We define the semi-canonical basis for a matrix space with respect to left and right individualization matrices.

\begin{definition}\label{def:semi_canonical_base}
Let $\matrixspaceA$ be a matrix space of dimension $d$ for some positive integer $d$.
Let $L$ be a left individualization matrix for $\matrixspaceA$ and $R$ be a right individualization matrix for $\matrixspaceA$. 
A matrix tuple $\AA = (A_1, \dots, A_d)$ is a semi-canonical basis of $\matrixspaceA$ with respect to $L$ and $R$ if the following two conditions hold:
\begin{enumerate}
\item $\langle A_1, \dots, A_d\rangle = \matrixspaceA$.
\item For each $1 \leq i \leq d$, $L A_i R \preceq L A R$ for all the $A$ in $\matrixspaceA$ but not in $ \langle A_1, \dots, A_{i - 1}\rangle $.
\end{enumerate}
\end{definition}

We prove some basic properties for a semi-canonical basis of a matrix space.
\begin{lemma}\label{lem:property_zero_basis}
Let $\matrixspaceA$ be a matrix space of dimension $d$ for some positive integer $d$.
Let $L$ be a left individualization matrix for $\matrixspaceA$ and $R$ be a right individualization matrix for $\matrixspaceA$. 
If $\zero_{L, R}(\matrixspaceA)$ only contains the zero matrix, then there is a unique semi-canonical basis of $\matrixspaceA$ with respect to $L$ and $R$.

If $\dim{\zero_{L, R}(\matrixspaceA)} > 0$, then 
for any semi-canonical basis $(A_1, \dots, A_d)$ of $\matrixspaceA$, 
$L A_i R$ is a zero matrix for all the $1 \leq i \leq \dim{\zero_{L, R}(\matrixspaceA)}$,
and $L A_i R$ is a non-zero matrix for all the $\dim{\zero_{L, R}(\matrixspaceA)} + 1 \leq i \leq d$. 
Furthermore, 
let $(A_1, \dots, A_d)$ and $(A_1', \dots, A_d')$ be
two semi-canonical bases of $\matrixspaceA$ with respect to $L$ and $R$.
For any $1 \leq i \leq d$,
$A_i = A_i' + A_i''$ for some $A_i'' \in \zero_{L, R}(\matrixspaceA)$. 
\end{lemma}
\begin{proof}
If $\zero_{L, R}(\matrixspaceA)$ only contains the zero matrix, then for every non-zero matrix $A \in \matrixspaceA$, $L A R$ is a non-zero matrix. 
Let $\matrixspaceA'$ be a subspace of $\matrixspaceA$ such that $LA'R \prec LAR$ for any $A' \in \matrixspaceA'$ and $A \notin \matrixspaceA'$.
There is a unique $A \in \matrixspaceA$ such that $A \notin \matrixspaceA'$,
and $LAR \preceq LA''R$ for all the $A'' \notin \matrixspaceA'$. Hence, the semi-canonical base is unique. 

If $\dim{\zero_{L, R}(\matrixspaceA)} > 0$, 
since the zero matrix is lexically smallest among all the matrices, 
for any semi-canonical basis $(A_1, \dots, A_d)$ of $\matrixspaceA$, 
$LA_i R$ is a zero matrix for all the $1 \leq i \leq \dim{\zero_{L, R}(\matrixspaceA)}$.
By Definition~\ref{def:semi_canonical_base}, the lemma holds.
\end{proof}

The following lemma shows that for a subspace of $M(m, n, \mathbb{F}_p)$ with dimension $d$, given left and right individualization matrices,  
a semi-canonical basis can be constructed in $p^d \cdot \mathrm{poly}(n, m, p, d)$ time.

\begin{lemma}\label{lem:algo_semi_canonial_base}
Given an arbitrary basis of 
 a matrix space $\matrixspaceA \leq M(m, n, \mathbb{F}_p)$ of dimension $d$,  a left individualization matrix $L$, and a right individualization matrix $R$ for $\matrixspaceA$,
 there is an algorithm to compute a semi-canonical basis of $\matrixspaceA$ with respect to $L$ and $R$ in time $p^d \cdot \mathrm{poly}(n, m, p, d)$
 if both $L$ and $R$ contain at most $\mathrm{poly}(n, m)$ rows and columns.
\end{lemma}

\begin{proof}
Consider the following algorithm: 
\begin{enumerate}
\item For $i = 1 $ to $d$, find a non-zero matrix $A_i$ in $\matrixspaceA$ but not in $\langle A_1, \dots, A_{i - 1}\rangle $ such that  $L A_i R \preceq L A R$ for all the matrices  $A$ in $ \matrixspaceA$ but not in $\langle A_1, \dots, A_{i - 1}\rangle $. 
\item Return $(A_1, \dots, A_d)$.
\end{enumerate}

The correctness of the algorithm is by Definition~\ref{def:semi_canonical_base}. 
Since the algorithm has $d$ iterations, and in each iteration, the algorithm enumerates all the matrices in $\matrixspaceA$, the running time of the algorithm is $p^d \cdot \mathrm{poly}(n, m, d, p)$.
\end{proof}

\section{Low rank matrix space characterization}\label{sec:low_rank}
If all the matrices in a skew-symmetric matrix space, 
except the zero matrix, are of high rank, 
then with proper left and right individualization matrices, 
the semi-canonical basis is unique (Lemma~\ref{lem:property_zero_basis}). 
Furthermore, if two skew-symmetric  matrix spaces are isometric and the two matrix spaces only contain  high rank matrices (excluding the zero matrix in each space), then one can fix left and right individualization matrices for one space such that the semi-canonical basis is unique and enumerate all the possible images of the left and right individualization matrices for the other space. 
With left and right individualization matrices for both spaces, we compute the unique semi-canonical bases for both matrices spaces. 
Then the skew-symmetric matrix space isometry problem reduces to the skew-symmetric matrix tuple isometry problem, which can be solved efficiently (Theorem~\ref{thm:skew_tuple_isometry_intro}). 

So, the hard case for the skew-symmetric matrix space isometry problem is that the matrix space contains some matrices of low rank. 
After applying the left and right individualization matrices, 
the resulting zero matrices correspond to a subspace of the skew-symmetric matrix space
such that all the matrices in the subspace are of low rank. 

In this section, we investigate the structure of low rank matrix spaces, i.e., matrix spaces in which every matrix is of low rank, to characterize some useful properties. 
In Section~\ref{sec:semi_canonical_form_tensor}, we will use these properties to construct semi-canonical forms for tensors obtained from skew-symmetric matrix spaces so that the skew-symmetric matrix space isometry problem reduces to the skew-symmetric matrix tuple isometry problem.

In particular, we show that for  a low rank matrix space  $\matrixspaceA \leq M(m, n, \mathbb{F}_p)$ such that every matrix in the space is of rank at most $r$,
there are matrices $P \in \mathrm{GL}(m, \mathbb{F}_p)$ and $Q \in \mathrm{GL}(n, \mathbb{F}_p)$ such that 
for each $A \in \matrixspaceA$, $PAQ$ has non-zero entries only in the last $O(r^2)$ rows or in the last $O(r^2)$ columns. 
Furthermore, if $\matrixspaceA$ is a skew-symmetric matrix space, 
then $Q = P^T$. 
Similar characterizations were studied in~\cite{flanders1962spaces, atkinson1981primitive}. But to the author's knowledge, all the previous results require that  the underlying field has at least $r + 1$ elements. \\ 

We first define the attribute set for a matrix space, as well as the kernel, complementary matrix, and formatting matrix for a matrix space and an attribute set. 
Since we require the property of $Q = P^T$ for the skew-symmetric matrix space, 
we give the definition for skew-symmetric matrix space and the definition for general (non-square) matrix space separately. 

Take the general matrix space as an example. 
Suppose $\matrixspaceA\leq M(m, n, \mathbb{F}_p)$ is matrix space,
and there are matrices $P \in \mathrm{GL}(m, \mathbb{F}_p)$ and $Q \in \mathrm{GL}(n, \mathbb{F}_p)$ such that 
for each $A \in \matrixspaceA$, $PAQ$ has non-zero entries only in the last $\alpha$ rows or in the last $\alpha$ columns for some integer $\alpha$. 
Roughly speaking, $P$ and $Q$ are left and right formatting matrices of $\matrixspaceA$.
The row vector space $S$ spanned by the first $m - \alpha$ rows of $PA$ for all the $A \in \matrixspaceA$ is a space of dimension $\alpha$.
To define $P$ and $Q$, 
the attribute set corresponds to a linear basis of the row vector space $S$, 
and the complementary matrix corresponds to the submatrix of $P$ on the last $\alpha$ rows.

\begin{definition}[kernel, complementary matrix, and formatting matrix for skew-symmetric matrix spaces]\label{def:skew_space_complement}
Let $\matrixspaceA$ be an $n \times n$ skew-symmetric matrix space over $\mathbb{F}_p$ for some prime $p > 2$ and positive integer $n$. 
An attribute set $\Lambda$ for $\matrixspaceA$ is a set of linearly independent row vectors in $\mathbb{F}_p^n$.
The kernel 
for $\matrixspaceA$ and $\Lambda$, denoted as $\ker_{\mathrm{skew}}(\matrixspaceA, \Lambda)$, is the space  spanned by all row vectors $v \in \mathbb{F}_p^n$ satisfying the following two conditions:
\begin{enumerate}
\item $x \cdot v^T = 0$ for each $x \in \Lambda$.
\item $\langle v \matrixspaceA\rangle$ is a subspace of $\langle \Lambda \rangle$.
\end{enumerate}

A matrix $C_{\mathrm{skew}}$ is a complementary matrix for skew-symmetric matrix space  $\matrixspaceA$ and attribute set $\Lambda$ if
\begin{enumerate}
\item $C_\mathrm{skew}$ is a full rank matrix in $M(n - \dim{\ker_{\mathrm{skew}}(\matrixspaceA, \lambda)}, n, \mathbb{F}_p)$.
\item The intersection of $\ker_{\mathrm{skew}}(\matrixspaceA, \Lambda)$ and the row vector space spanned by the rows of $C_\mathrm{skew}$
only contains the zero vector. 
\item Let $c_i$ be the $i$-th row of $C_\mathrm{skew}$. 
$x \cdot c_i^T = 0$ for all the $x \in \Lambda$ and $1 \leq i \leq n - \dim{\ker_{\mathrm{skew}}(\matrixspaceA, \lambda)} - |\Lambda|$.
\end{enumerate}


A matrix $P_\mathrm{skew} \in \mathrm{GL}(n, \mathbb{F}_p)$ is called a formatting matrix for attribute set $\Lambda$ and complementary matrix $C_\mathrm{skew}$ with respect to skew-symmetric matrix space $\matrixspaceA$,
where $C_\mathrm{skew}$ is a complementary matrix for $\matrixspaceA$ and $\Lambda$, 
if the following conditions hold:
\begin{enumerate}
\item $P_\mathrm{skew}$ is a full rank matrix.
\item The first $\dim{\ker_\mathrm{skew}(\matrixspaceA, \Lambda)}$ rows of $P_\mathrm{skew}$ form a linear basis of $\ker_\mathrm{skew}(\matrixspaceA, \Lambda)$.
\item The submatrix of $P_\mathrm{skew}$ on the last $n - \dim{\ker_\mathrm{skew}(\matrixspaceA, \Lambda)}$ rows equals $C_\mathrm{skew}$.
\end{enumerate}
\end{definition}

\begin{lemma}
Let $\matrixspaceA$ be an $n\times n$ skew-symmetric matrix space over $\mathbb{F}_p$, 
$\Lambda$ be an attribute set for $\matrixspaceA$,
and $C_\mathrm{skew}$ be a complementary matrix for $\matrixspaceA$ and $\Lambda$.
If $P_\mathrm{skew}$ is a formatting matrix for $\Lambda$ and $C_\mathrm{skew}$ with respect to $\matrixspaceA$, 
then for any $A \in \matrixspaceA$, 
\[\left(P_\mathrm{skew} A P_\mathrm{skew}^T\right)\left[1,  \dim{\ker_{\mathrm{skew}}(\matrixspaceA, \Lambda)}; 1, \dim{\ker_{\mathrm{skew}}(\matrixspaceA, \Lambda)}\right]\] 
is a zero matrix.
\end{lemma}
\begin{proof}
Let $v, v'$ be two arbitrary rows of the first $ \dim{\ker_{\mathrm{skew}}(\matrixspaceA, \Lambda)}$ rows of $P_\mathrm{skew}$. 
By Definition~\ref{def:skew_space_complement}, 
$v A$ is a linear combination of vectors in $\Lambda$.
Since $x \cdot v'^T = 0$ for each $x \in \Lambda$, $v A v'^T = 0$. 
\end{proof}

\begin{definition}[kernel, complementary matrix, and formatting matrix for general matrix spaces]\label{def:space_complement}
Let $\matrixspaceA$ be a matrix subspace of $M(m, n, \mathbb{F}_p)$. 
An attribute set $\Lambda$ for $\matrixspaceA$ is a set of linearly independent row vectors in $\mathbb{F}_p^n$.
The kernel for $\matrixspaceA$ and $\Lambda$, denoted as $\ker(\matrixspaceA, \Lambda)$, is  
the space  spanned by all row vectors $v \in \mathbb{F}_p^m$ such that
$\langle v \matrixspaceA\rangle$ is a subspace of $\langle \Lambda \rangle$.

A matrix $C$ is a complementary matrix for $\matrixspaceA$ and $\Lambda$ if the following conditions hold:
\begin{enumerate}
\item $C$ is a full rank matrix in $M(m - \dim{\ker(\matrixspaceA, \lambda)}, m, \mathbb{F}_p)$.
\item The intersection of 
$\ker(\matrixspaceA, \Lambda)$ and the row vector space spanned by all the rows of $C$
contains only the zero vector. 
\end{enumerate}


Given a complementary matrix $C$ of $\Lambda$ with respect to $\matrixspaceA$, 
a matrix $P$ is a left formatting matrix for $\Lambda$ and $C$ with respect to $\matrixspaceA$ if the following conditions hold:
\begin{enumerate}
\item $P$ is a matrix in $\mathrm{GL}(m, \mathbb{F}_p)$.
\item The first $ \dim{\ker(\matrixspaceA, \Lambda)}$ rows of $P$ form a linear basis of $\ker(\matrixspaceA, \Lambda)$.

\item The submatrix of $P$ on the last $ m - \dim{\ker(\matrixspaceA, \Lambda)}$ rows equal to $C$.
\end{enumerate}

A matrix $Q$ is a right formatting matrix for $\Lambda$ with respect to $\matrixspaceA$ if the following conditions hold:
\begin{enumerate}
\item $Q$ is a matrix in $\mathrm{GL}(n, \mathbb{F}_p)$.
\item Let $q_i$ be the $i$-th column vector of $Q$.
For any $1 \leq i \leq  n - |\Lambda|$, $x \cdot q_i = 0$ for any $x \in \Lambda$.
\end{enumerate}
\end{definition}

\begin{lemma}
Let $\matrixspaceA$ be an $m\times n$ matrix space, 
$\Lambda$ be an attribute set for $\matrixspaceA$, 
and $C$ be a complementary matrix for $\matrixspaceA$ and $\Lambda$.
If $P$ is a left formatting matrix for $\Lambda$ and $C$ with respect to $\matrixspaceA$, and 
$Q$ is a right formatting matrix for $\Lambda$ with respect to $\matrixspaceA$,
then for any $A \in \matrixspaceA$, \[(P A Q) [1, \dim{\ker(\matrixspaceA, \Lambda)}; 1, n - |\Lambda|]\] is a zero matrix. 

\end{lemma}

\begin{proof}
Let $v$ be an arbitrary row vector of the first $\dim{\ker(\matrixspaceA, \Lambda)}$ rows of $P$,
and $v'$ be an arbitrary column vector of the first $n - |\Lambda|$ columns of $Q$. 
For any $A \in \matrixspaceA$, 
by Definition~\ref{def:space_complement}, 
$v A$ is a linear combination of the row vectors in $\Lambda$.
Since $x \cdot v' = 0$ for each $x \in \Lambda$ by Definition~\ref{def:space_complement},  $v A v' = 0$. 
\end{proof}

As the main observation for the structure of low rank matrix spaces  (Lemma~\ref{lem:low_rank_main}), we show that for any low rank matrix space $\matrixspaceA$, there always exists a small attribute set such that the dimension of $\ker(\matrixspaceA, \Lambda)$ (or $\ker_\mathrm{skew}(\matrixspaceA, \Lambda)$ if $\matrixspaceA$ is a skew-symmetric matrix space) is large.

\begin{lemma}\label{lem:space_dim_prelim}
Let $A_1, A_2, \dots, A_k$ be $k$ matrices in $M(m, n, \mathbb{F}_p)$ for some prime $p$ and positive integers $k, m, n$. 
If there exist $d$ row vectors
 $x_1, x_2, \dots, x_d \in \mathbb{F}_p^m$ such that for every $1 \leq i < d$, 
the following condition holds
\begin{align}\label{equ:rank_bound}\left\langle \left\{x_j A_\ell : 1 \leq j \leq i, 1 \leq \ell \leq k\right\}\right\rangle \neq \left\langle \left\{x_j A_\ell : 1 \leq j \leq i + 1, 1 \leq \ell \leq k\right\} \right\rangle,
\end{align}
then there is a linear combination of $A_1, A_2, \dots, A_k$ with rank at least $(1 - 1 / p)d$.
\end{lemma}
\begin{proof}
Let $X$ be the $d \times n$ matrix such that $x_i$ is the $i$-th row of $X$.
To prove the lemma, it is sufficient to show that if
$\alpha_1, \dots, \alpha_k $ are uniformly and independently  sampled from $\{0, \dots, p-1\}$, then
we have the following expectation estimation. 
\[\Exp_{\alpha_1, \dots, \alpha_k}\left[\rank {X \left(\sum_{\ell=1}^k \alpha_\ell A_\ell\right) }\right] \geq \left(1 - \frac{1}{p}\right)d.\]
If this is the case, then there exist $\alpha_1, \dots, \alpha_k \in \mathbb{F}_p$ such that 
\[\rank{\sum_{\ell=1}^k \alpha_\ell A_\ell} \geq \rank {X \left(\sum_{\ell=1}^k \alpha_\ell A_\ell\right) } \geq \left(1 - \frac{1}{p}\right)d, \]
and then the lemma follows. 

For any $1 \leq i \leq d$, 
let $S_i$ be the set of row vectors
\[\left\{(\alpha_1, \dots, \alpha_k) \in \mathbb{F}_p^k : x_i \left(\sum_{\ell = 1}^k \alpha_\ell A_\ell \right) \in \left\langle \left\{x_j A_\ell : 1 \leq j \leq i - 1, 1 \leq \ell \leq k\right\} \right\rangle\right\}.\]
Since $\left\langle \left\{x_j A_\ell : 1 \leq j \leq i - 1, 1 \leq \ell \leq k\right\} \right\rangle$ is a row vector space,
if both $(\alpha_1, \dots, \alpha_k)$ and $(\beta_1, \dots, \beta_k)$ are in $S_i$, then $(\alpha_1 + \beta_1, \dots, \alpha_k + \beta_k)$ is also in $S_i$. 
Hence, the row vectors in $S_i$ form a subspace of $\mathbb{F}_p^k$.
By Inequality~(\ref{equ:rank_bound}), $S_i$ does not contain all the vectors in $\mathbb{F}_p^k$.
Hence $|S_i| \leq p^{k-1}$.
We have for each $1 \leq i \leq d$, 
\[\mathrm{Pr}_{\alpha_1, \dots, \alpha_k}\left[x_i \left(\sum_{\ell=1}^k \alpha_\ell A_\ell\right) \notin  \langle \{x_j A_\ell : 1 \leq j \leq i-1, 1 \leq \ell \leq k\} \rangle\right] = \frac{p^k - |S_i|}{p^k} \geq 1 - \frac{1}{p}.\]
If $x_i \left(\sum_{\ell=1}^k \alpha_\ell A_\ell\right)$ is not in $\left \langle \{x_j A_\ell : 1 \leq j < i, 1 \leq \ell \leq k\} \right\rangle$, 
then $x_i \left(\sum_{\ell=1}^k \alpha_\ell A_\ell\right)$ is not a linear combination of $x_1 \left(\sum_{\ell=1}^k \alpha_\ell A_\ell\right), \dots,  x_{i-1} \left(\sum_{\ell=1}^k \alpha_\ell A_\ell\right)$. 
Thus, we have \begin{align*}& \Exp_{\alpha_1, \dots, \alpha_k}\left[\rank {X \left(\sum_{\ell=1}^k \alpha_\ell A_\ell\right) }\right]  \\ \geq & 
\sum_{i=1}^d \mathrm{Pr}_{\alpha_1, \dots, \alpha_k}\left[x_i \left(\sum_{\ell=1}^k \alpha_\ell A_\ell\right) \notin  \langle \{x_j A_\ell : 1 \leq j < i, 1 \leq \ell \leq k\} \rangle\right]  \\
\geq & \left(1 - \frac{1}{p}\right)d.\end{align*}
\end{proof}

\begin{lemma}\label{lem:low_rank_main}
Let $\matrixspaceA$ be a matrix subspace of $M(m, n, \mathbb{F}_p)$ 
or a skew-symmetric matrix subspace of $\SS(n, \mathbb{F}_p)$
such that for each $A \in \matrixspaceA$, $\rank{A} \leq r$ for some positive integer $r$. 
There is an attribute set $\Lambda$ of size $O(r^2)$ for $\matrixspaceA$ such that
\begin{enumerate}
\item 
$\dim{\ker(\matrixspaceA, \Lambda)} \geq m - O(r)$ if $\matrixspaceA$ is a matrix subspace of $M(m, n, \mathbb{F}_p)$, or
\item 
$\dim{\ker_\mathrm{skew}(\matrixspaceA, \Lambda)} \geq n - O(r^2)$ if $\matrixspaceA$ is a skew-symmetric matrix subspace of $\SS(n, \mathbb{F}_p)$.
\end{enumerate}
\end{lemma}

\begin{proof}
We first consider the case that $\matrixspaceA$ is a matrix subspace of $M(m, n, \mathbb{F}_p)$.
Let $A_1, \dots, A_k$ be a linear basis for $\matrixspaceA$.
Let $x_1, x_2, \dots, x_d \in \mathbb{F}_p^m$ be $d$ row vectors for some positive integer $d$ 
 such that 
for any $1 \leq i \leq d$, there are at least $p^{k} / 2$ different $\alpha_1, \alpha_2, \dots, \alpha_k \in \mathbb{F}_p$ satisfying
\[\rank{X_i \left(\sum_{j =1}^k \alpha_j A_j\right)} > \rank{X_{i-1} \left(\sum_{j =1}^k \alpha_j A_j\right)},\] where for each $1 \leq i \leq d$,
$X_i$ is the $i \times m$ matrix with $x_j$ as the $j$-th row of $X_i$ for all $1 \leq j \leq i$.
Since every linear combination of $A_1, \dots, A_k$ is of rank at most $r$, we have 
\[d \cdot \frac{p^k}{2} \leq r \cdot p^k,\]
which implies $d \leq 2r$. Suppose there does not exist a vector $x_{d+1} \in \mathbb{F}_p^m$ such that there are at least $p^{k} / 2$ different $\alpha_1, \alpha_2, \dots, \alpha_k\in\mathbb{F}_p$ satisfying
\[\rank{X_{d+1} \left(\sum_{j =1}^k \alpha_j A_j\right)} > \rank{X_d \left(\sum_{j =1}^k \alpha_j A_j\right)}.\] 
For each row vector $v \in \mathbb{F}_p^m$ such that $v\notin \langle x_1, \dots, x_d\rangle$, 
there exist $\beta_1, \dots, \beta_d \in \mathbb{F}_p$ such that the following condition holds for at least $\frac{p^k}{ 2p^d}$ different $\alpha_1, \alpha_2, \dots, \alpha_k \in \mathbb{F}_p$
\[v \cdot \left(\sum_{i = 1}^k \alpha_k A_k\right) = \left(\sum_{i=1}^d \beta_i x_i\right)\left(\sum_{i = 1}^k \alpha_k A_k\right).\]
Thus, $\langle (v -  \sum_{i=1}^d \beta_i x_i) \matrixspaceA\rangle$ is a space of dimension at most $d + 1 = O(r)$. 
Hence, one can construct a matrix $P \in \mathrm{GL}(m, \mathbb{F}_p)$ satisfying the following properties: Let $p_i$ denote the row vector of the $i$-th row of $P$,
\begin{enumerate}
\item $p_i = x_i$  for any $1 \leq i \leq d$.
\item $\langle p_i \matrixspaceA\rangle$ is of dimension at most $O(r)$ for any $d + 1 \leq i \leq m$. 

\end{enumerate} 

Let $i_1, \dots, i_t \in \{d+1, \dots, m\}$ be a sequence of integers such that for each $1 \leq j \leq t$ 
\[\left\langle \bigcup_{\ell = 1}^{j - 1} \langle p_{i_\ell} \matrixspaceA \rangle \right\rangle \neq \left\langle \bigcup_{\ell = 1}^{j} \langle p_{i_\ell}\matrixspaceA \rangle \right\rangle. \]
By Lemma~\ref{lem:space_dim_prelim}, $t \leq r / (1 - 1/ p) \leq 2r$.
Hence, the dimension of
\[\mathfrak{S} = \left\langle \bigcup_{i = d+1}^{n} \langle p_i \matrixspaceA\rangle \right\rangle\]
is at most $2r \cdot O(r) = O(r^2)$.
Let $\Lambda$ be an arbitrary linear basis of $\mathfrak{S}$. 
We have $|\Lambda| = O(r^2)$. 
Since $\langle p_i \matrixspaceA \rangle$ is a subspace of $\langle \Lambda\rangle$ for each $d + 1 \leq i \leq m$, 
$\langle p_{d+1}, \dots, p_m \rangle$ is a subspace of $\ker(\matrixspaceA, \Lambda)$. 
Hence, $\dim{\ker(\matrixspaceA, \Lambda)} \geq m - d = m - O(r)$.

Now we consider the case that $\matrixspaceA$ is a skew-symmetric matrix subspace of $\SS(n, \mathbb{F}_p)$.
With a similar argument above, 
there is a set  $\Lambda$ of linearly independent row vectors over $\mathbb{F}_p^n$ satisfying the following conditions:
\begin{enumerate}
\item $|\Lambda| = O(r^2)$.
\item 
Let $\mathfrak{S}$ be the space spanned by row vectors $v \in \mathbb{F}_p^n$ such that $\langle v \matrixspaceA \rangle$ is a subspace of $\langle \Lambda \rangle$. 
$\dim{\mathfrak{S}} \geq n - O(r)$.
\end{enumerate}
Let $\mathfrak{T}$ be the space spanned by row vectors $v \in \mathbb{F}_p^n$ such that $x \cdot v^T = 0$ for all the $x\in \Lambda$. 
We have $\dim{\mathfrak{T}} = n - |\Lambda|$.
Hence, 
\begin{align*}\dim{\ker_{\mathrm{skew}}(\matrixspaceA, \Lambda)} = & \dim{\mathfrak{S} \cap \mathfrak{T}}  \\ \geq &  \dim{\mathfrak{S}} + \dim{\mathfrak{T}} - n \\ = & n - O(r) + n - |\Lambda| - n \\
= & n - O(r^2).
\end{align*}
\end{proof}

We also prove some useful properties for kernels with respect to a matrix space and an attribute set.
\begin{lemma}\label{lem:kernel_size_larger_lambda}
Let $\matrixspaceA$ be a matrix space. 
Let $\Lambda$ and $\Lambda'$ be two attribute sets of $\matrixspaceA$ such that $\Lambda$ is a subset of $\Lambda'$. 
Then $\ker(\matrixspaceA, \Lambda)$ is a subspace of $\ker(\matrixspaceA, \Lambda')$. 

Let $\matrixspaceB$ be a skew-symmetric subspace such that each matrix in $\matrixspaceB$ is of dimension $n \times n$. 
Let $\Delta$ and $\Delta'$ be two attribute sets of $\matrixspaceB$
 such that $\Delta$ is a subset of $\Delta'$. 
Then $\dim{\ker_\mathrm{skew}(\matrixspaceA, \Delta')}  \geq \dim{\ker_\mathrm{skew}(\matrixspaceA, \Delta)} - |\Delta'| + |\Delta|$.
\end{lemma}
\begin{proof}
By Definition~\ref{def:space_complement}, 
every row vector $x \in \ker(\matrixspaceA, \Lambda)$ is also a vector in $\ker(\matrixspaceA, \Lambda')$. Hence, $\ker(\matrixspaceA, \Lambda)$ is a subspace of $\ker(\matrixspaceA, \Lambda')$. 

For $\matrixspaceB$, 
every row vector $v \in \ker_\mathrm{skew}(\matrixspaceB, \Delta)$ satisfies the condition that $\langle v \matrixspaceB\rangle$ is a subspace of $\langle \Delta' \rangle$.
Let $\Delta'' = \Delta' \setminus \Delta$, and $S$ denote the space 
\[\langle \{v \in \mathbb{F}_p^n : x \cdot v^T = 0 \text{ for all } x \in \Delta''\}\rangle.\]
Hence, $\ker_\mathrm{skew}(\matrixspaceB, \Delta) \cap S$ is a subspace of $\ker_\mathrm{skew}(\matrixspaceB, \Delta')$.
Since $\dim{S} = n - |\Delta''|$,
\begin{align*}\dim{\ker_\mathrm{skew}(\matrixspaceB, \Delta')} \geq & \dim{\ker_\mathrm{skew}(\matrixspaceB, \Delta)} + \dim{S} - n \\
= &  \dim{\ker_\mathrm{skew}(\matrixspaceB, \Delta)} + n - |\Delta''| - n\\
= &  \dim{\ker_\mathrm{skew}(\matrixspaceB, \Delta)} - |\Delta'| + |\Delta|.
\end{align*}
\end{proof}


\begin{lemma}
Let $\matrixspaceA$ be a matrix subspace of $M(m, n, \mathbb{F}_p)$, 
$X$ be a matrix in $\mathrm{GL}(m, \mathbb{F}_p)$, and $Y$ be a matrix in $\mathrm{GL}(n, \mathbb{F}_p)$. 
For any attribute set $\Lambda$ of $\matrixspaceA$, 
let $\Lambda' = \{x Y : x \in \Lambda\}$. 
Then $\ker(\matrixspaceA, \Lambda) X^{-1} = \ker(X\matrixspaceB Y, \Lambda')$.  

Let $\matrixspaceB$ be a skew-symmetric matrix subspace such that each matrix in $\matrixspaceB$ is of dimension $n \times n$, 
and $S$ be a matrix in $\mathrm{GL}(n, \mathbb{F}_p)$. 
For any attribute set $\Delta$ of $\matrixspaceB$, 
let $\Delta' = \{x S^T: x \in \Delta\}$.
Then $\ker_\mathrm{skew}(\matrixspaceA, \Delta) S^{-1} = \ker_\mathrm{skew}(S\matrixspaceB S^T, \Delta')$.  

\end{lemma}

\begin{proof}
Let $v$ be an arbitrary row vector in $\ker(\matrixspaceA, \Lambda)$. By Definition~\ref{def:space_complement}, $\langle v A\rangle$ is a subspace of $\langle \Lambda \rangle$ for any $A \in \matrixspaceA$.
Hence, $\langle v X^{-1} X A\rangle$ is a subspace of $\langle \Lambda \rangle$, and thus 
$\langle v X^{-1} X A Y\rangle$ is a subspace of $\langle \Lambda' \rangle$.
By Definition~\ref{def:space_complement}, $v X^{-1}$ is a row vector in $\ker(X\matrixspaceB Y, \Lambda')$. 

Similarly, let $v'$ be an arbitrary row vector in $\ker(X\matrixspaceA Y, \Lambda')$. By Definition~\ref{def:space_complement}, $\langle v' X A Y\rangle$ is a subspace of $\langle \Lambda' \rangle$ for any $A \in \matrixspaceA$.
Hence, $\langle v' X A \rangle$ is a subspace of $\langle \Lambda \rangle$.
By Definition~\ref{def:space_complement}, $v' X$ is a row vector in $\ker(\matrixspaceB, \Lambda)$. 
Hence,  $\ker(\matrixspaceA, \Lambda) X^{-1} =  \ker(X\matrixspaceB Y, \Lambda')$.

Now we consider $\matrixspaceB$.
Let $v$ be an arbitrary row vector in $\ker_\mathrm{skew}(\matrixspaceB, \Delta)$. By Definition~\ref{def:skew_space_complement}, we have 
\begin{enumerate}
\item $\langle v B\rangle$ is a subspace of $\langle \Delta \rangle$ for any $B \in \matrixspaceA$.
\item $x \cdot v^T = 0$ for any $x\in \Delta$.
\end{enumerate}
Hence, $\langle v S^{-1} S B\rangle$ is a subspace of $\langle \Delta \rangle$, and thus 
$\langle v S^{-1} S B S^T\rangle$ is a subspace of $\langle \Delta' \rangle$.
In addition, for any $x' \in \Delta'$, $x' (S^{T})^{-1}$ is a vector in $\Delta$, and thus we have
\[x' \cdot (v S^{-1})^T = x' (S^{-1})^T \cdot v^T = x' (S^T)^{-1} \cdot v^T = 0.\]
Hence $v S^{-1}$ is a vector in $\ker_\mathrm{skew}(S\matrixspaceB S^T, \Delta')$.

Similarly, let $v'$ be an arbitrary row vector in $\ker_\mathrm{skew}(S\matrixspaceB S^T, \Delta')$. By Definition~\ref{def:skew_space_complement}, we have 
\begin{enumerate}
\item $\langle v' S BS^T\rangle$ is a subspace of $\langle \Delta' \rangle$ for any $B \in \matrixspaceA$.
\item $x' \cdot v'^T = 0$ for any $x'\in \Delta'$.
\end{enumerate}
Hence, $\langle v'  S B\rangle$ is a subspace of $\langle \Delta \rangle$.
In addition, for any $x \in \Delta$, $x S^T$ is a vector in $\Delta'$, and thus we have
\[x \cdot (v' S)^T = x \cdot S^T v'^T = x S^T \cdot v'^T = 0.\]
Hence $v' S$ is a vector in $\ker_\mathrm{skew}(\matrixspaceB, \Delta)$.
Hence,  $\ker_\mathrm{skew}(\matrixspaceB, \Delta) S^{-1} =  \ker_\mathrm{skew}(S\matrixspaceB S^T, \Delta')$.
\end{proof}

\begin{lemma}\label{lem:formatting_matrix_skew_non}
Let $\matrixspaceA$ be a matrix subspace of $M(m, n, \mathbb{F}_p)$ and $\matrixspaceB$ be a skew-symmetric matrix subspace of $\SS(n, \mathbb{F}_p)$. 
Let $\Lambda$ be a set of linearly independent row vectors in $\mathbb{F}_p^n$, and $C_\mathrm{skew}$ be a complementary matrix for $\matrixspaceB$ and $\Lambda$.
Then for any formatting matrix $P_\mathrm{skew}$ for $\Lambda$ and $C_\mathrm{skew}$ with respect to $\matrixspaceB$,
$P_\mathrm{skew}^T$ is a right formatting matrix for $\matrixspaceA$ and $\Lambda$.
\end{lemma}

\begin{proof}
By Definition~\ref{def:skew_space_complement}, for any row vector $v$ that corresponds to one of the first $n - |\Lambda|$ rows of $P_\mathrm{skew}$, 
$x \cdot v^T = 0$ for any $x \in \Lambda$.
By Definition~\ref{def:space_complement},  $P_\mathrm{skew}^T$ is a right formatting matrix for $\matrixspaceA$ and $\Lambda$.
\end{proof}
\section{Semi-canonical form of skew-symmetric matrix space tensors}\label{sec:semi_canonical_form_tensor}

In this section, we combine the matrix space individualization-refinement developed in Section~\ref{sec:individualization} and the low rank matrix space characterization developed in Section~\ref{sec:low_rank} 
to analyze the skew-symmetric matrix spaces.

The main result of this section is a structure that accommodates the matrix space  individualization-refinement and the low rank matrix space characterization. 
The purpose of such a structure is to establish a partial correspondence between matrices from two skew-symmetric matrix spaces after applying the two techniques to the skew-symmetric matrix spaces.


For convenience, we  use a 3-tensor representation of skew-symmetric matrix spaces and define the semi-canonical form of skew-symmetric matrix space tensors.
\subsection{Skew-symmetric matrix space tensors}
Following~\cite{li2017linear}, 
we define the tensor representation of a skew-symmetric matrix space. 

\begin{definition}\label{def:tensor_from_space}

Let $\matrixspaceG$ be a skew-symmetric matrix subspace of $\mathrm{SS}(n, \mathbb{F}_p)$ with dimension $m$.
A 3-tensor $\mathbf{G} \in \mathbb{F}_p^{m \times n \times n}$ is a tensor of $\matrixspaceG$ if 
$\matrixspaceG$ is equal to the space spanned by $A_1, \dots, A_m$, where 
$A_i[j, k] = \mathbf{G}[i, j, k]$ for all the $1 \leq i \leq m$ and $1 \leq j, k \leq n$, and $\mathbf{G}[i, j, k]$ is the $(i, j, k)$-th entry of $\mathbf{G}$.


\end{definition}




Given a skew-symmetric matrix space tensor $\mathbf{G}$, 
we use $\matrixspaceX_{\mathbf{G}, i}$ to denote the $n \times n$ matrix such that $\matrixspaceX_{\mathbf{G}, i}[j, k] = \mathbf{G}[i, j, k]$, 
use $\matrixspaceY_{\mathbf{G}, j}$ to denote the $m \times n$ matrix such that $\matrixspaceY_{\mathbf{G}, j}[i, k] = \mathbf{G}[i, j, k]$, and 
use $\matrixspaceZ_{\mathbf{G}, k}$ to denote the $m \times n$ matrix such that $\matrixspaceZ_{\mathbf{G}, k}[i, j] =  \mathbf{G}[i, j, k]$ for all the $1 \leq i \leq m, 1 \leq j, k \leq n$.
We also use $\matrixspaceX_\mathbf{G}$ to denote the space $\langle \matrixspaceX_{\mathbf{G}, 1}, \dots \matrixspaceX_{\mathbf{G}, m}\rangle $,
use $\matrixspaceY_\mathbf{G}$ to denote the space $\langle \matrixspaceY_{\mathbf{G}, 1}, \dots \matrixspaceY_{\mathbf{G}, n}\rangle $,
and use $\matrixspaceZ_\mathbf{G}$ to denote the space $\langle \matrixspaceZ_{\mathbf{G}, 1}, \dots \matrixspaceZ_{\mathbf{G}, n}\rangle $.


\begin{fact}\label{fact:basic_tensor}
Let $\mathbf{G}\in \mathbb{F}_p^{m \times n \times n}$ be the tensor for a skew-symmetric matrix space.
Then the following properties hold:
\begin{enumerate}
\item $\matrixspaceX_{\mathbf{G}}$ is a skew-symmetric matrix space of dimension $m$.
\item $\matrixspaceY_{\mathbf{G}}$ and $\matrixspaceZ_{\mathbf{G}}$ are matrix spaces of dimension $n$.
\item $\matrixspaceY_{\mathbf{G}, j} = - \matrixspaceZ_{\mathbf{G}, j}$ for all the $1 \leq j \leq n$.
\end{enumerate}
\end{fact}
\begin{proof}
The first and second properties are obtained by the definition of $\mathbf{G}$ and Fact~\ref{fact:baer_correspondance_tensor}. 
For the third property, since $\mathbf{G}[i, j, k] = - \mathbf{G}[i, k, j]$ holds for any $1 \leq i \leq m, 1 \leq j, k \leq n$,
we have \[\matrixspaceY_{\mathbf{G}, j}[i, k] = \mathbf{G}[i, j, k] = - \mathbf{G}[i, k, j]  = - \matrixspaceZ_{\mathbf{G}, j}[i, k]\]
for any $1 \leq i \leq m, 1 \leq j, k \leq n$.
\end{proof}

Let $N$ be a matrix in $\mathrm{GL}(n, \mathbb{F}_p)$ and  $M$ be a matrix in $\mathrm{GL}(m, \mathbb{F}_p)$.
The transform of $\mathbf{G}$ by $N$ and $M$, denoted as $\trans_{N, M}(\mathbf{G})$, is the tensor $\mathbf{H} \in \mathbb{F}_p^{m \times n \times n}$ 
such that 
\[\matrixspaceX_{\mathbf{H}, i} = \sum_{i' = 1}^m  M[i, i'] \cdot \left( N  \cdot \matrixspaceX_{\mathbf{G}, i'} \cdot N^T\right).\]
We define the isometry of two tensors. 

\begin{definition}\label{def:tensor_iso}
Let $\mathbf{G}, \mathbf{H} \in \mathbb{F}_p^{m \times n \times n}$ be tensors of two skew-symmetric matrix spaces. 
$\mathbf{G}$ and $\mathbf{H}$ are isometric if there are two matrices $N \in \mathrm{GL}(n, \mathbb{F}_p)$ and  $M \in \mathrm{GL}(m, \mathbb{F}_p)$ 
such that $\trans_{N, M}(\mathbf{G}) = \mathbf{H}$.
\end{definition}

\begin{lemma}\label{lem:matrixspace_tensor_iso_equivalent}
Let $\matrixspaceG$ and $\matrixspaceH$ be two skew-symmetric matrix spaces.  
$\matrixspaceG$ and $\matrixspaceH$ are isometric if and only if their tensors are isometric.
\end{lemma}
\begin{proof}
Let $\mathbf{G}$ and $\mathbf{H}$ be the tensors of $\matrixspaceG$ and $\matrixspaceH$, respectively.
Let $(G_1, \dots, G_m)$ be the linear basis of $\matrixspaceG$ such that 
$\mathbf{G}[i, j, k] = G_i[j, k]$ for all the $1 \leq i \leq m$ and $1 \leq j, k\leq n$.
Let $(H_1, \dots, H_m)$ be the linear basis of $\matrixspaceH$ such that 
 $\mathbf{H}[i, j, k] = H_i[j, k]$ for all the $1 \leq i \leq m$ and $1 \leq j, k\leq n$.

If $\matrixspaceG$ and $\matrixspaceH$ are isometric, 
then by Definition~\ref{def:skew_space_isometry}, 
there is a matrix $N \in \mathrm{GL}(n, \mathbb{F}_p)$ such that $N\matrixspaceG N^T = \matrixspaceH$.
Hence, there is a matrix $M \in \mathrm{GL}(m, \mathbb{F}_p)$ such that 

\[H_i = \sum_{i' = 1}^m M[i, i']\cdot\left( N \cdot G_i \cdot N^T\right)\] for all the $1 \leq i \leq m$. 
By the definition of tensor transform, we have $\trans_{N, M}(\mathbf{G}) = \mathbf{H}$.

If $\mathbf{G}$ and $\mathbf{H}$ are isometric, 
then by Definition~\ref{def:tensor_iso}, 
there are matrices $N \in \mathrm{GL}(n, \mathbb{F}_p)$ and $M \in \mathrm{GL}(m, \mathbb{F}_p)$ such that $\trans_{N, M}(\mathbf{G}) = \mathbf{H}$.
Hence, we have $H_i = \sum_{i' = 1}^m M[i, i']\cdot N \cdot G_i \cdot N^T$ for all the $1 \leq i \leq m$.
Thus, $\matrixspaceG$ and $\matrixspaceH$ are isometric. 
\end{proof}

\begin{fact}\label{fact:transform_product}
Let $\mathbf{G}$ be a skew symmetric matrix space tensor in $\mathbb{F}_p^{m \times n \times n}$.
Let $L_1$ and $L_2$ be two matrices in $\mathrm{GL}(n, \mathbb{F}_p)$, and $R_1$ and $R_2$ be two matrices in $\mathrm{GL}(m, \mathbb{F}_p)$.
Then 
\[\trans_{L_1 \cdot L_2, R_1\cdot R_2}(\mathbf{G}) = \trans_{L_1, R_1} \left(\trans_{L_2, R_2} (\mathbf{G})\right).\]
\end{fact}
\begin{proof}
By the definition of $\trans_{L_1 \cdot L_2, R_1\cdot R_2}(\mathbf{G}) $, 
we have 
\begin{align*} 
\matrixspaceX_{\trans_{L_1 \cdot L_2, R_1\cdot R_2}(\mathbf{G}) , i} = & \sum_{i' = 1}^m  (R_1\cdot R_2)[i, i'] \cdot \left( (L_1\cdot L_2)  \cdot \matrixspaceX_{\mathbf{G}, i'} \cdot (L_1\cdot L_2)^T  \right)\\
= & \sum_{i' = 1}^m \sum_{i'' = 1}^m R_1[i, i''] \cdot R_2[i'', i'] \cdot \left(L_1 \left( L_2  \cdot \matrixspaceX_{\mathbf{G}, i'} \cdot L_2^T \right) L_1^T \right) \\
= &   \sum_{i'' = 1}^m R_1[i, i'']  \left(L_1 \left( \sum_{i' = 1}^m \cdot R_2[i'', i'] \cdot  L_2  \cdot \matrixspaceX_{\mathbf{G}, i'} \cdot L_2^T \right) L_1^T \right)  \\
= &   \sum_{i'' = 1}^m R_1[i, i'']  \left(L_1 \cdot \matrixspaceX_{\trans_{L_2, R_2}(\mathbf{G}), i''}\cdot L_1^T \right) \\
= &   \matrixspaceX_{\trans_{L_1, R_1}\left(\trans_{L_2, R_2}(\mathbf{G})\right), i}.
\end{align*}

\end{proof}
\subsection{Semi-canonical form of skew-symmetric matrix space tensors}


We first give the intuition behind the semi-canonical form of a skew-symmetric tensor. 
Consider a skew-symmetric matrix space tensor $\mathbf{G}$. 
Suppose we apply left and right individualization matrices $L_\mathrm{skew}$ and $R_\mathrm{skew}$ for space $\matrixspaceX_\mathbf{G}$, and reorder matrices of $\matrixspaceX_{\mathbf{G}, 1}, \dots, \matrixspaceX_{\mathbf{G}, m}$ such that $(\matrixspaceX_{\mathbf{G}, 1}, \dots, \matrixspaceX_{\mathbf{G}, m})$ becomes a semi-canonical basis of $\matrixspaceX_{\mathbf{G}}$  with respect to $L_\mathrm{skew}$ and $R_\mathrm{skew}$.
Let $P_\mathrm{skew}$ be a formatting matrix for $\zero_{L_\mathrm{skew}, R_\mathrm{skew}}(\matrixspaceX_\mathbf{G})$ and an attribute set. 
If we apply the formatting matrix $P_\mathrm{skew}$ 
to each of $\matrixspaceX_{\mathbf{G}, 1}, \dots, \matrixspaceX_{\mathbf{G}, m}$ in a way that 
$\matrixspaceX_{\mathbf{G}, i}$ becomes $P_\mathrm{skew} \matrixspaceX_{\mathbf{G}, i} P_\mathrm{skew}^T$, then the first $\dim{\zero_{L_\mathrm{skew}, R_\mathrm{skew}}(\matrixspaceX_\mathbf{G})}$ matrices of $\matrixspaceX_{\mathbf{G}, 1}, \dots, \matrixspaceX_{\mathbf{G}, m}$ have non-zero entries only in the last few rows or columns.
See Figure~\ref{fig:ind_ref_low_rank}(a) for an illustration: The black layers correspond to the first  $\dim{\zero_{L_\mathrm{skew}, R_\mathrm{skew}}(\matrixspaceX_\mathbf{G})}$ matrices of $\matrixspaceX_{\mathbf{G}, 1}, \dots, \matrixspaceX_{\mathbf{G}, m}$, and the red layers correspond to the remaining matrices of $\matrixspaceX_{\mathbf{G}, 1}, \dots, \matrixspaceX_{\mathbf{G}, m}$. 
The rectangles enclosed by black dashed lines are zero submatrices in the first  $\dim{\zero_{L_\mathrm{skew}, R_\mathrm{skew}}(\matrixspaceX_\mathbf{G})}$ matrices of $\matrixspaceX_{\mathbf{G}, 1}, \dots, \matrixspaceX_{\mathbf{G}, m}$.
\begin{figure}[h]
\includegraphics[width=\textwidth]{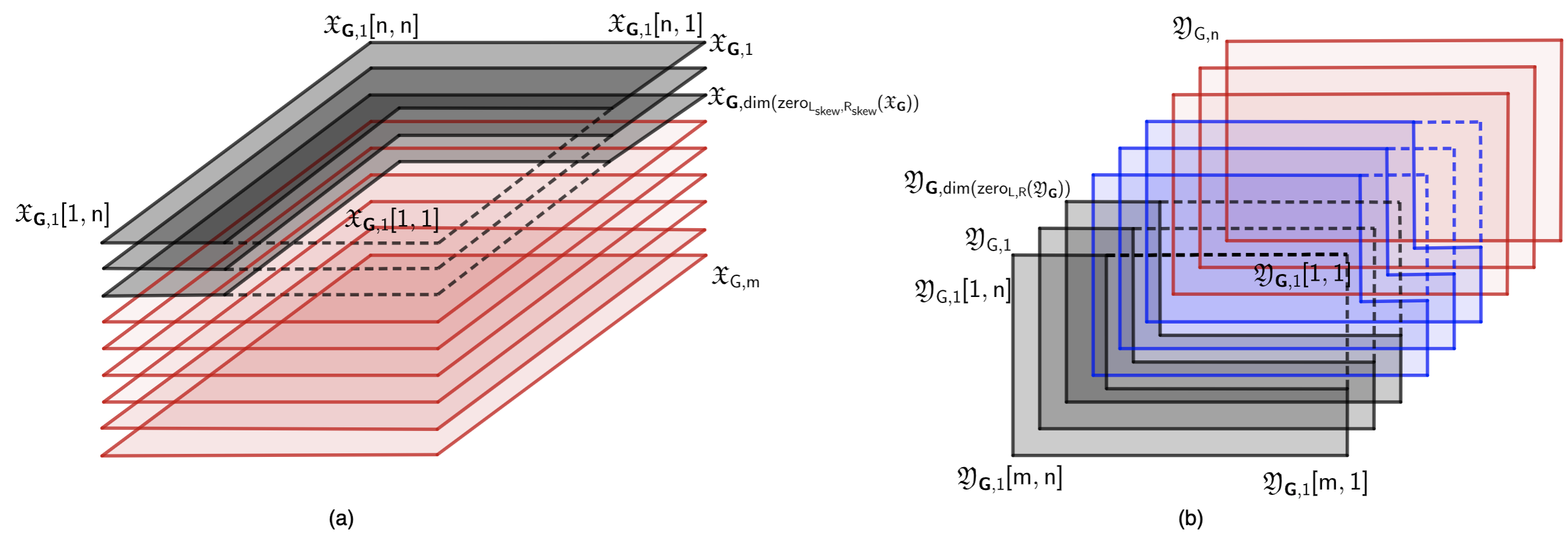}
\caption{(a) The matrix space individualization-refinement and the low rank matrix characterization of $\matrixspaceX_\mathbf{G}$, and (b) the matrix space individualization-refinement and the low rank matrix characterization of $\matrixspaceY_\mathbf{G}$}\label{fig:ind_ref_low_rank}
\end{figure}

We apply the same operation for the matrix space $\matrixspaceY_{\mathbf{G}}$.
Suppose we apply left and right individualization matrices $L$ and $R$ for space $\matrixspaceY_G$, and reorder matrices of $\matrixspaceY_{\mathbf{G}, 1}, \dots, \matrixspaceY_{\mathbf{G}, n}$ such that $(\matrixspaceY_{\mathbf{G}, 1}, \dots, \matrixspaceY_{\mathbf{G}, n})$ becomes a semi-canonical basis of $\matrixspaceY_{\mathbf{G}}$  with respect to $L$ and $R$.
If we further apply left formatting matrix $P$  and right formatting matrix $Q$ to each of $\matrixspaceY_{\mathbf{G}, 1}, \dots, \matrixspaceY_{\mathbf{G}, n}$ in a way that 
$\matrixspaceY_{\mathbf{G}, i}$ becomes $P \cdot \matrixspaceY_{\mathbf{G}, i} \cdot Q$, then the first $\dim{\zero_{L, R}(\matrixspaceY_\mathbf{G})}$ matrices of $\matrixspaceY_{\mathbf{G}, 1}, \dots, \matrixspaceY_{\mathbf{G}, n}$ have non-zero entries only in the last few rows or columns. 

To maintain the skew-symmetric property, we also apply the same operation on $\matrixspaceZ_\mathbf{G}$ using $\matrixspaceZ_{\mathbf{G}, j} = -\matrixspaceY_{\mathbf{G}, j}$ for any $1 \leq j \leq n$. Then the first $\dim{\zero_{L, R}(\matrixspaceZ_\mathbf{G})}$ matrices of $\matrixspaceZ_{\mathbf{G}, 1}, \dots, \matrixspaceZ_{\mathbf{G}, n}$ have non-zero entries only in the last few rows or columns. 
See Figure~\ref{fig:ind_ref_low_rank}(b) for an illustration: The black layers correspond to the first  $\dim{\zero_{L, R}(\matrixspaceY_\mathbf{G})}$ matrices of $\matrixspaceY_{\mathbf{G}, 1}, \dots, \matrixspaceY_{\mathbf{G}, n}$. 
The blue layers correspond to the matrices affected by the operation on $\matrixspaceZ_\mathbf{G}$ so that the last few columns have zero entries in the first few rows. 
The red layers are the remaining matrices of $\matrixspaceY_{\mathbf{G}, 1}, \dots, \matrixspaceY_{\mathbf{G}, n}$. 


We show that if the attribute sets used for $\matrixspaceX_{\mathbf{G}}$ and $\matrixspaceY_{\mathbf{H}}$ are the same, 
then 
with fixed individualization matrices and complementary matrices,
by carefully combining the matrix space individualization-refinement and low rank space characterization for $\matrixspaceX_\mathbf{G}$, $\matrixspaceY_\mathbf{G}$, and $\matrixspaceZ_\mathbf{G}$ together,  we obtain the semi-canonical form of the tensor as shown in Figure~\ref{fig:semi_tensor_intro}(a). 
The blue region is called the surface of the semi-canonical form of the tensor. 
The other region (the union of the red cube and the transparent region) is called the kernel of the semi-canonical form of the tensor. 
This is also the underlying reason for the term ``semi-canonical form'': The semi-canonical forms can be different with respect to the fixed attribute set,
individualization matrices, and complementary matrices, as the formatting matrices used can be different. But the kernels of different semi-canonical forms are the same 
(Lemma~\ref{lem:property_semi_canonical_form_tensor_iso}). 



If two tensors are isometric, and semi-canonical forms of the two tensors are obtained by 
the same individualization matrices, complementary matrices, and attribute sets (up to an isometry between the two tensors), 
then the kernels of the two semi-canonical forms are the same. 
So, to determine whether the two tensors are isometric, one only needs to check further if there are formatting matrices 
to make the surfaces of the two tensors to be identical.  \\








Before formally defining the semi-canonical form of a tensor, we first define the characterization tuple, which consists of the 
individualization matrices, the attribute set, and the complementary matrices used to define the semi-canonical form. 

\begin{definition}
Let $\mathbf{G} \in \mathbb{F}_{p}^{n\times n\times m}$ be a tensor for a skew-symmetric matrix space. 
A characterization tuple for a skew-symmetric matrix space tensor $\mathbf{G}$ is a 5-tuple 
\[(L_\mathrm{skew}, L, \Lambda, C_\mathrm{skew}, C)\] satisfying the following conditions:
\begin{enumerate}
\item $L_\mathrm{skew}$ is a matrix with $n$ columns such that $L_\mathrm{skew}$ and $L_\mathrm{skew}^T$ are left and right individualization matrices for the skew-symmetric matrix space $\matrixspaceX_{\mathbf{G}}$.
\item $L$ is a matrix with $m$ columns such that 
$L$ and $L_\mathrm{skew}^T$ are left and right individualization matrices for the matrix space $\matrixspaceY_{\mathbf{G}}$. 
\item $\Lambda$ is 
is an attribute set for both 
$\zero_{L_\mathrm{skew}, L_\mathrm{skew}^T}(\matrixspaceX_\mathbf{G})$ and $\zero_{L, L_\mathrm{skew}^T}(\matrixspaceY_\mathbf{G})$.
\item $C_\mathrm{skew}$ is a complementary matrix for $\zero_{L_\mathrm{skew}, L_\mathrm{skew}^T}(\matrixspaceX_\mathbf{G})$ and $\Lambda$.
\item 
$C$ is a complementary matrix for $\zero_{L, L_\mathrm{skew}^T}(\matrixspaceY_\mathbf{G})$ and $\Lambda$.
\end{enumerate}
\end{definition}

\begin{remark}
We remark that in the 5-tuple $(L_\mathrm{skew}, L, \Lambda, C_\mathrm{skew}, C)$, $L_\mathrm{skew}, L$ and $\Lambda$ can be arbitrary, but 
$C_\mathrm{skew}$ and $C$ are not arbitrary. 
$C_\mathrm{skew}$ needs to be a complementary matrix for  $\zero_{L_\mathrm{skew}, L_\mathrm{skew}^T}(\matrixspaceX_\mathbf{G})$ and $\Lambda$, and 
$C$ needs to be a complementary matrix for  $\zero_{L, L_\mathrm{skew}^T}(\matrixspaceY_\mathbf{G})$ and $\Lambda$.
\end{remark}



We further define a few notations for convenience.
Let \[\matrixspaceX_{\mathbf{G}, \ker(\matrixspaceY_{\mathbf{G}}, \Lambda)}\] be the space spanned by $\sum_{i = 1}^m v[i]\cdot \matrixspaceX_{\mathbf{G}, i}$
for all the $v \in \ker(\matrixspaceY_{\mathbf{G}}, \Lambda)$, and 
\[\matrixspaceY_{\mathbf{G}, \ker_\mathrm{skew}(\matrixspaceX_{\mathbf{G}}, \Lambda)}\] be the space spanned by $\sum_{j = 1}^n v[j]\cdot \matrixspaceY_{\mathbf{G}, j}$
for all the $v \in \ker_\mathrm{skew}(\matrixspaceX_{\mathbf{G}}, \Lambda)$.
Let 
\[\alpha_{\matrixspaceX, \mathbf{G}, L_\mathrm{skew}} \coloneqq \mathrm{dim}\left(\zero_{L_\mathrm{skew}, L_\mathrm{skew}^T}\left(\matrixspaceX_{\mathbf{G}, \ker(\matrixspaceY_{\mathbf{G}}, \Lambda)}\right)\right),\]  
\[\beta_{\matrixspaceX,  \mathbf{G}, L_\mathrm{skew}} \coloneqq \mathrm{dim}\left(L_\mathrm{skew}\cdot\matrixspaceX_{\mathbf{G}, \ker(\matrixspaceY_{\mathbf{G}}, \Lambda)}\cdot L_\mathrm{skew}^T\right),\]
\[\alpha_{\matrixspaceY,  \mathbf{G}, L_\mathrm{skew}, L} \coloneqq \mathrm{dim}\left( \zero_{L, L_\mathrm{skew}^T}\left(\matrixspaceY_{\mathbf{G}, \ker_\mathrm{skew}(\matrixspaceX_{\mathbf{G}}, \Lambda)}\right)\right),\] 
and 
\[\beta_{\matrixspaceY,  \mathbf{G}, L_\mathrm{skew}, L}  \coloneqq L\cdot\matrixspaceY_{\mathbf{G}, \ker_\mathrm{skew}(\matrixspaceX_{\mathbf{G}}, \Lambda)}\cdot L_\mathrm{skew}^T.\]
If $\mathbf{G}, L_\mathrm{skew}$ and $L$ are fixed and there is no confusion, 
we use $\alpha_{\matrixspaceX}, \beta_{\matrixspaceX}, \alpha_{\matrixspaceY}$, and $\beta_{\matrixspaceY}$ to denote
$\alpha_{\matrixspaceX, \mathbf{G}, L_\mathrm{skew}}$, $\beta_{\matrixspaceX,  \mathbf{G}, L_\mathrm{skew}}$, $\alpha_{\matrixspaceY,  \mathbf{G}, L_\mathrm{skew}, L}$ and 
$\beta_{\matrixspaceY,  \mathbf{G}, L_\mathrm{skew}, L}$, respectively.
We define the semi-canonical form for a skew-symmetric tensor space tensor as follows.

\begin{definition}\label{def:semi_canonial_tensor}
A semi-canonical form of $\mathbf{G} \in \mathbb{F}_p^{m\times n\times n}$ with respect to $(L_\mathrm{skew}, L, \Lambda, C_\mathrm{skew}, C)$, denoted as $\semic_{L_\mathrm{skew}, L, \Lambda, C_\mathrm{skew}, C}(\mathbf{G})$ (or $\semic(\mathbf{G})$ if there is no confusion), is a tensor with the same dimension as $\mathbf{G}$
such that $\semic_{L_\mathrm{skew}, L, \Lambda, C_\mathrm{skew}, C}(\mathbf{G}) = \trans_{N, M}(\mathbf{G})$ for some $N \in \mathrm{GL}(n, \mathbb{F}_p)$ and $M \in \mathrm{GL}(m, \mathbb{F}_p)$ satisfying the following two conditions:
\begin{enumerate}
\item $N$ is a formatting matrix for $\Lambda$ and $C_\mathrm{skew}$ with respect to the skew-symmetric matrix space $\zero_{L_\mathrm{skew}, L_\mathrm{skew}^T}(\matrixspaceX_\mathbf{G})$ such that for $\trans_{N, I_m}(\mathbf{G})$, 
\[\left(\matrixspaceY_{\trans_{N, I_m}(\mathbf{G}), 1} \cdot (N^{T})^{-1}, \dots, \matrixspaceY_{\trans_{N, I_m}(\mathbf{G}), \alpha_\matrixspaceY+\beta_\matrixspaceY}\cdot (N^{T})^{-1}\right)\] is a semi-canonical basis of the matrix space  $\matrixspaceY_{\mathbf{G}, \ker_\mathrm{skew}(\matrixspaceX_{\mathbf{G}}, \Lambda)}$ with respect to $L$ and $L_\mathrm{skew}^T$. 
\item $M$ is a left formatting matrix for $\Lambda$ and $C$ with respect to matrix space $\zero_{L, L_\mathrm{skew}^T}(\matrixspaceY_\mathbf{G})$ such that for $\trans_{I_n, M}(\mathbf{G})$, 
\[(\matrixspaceX_{\trans_{I_n, M}(\mathbf{G}), 1}, \dots, \matrixspaceX_{\trans_{I_n, M}(\mathbf{G}), \alpha_\matrixspaceX + \beta_\matrixspaceX})\] is a semi-canonical basis of matrix space $\matrixspaceX_{\mathbf{G}, \ker(\matrixspaceY_{\mathbf{G}}, \Lambda)}$ with respect to $L_\mathrm{skew}$ and $L_\mathrm{skew}^T$. 

\end{enumerate}
In addition, $\alpha_\matrixspaceX, \beta_\matrixspaceX, \alpha_\matrixspaceY$ and  $\beta_\matrixspaceY$ are the parameters of the semi-canonical form. 
\end{definition}


We remark that the semi-canonical form, with respect to fixed $(L_\mathrm{skew}, L, \Lambda, C_\mathrm{skew}, C)$, might not be unique using different transforming matrices $M$ and $N$.
But we show that the kernels of the semi-canonical forms with a fixed characterization tuple are identical (Lemma~\ref{lem:property_semi_canonical_form_tensor_iso}). 

\begin{lemma}\label{lem:property_semi_canonical_form_tensor}
Let $\semic(\mathbf{G})$ be a semi-canonical form of $\mathbf{G}$ with parameters $\alpha_\matrixspaceX, \beta_\matrixspaceX, \alpha_\matrixspaceY$, and $ \beta_\matrixspaceY$. Then $\semic(\mathbf{G})[i, j, k] = 0 $ if at least one of the following conditions holds:
\begin{enumerate}
\item $1 \leq i \leq \alpha_{\matrixspaceX}$, $1 \leq j, k \leq  \alpha_{\matrixspaceY} + \beta_\matrixspaceY$;
\item $1 \leq i \leq \alpha_{\matrixspaceX} + \beta_{\matrixspaceX}$, 
$1 \leq j \leq  \alpha_{\matrixspaceY}$, $1 \leq k \leq  \alpha_{\matrixspaceY} + \beta_\matrixspaceY$;
\item $1 \leq i \leq \alpha_{\matrixspaceX} + \beta_{\matrixspaceX}$, 
$1 \leq j \leq  \alpha_{\matrixspaceY} + \beta_\matrixspaceY $, $1 \leq k \leq  \alpha_{\matrixspaceY}$.
\end{enumerate}

\end{lemma}
\begin{proof}
By Definition~\ref{def:semi_canonial_tensor}, there are two matrices $M$ and $N$ such that $\semic(\mathbf{G}) = \trans_{N, M}(\mathbf{G})$.
Recall that $N$ is a formatting matrix for $\Lambda$ and $C_\mathrm{skew}$ with respect to \[\zero_{L_\mathrm{skew}, L_\mathrm{skew}^T}(\matrixspaceX_\mathbf{G}).\]
By Definition~\ref{def:semi_canonial_tensor},
for any $1 \leq i \leq \alpha_\matrixspaceX$, 
$\matrixspaceX_{\trans_{I_n, M}(\mathbf{G}), i} $ is a matrix in 
\[\zero_{L_\mathrm{skew}, L_\mathrm{skew}^T}\left(\matrixspaceX_{\mathbf{G}, \ker(\matrixspaceY_{\mathbf{G}}, \Lambda)}\right),\]
which is a subspace of $\zero_{L_\mathrm{skew}, L_\mathrm{skew}^T}\left(\matrixspaceX_{\mathbf{G}}\right)$.
Since 
\[\matrixspaceX_{\semic(\mathbf{G}), i} = N \cdot \matrixspaceX_{\trans_{I_n, M}(\mathbf{G}), i} \cdot N^T,\]
$\matrixspaceX_{\semic(\mathbf{G}), i}[1,  \alpha_\matrixspaceY + \beta_\matrixspaceY; 1,  \alpha_\matrixspaceY + \beta_\matrixspaceY]$
is a zero matrix.
Hence, if the first condition holds, then $\semic(\mathbf{G})[i, j, k] = 0 $.


Notice that $M$ is a left formatting matrix for $\Lambda$ and $C$ with respect to \[\zero_{L, L_\mathrm{skew}^T}(\matrixspaceY_\mathbf{G}),\]
and $N^T$ is a right formatting matrix for $\Lambda$ and $C$ with respect to \[\zero_{L, L_\mathrm{skew}^T}(\matrixspaceY_\mathbf{G})\] by Lemma~\ref{lem:formatting_matrix_skew_non}.
Since for any $1 \leq j \leq \alpha_\matrixspaceY$
\[\matrixspaceY_{\semic(\mathbf{G}), j} = M \cdot \matrixspaceY_{\trans_{N, I_m}(\mathbf{G}), i} = M \cdot Y_i \cdot N^T\]
for some $Y_i \in \matrixspaceY_{\mathbf{G}, \ker_\mathrm{skew}(\matrixspaceX_{\mathbf{G}}, \Lambda)}$ by Definition~\ref{def:semi_canonial_tensor},
$\matrixspaceY_{\semic(\mathbf{G}), j}[1,  \alpha_\matrixspaceX + \beta_\matrixspaceX; 1,  \alpha_\matrixspaceY + \beta_\matrixspaceY]$
is a zero matrix.
Hence, if the second condition holds, then $\semic(\mathbf{G})[i, j, k] = 0 $.


Since $\matrixspaceZ_{\semic(\mathbf{G}), j} = - \matrixspaceY_{\semic(\mathbf{G}), j}$ 
for all the $1 \leq j \leq n$, if the third condition holds, then $\semic(\mathbf{G})[i, j, k] = 0 $.
\end{proof}

\begin{lemma}\label{lem:property_semi_canonical_form_tensor_iso}
Let $\mathbf{G} \in \mathbb{F}_p^{n \times n \times m}$ be a tensor for a skew-symmetric matrix space.
If $\semic(\mathbf{G})$ and $\semic(\mathbf{G})'$ are two semi-canonical forms of $\mathbf{G}$ with respect to the same characterization tuple with parameters $\alpha_\matrixspaceX, \beta_\matrixspaceX, \alpha_\matrixspaceY$,  and $ \beta_\matrixspaceY$, 
then we have the following properties. 
\begin{enumerate}
\item There are two matrices 
\begin{equation}\label{equ:transform}
M^\dagger = \left(\begin{array}{ccc} X & 0 & 0 \\ Y & I_{\beta_\matrixspaceX} & 0 \\ 0 & 0 & I_{m - \alpha_\matrixspaceX - \beta_\matrixspaceX } \end{array} \right) \text{ and }
N^\dagger = \left(\begin{array}{ccc} X' & 0 & 0 \\ Y' & I_{\beta_\matrixspaceY} & 0 \\ 0 & 0 & I_{n -  \alpha_\matrixspaceY -  \beta_\matrixspaceY} \end{array} \right)
\end{equation}
for some $X \in \mathrm{GL}(\alpha_\matrixspaceX, \mathbb{F}_p)$,  $X' \in \mathrm{GL}(\alpha_\matrixspaceY, \mathbb{F}_p)$,
$Y \in M(\beta_\matrixspaceX, \alpha_\matrixspaceX, \mathbb{F}_p)$, and $Y' \in M(\beta_\matrixspaceY, \alpha_\matrixspaceY, \mathbb{F}_p)$
such that $\trans_{N^\dagger, M^\dagger}(\semic(\mathbf{G})) = \semic(\mathbf{G})'$.

\item 
$\semic(\mathbf{G})[i, j, k] = \semic(\mathbf{G})'[i, j, k]$ for any $1 \leq i \leq \alpha_{\matrixspaceX} + \beta_{\matrixspaceX}, 1 \leq j, k \leq \alpha_{\matrixspaceY} + \beta_{\matrixspaceY}$.
\end{enumerate}
\end{lemma}
\begin{proof}
Let $N$ and $M$ be two matrices such that $\semic(\mathbf{G}) = \trans_{N, M}(\mathbf{G})$. 
Let $N'$ and $M'$ be two matrices such that $\semic(\mathbf{G})' = \trans_{N', M'}(\mathbf{G})$.
We prove the first property with $N^\dagger = N' \cdot N^{-1}$ and  $M^\dagger = M' \cdot M^{-1}$.
By Fact~\ref{fact:transform_product}, we have $\trans_{N^\dagger, M^\dagger}(\semic(\mathbf{G})) = \semic(\mathbf{G})'$. So we only need to prove that $M^\dagger$ and $N^\dagger$ satisfy Equation~(\ref{equ:transform}). 

By Definition~\ref{def:semi_canonial_tensor} and Lemma~\ref{lem:property_zero_basis}, we have the following properties for the matrix space $\matrixspaceX_{\mathbf{G}}$:
\begin{enumerate}
\item[(a).] For the $\matrixspaceX_{\mathbf{G}, \ker(\matrixspaceY, \Lambda)}$ we have
\begin{align*}\zero_{L_\mathrm{skew}, L_\mathrm{skew}^T}\left(\matrixspaceX_{\mathbf{G}, \ker(\matrixspaceY, \Lambda)}\right) = & \left \langle \matrixspaceX_{\trans_{I_n, M}(\mathbf{G}), 1}, \dots, \matrixspaceX_{\trans_{I_n, M}(\mathbf{G}), \alpha_\matrixspaceX}\right\rangle \\ 
= &\left\langle \matrixspaceX_{\trans_{I_n, M'}(\mathbf{G}), 1}, \dots, \matrixspaceX_{\trans_{I_n, M'}(\mathbf{G}), \alpha_\matrixspaceX}\right\rangle.
\end{align*}
\item[(b).] For all the $\alpha_\matrixspaceX + 1 \leq i \leq \alpha_\matrixspaceX + \beta_\matrixspaceX$,
\[\matrixspaceX_{\trans_{I_n, M'}(\mathbf{G}), i} = \matrixspaceX_{\trans_{I_n, M}(\mathbf{G}), i} + X_i\] for some $X_i \in \zero_{L_\mathrm{skew}, L_\mathrm{skew}^T}\left(\matrixspaceX_{\mathbf{G}, \ker(\matrixspaceY, \Lambda)}\right)$
\item[(c).] For all the  $ \alpha_\matrixspaceX + \beta_\matrixspaceX + 1 \leq i \leq m$,
\[\matrixspaceX_{\trans_{I_n, M}(\mathbf{G}), i} = \matrixspaceX_{\trans_{I_n, M'}(\mathbf{G}), i} \]
\end{enumerate}
The property (a) implies that each of the first $\alpha_{\matrixspaceX}$ rows of $M'$ is a linear combination of the  first $\alpha_{\matrixspaceX}$ rows of $M$. 
The property (b) implies that 
for all the $\alpha_{\matrixspaceX}+ 1\leq i \leq \alpha_{\matrixspaceX}+\beta_{\matrixspaceX}$,
the $i$-th row of $M'$ is the $i$-th row of  $M$ plus a linear combination of the  first $\alpha_{\matrixspaceX}$ rows of $M$. 
The property (c) implies that 
for all the $\alpha_{\matrixspaceX}+\beta_{\matrixspaceX} \leq i \leq m$,
$i$-th row of $M'$ is the same as the $i$-th row of  $M$. 
Hence, $M^\dagger$ satisfies Equation~(\ref{equ:transform}). 

Also, by Definition~\ref{def:semi_canonial_tensor} and Lemma~\ref{lem:property_zero_basis}, we have the following properties for  the matrix space $\matrixspaceY$:
\begin{enumerate}
\item[(d).] For the $\matrixspaceY_{\mathbf{G}, \ker_\mathrm{skew}(\matrixspaceX, \Lambda)}$ we have
\begin{align*} & \zero_{L, L_\mathrm{skew}^T}\left(\matrixspaceY_{\mathbf{G}, \ker_\mathrm{skew}(\matrixspaceX, \Lambda)}\right) \\ = & \left\langle \matrixspaceY_{\trans_{N, I_m}(\mathbf{G}), 1} \cdot \left(N^T\right)^{-1}, \dots, \matrixspaceY_{\trans_{N, I_m}(\mathbf{G}), \alpha_\matrixspaceY } \cdot \left(N^T\right)^{-1}\right\rangle \\ 
= &\left\langle \matrixspaceY_{\trans_{N', I_m}(\mathbf{G}), 1} \cdot \left(N'^T\right)^{-1}, \dots, \matrixspaceY_{\trans_{N', I_m}(\mathbf{G}), \alpha_\matrixspaceY} \cdot \left(N'^T\right)^{-1}\right\rangle.
\end{align*}
\item[(e).] For all the $\alpha_\matrixspaceY + 1 \leq j \leq \alpha_\matrixspaceY + \beta_\matrixspaceY$,
\[\matrixspaceY_{\trans_{N', I_m}(\mathbf{G}), i}  \cdot \left(N'^T\right)^{-1}= (\matrixspaceY_{\trans_{ N, I_m}(\mathbf{G}), i} + Y_i) \cdot \left(N^T\right)^{-1}\] for some $Y_i \in \left\langle \matrixspaceY_{\trans_{N, I_m}(\mathbf{G}), 1} \dots, \matrixspaceY_{\trans_{N, I_m}(\mathbf{G}), \alpha_{\matrixspaceY}}\right\rangle $
\item[(d).] For all the  $ \alpha_\matrixspaceY + \beta_\matrixspaceY  + 1\leq j \leq n$,
\[\matrixspaceY_{\trans_{N', I_m}(\mathbf{G}), j} \cdot \left(N'^{T}\right)^{-1} = \matrixspaceY_{\trans_{N, I_m}\mathbf{G}), j} \cdot\left (N^T\right)^{-1}\]
\end{enumerate}

The property (d) implies that each of the first $\alpha_{\matrixspaceY}$ rows of $N'$ is a linear combination of the  first $\alpha_{\matrixspaceY}$ rows of $N$. 
The property (e) implies that 
for all $\alpha_{\matrixspaceY}+ 1\leq j \leq \alpha_{\matrixspaceY}+\beta_{\matrixspaceY}$,
the $j$-th row of $N'$ is the $j$-th row of  $N$ plus a linear combination of the  first $\alpha_{\matrixspaceY}$ rows of $N$. 
The property (f) implies that 
for all $\alpha_{\matrixspaceY}+\beta_{\matrixspaceY} + 1 \leq j \leq n$,
the $j$-th row of $N'$ is the same as the $j$-th row of  $N$. 
Hence, $N^\dagger$ satisfies Equation~(\ref{equ:transform}). Hence, the first property of the current lemma holds. 

The second property of the lemma is obtained by the first property and Lemma~\ref{lem:property_semi_canonical_form_tensor}.
\end{proof}

By the results from Section~\ref{sec:individualization} and Section~\ref{sec:low_rank}, we show that there is always a characterization tuple such that 
the attribute set contains a small number of row vectors,
and each of the individualization and complementary matrices contains a small number of rows. 
This means that for the isometry testing of two skew-symmetric matrix space tensors,
one can enumerate the characterization tuples so that the isometry testing of two tensors 
reduces to isometry testing of the semi-canonical forms of the two tensors.

\begin{lemma}\label{lem:small_individualization_semi_canonical_tensor}
For every skew-symmetric matrix space tensor $\mathbf{G} \in \mathbb{F}_p^{m \times n \times n}$, there exists a characterization tuple $(L_\mathrm{skew}, L, \Lambda, B_\mathrm{skew}, B)$ satisfying the following conditions: 
\begin{enumerate}
\item $L_\mathrm{skew}$ is a matrix in $M(O(\max\{m, n\} \log(p)/ n^{0.2}), n, \mathbb{F}_p)$.
\item $L$ is a matrix in $M(O(n \log(p) / n^{0.2}), m, \mathbb{F}_p)$.
\item $|\Lambda| = O(n^{0.4})$.
\item $C_\mathrm{skew}$ is a matrix in $M(O(n^{0.8}), n, \mathbb{F}_p)$.
\item $C$ is a matrix in $M(O(n^{0.8}), m, \mathbb{F}_p)$.
\end{enumerate}

\end{lemma}
\begin{proof}
Let $r = n^{0.4}$.
By Lemma~\ref{lem:individualzation_main}, there exist two matrices 
\[L_1 \in M(O(\max\{m \log(p), r\} / \sqrt{r}), n, \mathbb{F}_p) \text{ and }R_1 \in M(n, O(\max\{m \log(p), r\} / \sqrt{r}), \mathbb{F}_p)\] such that
$L_1 X R_1$ is a non-zero matrix for each $X \in \matrixspaceX_{\mathbf{G}}$ of rank at least $r$.
By Lemma~\ref{lem:individualzation_main}, there exist two matrices 
\[L_2 \in M(O(\max\{n \log(p), r\} / \sqrt{r}), m, \mathbb{F}_p) \text{ and } R_2 \in M(n, O(\max\{n \log(p), r\} / \sqrt{r}), \mathbb{F}_p)\] such that
$L_2 Y R_2$ is a non-zero matrix for each $Y \in \matrixspaceY_{\mathbf{G}}$ of rank at least $r$.
Let $L_\mathrm{skew}$ be a matrix such that every row vector of $L_1$, $R_1^T$, and $R_2^T$ is a linear combination of the row vectors of $L_\mathrm{skew}$, and $L$ be $L_2$.
Thus $L_\mathrm{skew}$ has \[O(\max\{m \log(p), n \log(p), r\} / \sqrt{r}) = O(\max\{m, n\}\log (p) / \sqrt{r}) = O(\max\{m, n\}\log (p) / n^{0.2})\] rows,
and $L$ has $O(n \log (p) / n^{0.2})$ rows.
We have that $L_\mathrm{skew} X L_\mathrm{skew}^T$ is a non-zero matrix for each $X \in \matrixspaceX_{\mathbf{G}}$ of rank at least $r$, and $L Y L_\mathrm{skew}^T$
is a non-zero matrix for each $Y \in \matrixspaceY_{\mathbf{G}}$ of rank at least $r$.

By Lemma~\ref{lem:low_rank_main}, there is a set of linearly independent row vectors in $\mathbb{F}_p^n$, denoted as $\Lambda'$, such that the following two conditions hold
\begin{enumerate}
\item $|\Lambda'| \leq O(r^2)$. 
\item $\dim{\ker_{\mathrm{skew}}(\zero_{L_\mathrm{skew}, L_\mathrm{skew}^T}(\matrixspaceX_{\mathbf{G}}}), \Lambda') \geq n - O(r^2)$.
\end{enumerate}

Also, by Lemma~\ref{lem:low_rank_main}, there is a set of linearly independent row vectors in $\mathbb{F}_p^n$, denoted as $\Lambda''$,  such that  the following two conditions hold
\begin{enumerate}
\item $|\Lambda''| \leq O(r^2)$. 
\item $\dim{\ker(\zero_{L, L_\mathrm{skew}^T}(\matrixspaceY_{\mathbf{G}}), \Lambda'')} \geq m - O(r)$.
\end{enumerate}

Let $\Lambda$ be a set of linear independent row vectors of size at most $|\Lambda' | + |\Lambda''| = O(r^2)$ such that $\langle \Lambda \rangle = \langle \Lambda' \cup \Lambda'' \rangle$.
By Lemma~\ref{lem:kernel_size_larger_lambda}, 
we have \[\ker\left(\zero_{L, L_\mathrm{skew}^T}(\matrixspaceY_{\mathbf{G}}), \Lambda''\right)  \leq \ker\left(\zero_{L, L_\mathrm{skew}^T}(\matrixspaceY_{\mathbf{G}}), \Lambda\right)\] and 
\[\mathrm{dim}\left(\ker_{\mathrm{skew}}\left(\zero_{L_\mathrm{skew}, L_\mathrm{skew}^T}(\matrixspaceX_{\mathbf{G}}), \Lambda\right)\right) \geq \mathrm{dim}\left(\ker_{\mathrm{skew}}\left(\zero_{L_\mathrm{skew}, L_\mathrm{skew}^T}(\matrixspaceX_{\mathbf{G}}), \Lambda'\right)\right) - |\Lambda|.\]
Thus, we have 
\[\mathrm{dim}\left(\ker_{\mathrm{skew}}\left(\zero_{L_\mathrm{skew}, L_\mathrm{skew}^T}(\matrixspaceX_{\mathbf{G}}), \Lambda\right)\right) \geq n - O(r^2) - |\Lambda| \geq n - O(n^2)\]
and 
\[\mathrm{dim}\left(\ker\left(\zero_{L, L_\mathrm{skew}^T}(\matrixspaceY_{\mathbf{G}}), \Lambda\right)\right) \geq m - O(r) \geq m - O(r^2).\]
Finally, any complementary matrix for the matrix space $\zero_{L_\mathrm{skew}, L_\mathrm{skew}^T}(\matrixspaceX_\mathbf{G})$ and $\Lambda$ has 
\[n - \mathrm{dim}\left(\ker_{\mathrm{skew}}\left(\zero_{L_\mathrm{skew}, L_\mathrm{skew}^T}(\matrixspaceX_{\mathbf{G}}), \Lambda\right)\right) = O(r^2)\] rows, and any complementary matrix for the matrix space $\zero_{L, L_\mathrm{skew}^T}(\matrixspaceY_\mathbf{G})$ and $\Lambda$ has 
\[m - \mathrm{dim}\left(\ker(\zero_{L, L_\mathrm{skew}^T}\left(\matrixspaceY_{\mathbf{G}}), \Lambda\right)\right) = O(r^2)\] rows.
\end{proof}

We present an algorithm to compute a semi-canonical form of a given skew-symmetric matrix space tensor based on a given characterization tuple. 

\begin{framed}
\noindent \textbf{Tensor Semi-Canonical Form Construction Algorithm}

\noindent \textbf{Input:} A skew-symmetric tensor $\mathbf{G}$ and a characterization tuple $(L_\mathrm{skew}, L, \Lambda, C_\mathrm{skew}, C)$

\noindent \textbf{Output:} $\semic_{L_\mathrm{skew}, L, \Lambda, C_\mathrm{skew}, C}(\mathbf{G})$

\begin{enumerate}
\item Compute a linear basis of $\ker_\mathrm{skew}(\matrixspaceX_\mathbf{G}, \Lambda)$ and a linear basis of $\ker(\matrixspaceY_\mathbf{G}, \Lambda)$. 
\item Compute a semi-canonical basis $(X_1, \dots, X_{\alpha_\matrixspaceX + \beta_\matrixspaceX})$ of $\matrixspaceX_{\mathbf{G}, \ker(\matrixspaceY_{\mathbf{G}}, \Lambda)}$ 
with respect to $L_\mathrm{skew}$ and $L_\mathrm{skew}^T$,  
and a semi-canonical basis $(Y_1, \dots, Y_{\alpha_\matrixspaceY + \beta_\matrixspaceY})$ of $\matrixspaceY_{\mathbf{G}, \ker_\mathrm{skew}(\matrixspaceX_{\mathbf{G}}, \Lambda)}$ 
with respect to $L$ and $L_\mathrm{skew}^T$.
\item Compute a matrix $M \in \mathrm{GL}(m, \mathbb{F}_p)$ satisfying the following two conditions
\begin{itemize}
\item  Let $m_i$ be the $i$-th row of $M$.
$\sum_{i'=1}^m m_i[i'] \cdot \matrixspaceX_{\mathbf{G}, i'}= X_i$ for all the $1 \leq i \leq \alpha_\matrixspaceX + \beta_\matrixspaceX$.
\item The $i$-th row $M$ is the same as the $(i - \alpha_\matrixspaceX - \beta_\matrixspaceX)$-th row of $C$ for all the $\alpha_\matrixspaceX + \beta_\matrixspaceX  + 1\leq i \leq m$.
\end{itemize}
\item Compute a matrix $N \in \mathrm{GL}(n, \mathbb{F}_p)$ satisfying the following two conditions
\begin{itemize}
\item  Let $n_i$ be the $i$-th row of $N$.
$\sum_{i'=1}^n n_i[i'] \cdot \matrixspaceY_{\mathbf{G}, i'}= Y_i$ for all the $1 \leq i \leq \alpha_\matrixspaceY + \beta_\matrixspaceY$.
\item The $i$-th row $N$ is the same as the $(i - \alpha_\matrixspaceY - \beta_\matrixspaceY)$-th row of $C_\mathrm{skew}$ for all the $\alpha_\matrixspaceY + \beta_\matrixspaceY + 1 \leq i \leq n$.
\end{itemize}
\item Return $\trans_{N, M}(\mathbf{G})$.
\end{enumerate}
\end{framed}

\begin{lemma}\label{lem:algo_semi_canonical_form_tensor_construction}
Given a skew-symmetric matrix space tensor $\mathbf{G} \in \mathbb{F}_p^{m \times n \times n}$ and a characterization tuple $(L_\mathrm{skew}, L, \Lambda, C_\mathrm{skew}, C)$, there is an algorithm to construct a semi-canonical form of $\mathbf{G}$ with respect to $(L_\mathrm{skew}, L, \Lambda, C_\mathrm{skew}, C)$ in time $p^{O(n + m)}\cdot \mathrm{poly}(n, m, p)$. 
\end{lemma}
\begin{proof}
The correctness of the algorithm is obtained by Definition~\ref{def:semi_canonial_tensor}. 
Now we bound the running time. 
For the first step of the algorithm, to compute a linear basis of $\ker_\mathrm{skew}(\matrixspaceX_\mathbf{G}, \Lambda)$, 
one can compute all the vectors in $\ker_\mathrm{skew}(\matrixspaceX_\mathbf{G}, \Lambda)$ by enumerating all the possible vectors and then compute a linear basis.
$\ker(\matrixspaceY_\mathbf{G}, \Lambda)$ can be computed similarly.
Hence, the first step takes $p^{O(n + m)}\cdot \mathrm{poly}(n, m, p)$ time. 
By Lemma~\ref{lem:algo_semi_canonial_base}, the second step of the algorithm takes $p^{O(n + m)}\cdot \mathrm{poly}(n, m, p)$ time. 
For the third and fourth steps, computing a row of $N$ (or $M$) can be done by enumerating all the row vectors of dimension $n$ (or $m$). Hence, the third and fourth steps  take $p^{O(n+ m)}\cdot \mathrm{poly}(n, m, p)$ time. 
\end{proof}

\subsection{Isometry of skew-symmetric matrix space tensor semi-canonical forms}
We define the isometry between semi-canonical forms of two tensors and give an algorithm for the isometry testing of two skew-symmetric matrix space tensors 
assuming there is an algorithm for the isometry testing of tensor semi-canonical forms.


\begin{definition}[Tensor semi-canonical form isometry]\label{def:skew_tensor_semi-canonical_form_iso}
Two skew-symmetric matrix space tensor semi-canonical forms $\semic(\mathbf{G})$ and $\semic(\mathbf{H})$ in $\mathbb{F}_p^{m \times n \times n}$ are isometric if 
the following two conditions hold:
\begin{enumerate}
\item The parameters $\alpha_\matrixspaceX, \beta_\matrixspaceX, \alpha_\matrixspaceY$ and $\beta_\matrixspaceY$ of the two semi-canonical forms are the same.
\item There exist matrices
\begin{equation}\label{equ:skew_tensor_semi-canonical_form_iso}
M = \left(\begin{array}{ccc} X & 0 & 0 \\ Y & I_{\beta_\matrixspaceX} & 0 \\ 0 & 0 & I_{m - \alpha_\matrixspaceX - \beta_\matrixspaceX } \end{array} \right) \text{ and }
N = \left(\begin{array}{ccc} X' & 0 & 0 \\ Y' & I_{\beta_\matrixspaceY} & 0 \\ 0 & 0 & I_{n -  \alpha_\matrixspaceY -  \beta_\matrixspaceY} \end{array} \right)
\end{equation}
for some $X \in \mathrm{GL}(\alpha_\matrixspaceX, \mathbb{F}_p)$,  $X' \in \mathrm{GL}(\alpha_\matrixspaceY, \mathbb{F}_p)$,
$Y \in M(\beta_\matrixspaceX, \alpha_\matrixspaceX, \mathbb{F}_p)$, and $Y' \in M(\beta_\matrixspaceY, \alpha_\matrixspaceY, \mathbb{F}_p)$
such that $\trans_{N, M}(\semic(\mathbf{G})) = \semic(\mathbf{H})$.
\end{enumerate}
\end{definition}

We give an algorithm for the isometry testing of two skew-symmetric matrix space tensors 
assuming there is an algorithm for the isomorphism testing of tensor semi-canonical forms.

\begin{framed}
\noindent \textbf{Skew-Symmetric Matrix Space Tensor Isometry Testing Algorithm} 

\noindent \textbf{Input:} Two skew-symmetric matrix space tensors $\mathbf{G}, \mathbf{H} \in \mathbb{F}_p^{m \times n \times n}$ for some prime $p > 2$ and positive integers $n, m$.


\noindent \textbf{Output:} Yes or no.

\begin{enumerate}
\item Let $\ell_1 = O((m+n) \log(p) / n^{0.2})$, $\ell_2 = O(n \log(p) / n^{0.2})$, and $\ell_3 = O(n^{0.4})$ and $\ell_4 = O(n^{0.8})$.
\item For each $(L_{\mathbf{G}, \mathrm{skew}}, L_{\mathbf{G}}, \Lambda_{\mathbf{G}}, C_{\mathbf{G}, \mathrm{skew}}, C_{\mathbf{G}})$
and $(L_{\mathbf{H}, \mathrm{skew}}, L_{\mathbf{H}}, \Lambda_{\mathbf{H}}, C_{\mathbf{H}, \mathrm{skew}}, C_{\mathbf{H}})$ satisfying the following conditions:

\begin{itemize}
\item $L_{\mathbf{G}, \mathrm{skew}}, L_{\mathbf{H}, \mathrm{skew}} \in M(\ell_1, n, \mathbb{F}_p),  L_{\mathbf{G}} , L_{\mathbf{H}} \in M(\ell_2, m, \mathbb{F}_p)$;
\item $|\Lambda_\mathbf{G}| = |\Lambda_\mathbf{H}| = \ell_3$;
\item $\dim{\ker(\zero_{L_{\mathbf{G}, \mathrm{skew}}, L_{\mathbf{G}, \mathrm{skew}}^T}(\matrixspaceX_\mathbf{G}), \Lambda_\mathbf{G})} \geq n - \ell_4$;
\item $\dim{\ker(\zero_{L_\mathbf{G}, L_{\mathbf{G}, \mathrm{skew}}^T}(\matrixspaceY_\mathbf{G}), \Lambda_\mathbf{G})} \geq m - \ell_4$;
\item $\dim{\ker(\zero_{L_{\mathbf{H}, \mathrm{skew}}, L_{\mathbf{H}, \mathrm{skew}}^T}(\matrixspaceX_\mathbf{H}), \Lambda_\mathbf{H})} \geq n - \ell_4$;
\item $\dim{\ker(\zero_{L_\mathbf{H}, L_{\mathbf{H}, \mathrm{skew}}^T}(\matrixspaceY_\mathbf{H}), \Lambda_\mathbf{H})} \geq m - \ell_4$,
\end{itemize}
run the following algorithm

\begin{enumerate}
\item Construct $\semic_{L_{\mathbf{G}, \mathrm{skew}}, L_{\mathbf{G}}, \Lambda_{\mathbf{G}}, C_{\mathbf{G}, \mathrm{skew}}, C_{\mathbf{G}}}(\mathbf{G})$ and $\semic_{L_{\mathbf{H}, \mathrm{skew}}, L_{\mathbf{H}}, \Lambda_{\mathbf{H}}, C_{\mathbf{H}, \mathrm{skew}}, C_{\mathbf{H}}}(\mathbf{H})$, and denote the resulting semi-canonical forms as $\semic_1$ and $\semic_2$, respectively.
\item If  $\semic_1$ and $\semic_2$
have different parameters, then continue. 
\item Run the algorithm for the isometry testing of semi-canonical forms for two skew-symmetric matrix tensors with $\semic_1$ and 
$\semic_2$. If the algorithm returns yes, then return yes. 
\end{enumerate}
\item Return no.
\end{enumerate}

\end{framed}

\begin{lemma}\label{lem:tensor_iso_algo}
If there is an algorithm to determine whether the semi-canonical forms of two skew-symmetric matrix space tensors in $\mathbb{F}_p^{m \times n \times n}$ are isometric with running time $T(p, n, m)$,
then there is an algorithm for isometry testing of two skew-symmetric matrix space tensors in time 
\[p^{O((m+ n) n^{0.8} \log(p))}\cdot T(p, n, m) \cdot \mathrm{poly}(p, n, m).\]
\end{lemma}
\begin{proof}
Let $\mathbf{G}$ and $\mathbf{H}$ be the two input tensors. 
We first prove the correctness of the algorithm. 
If $\mathbf{G}$ and $\mathbf{H}$ are isometric, then there are $M\in\mathrm{GL}(m, \mathbb{F}_p)$ and  $N\in\mathrm{GL}(n, \mathbb{F}_p)$ such that 
$\trans_{N, M}(\mathbf{G}) = \mathbf{H}$.
By Lemma~\ref{lem:small_individualization_semi_canonical_tensor}, there is a characterization tuple \[(L_{\mathbf{G}, \mathrm{skew}}, L_{\mathbf{G}}, \Lambda_{\mathbf{G}}, C_{\mathbf{G}, \mathrm{skew}}, C_{\mathbf{G}})\] for tensor $\mathbf{G}$ enumerated by the algorithm  satisfying the conditions of Lemma~\ref{lem:small_individualization_semi_canonical_tensor}.
Let  $ \Lambda_{\mathbf{H}} = \{x\cdot N^T : x \in \Lambda_\mathbf{G}\}$.
Then the tuple
\[(L_{\mathbf{G}, \mathrm{skew}}  \cdot N^{-1}, L_{\mathbf{G}} \cdot M^{-1}, \Lambda_{\mathbf{H}}, C_{\mathbf{G}, \mathrm{skew}} \cdot N^{-1}, C_{\mathbf{G}} \cdot M^{-1})\] is a characterization tuple for  $\mathbf{H}$, and is enumerated in the algorithm.
Hence, the algorithm returns yes after running the algorithm for skew-symmetric tensor semi-canonical form isometry testing on the semi-canonical forms of $\mathbf{G}$ and $\mathbf{H}$ with respect to the above two characterization tuples, respectively.

On the other hand, 
if the algorithm for skew-symmetric tensor semi-canonical form isometry testing returns yes on two semi-canonical forms, then there is a transform to make the two tensors equal.
Hence, the algorithm returns yes if and only 
if $\mathbf{G}$ and $\mathbf{H}$ are isometric. 

Now we bound the running time of the algorithm.
Since $L_{\mathbf{G}, \mathrm{skew}}$ and $L_{\mathbf{H}, \mathrm{skew}}$ are of dimension $O(\max\{m, n\} \log(p)/ n^{0.2}) \times n$, 
$L_{\mathbf{G}}$ and $L_{\mathbf{H}}$ are of dimension $O(n \log(p) / n^{0.2}) \times m$, 
$\Lambda_{\mathbf{G}}$ and  $\Lambda_{\mathbf{H}}$ contain at most $O(n^{0.4})$ vectors, 
each of $C_{\mathbf{G}, \mathrm{skew}}$  
$C_{\mathbf{H}, \mathrm{skew}}$, $C_{\mathbf{G}}$ and 
$C_{\mathbf{H}}$ contains at most $O(n^{0.8})$ rows, 
there are at most 
\[p^{O((m+n) n^{0.8}\log (p))} \cdot p^{O(n^{0.4} \cdot n)} \cdot p^{O(n^{0.8} (n + m))} = p^{O((n + m)n^{0.8}\log(p))}\] 
different pairs of $(L_{\mathbf{G}, \mathrm{skew}}, L_{\mathbf{G}}, \Lambda_{\mathbf{G}}, C_{\mathbf{G}, \mathrm{skew}}, C_{\mathbf{G}})$
and $(L_{\mathbf{H}, \mathrm{skew}}, L_{\mathbf{H}}, \Lambda_{\mathbf{H}}, C_{\mathbf{H}, \mathrm{skew}}, C_{\mathbf{H}})$ enumerated in step 2 of the algorithm.
By Lemma~\ref{lem:algo_semi_canonical_form_tensor_construction}, for each enumerated pair of 
\[(L_{\mathbf{G}, \mathrm{skew}}, L_{\mathbf{G}}, \Lambda_{\mathbf{G}}, C_{\mathbf{G}, \mathrm{skew}}, C_{\mathbf{G}}) \text{ and }
(L_{\mathbf{H}, \mathrm{skew}}, L_{\mathbf{H}}, \Lambda_{\mathbf{H}}, C_{\mathbf{H}, \mathrm{skew}}, C_{\mathbf{H}}),\] the running time of step 2(a) to step 2(c) is \[p^{O(n+m)}\cdot \mathrm{poly}(n, m, p) + T(p, n, m).\]
Then we obtain the desired overall running time. 
\end{proof}

\section{Isometry testing for skew-symmetric matrix space tensor semi-canonical forms}\label{sec:reduction}

In this section, we present an algorithm to determine whether the semi-canonical forms of two skew-symmetric tensors in $\mathbb{F}_p^{m \times n \times n}$ are isometric in $\mathrm{poly}(p, n, m)$ time for any prime $p > 2$.

Our approach constructs a skew matrix tuple for each semi-canonical form so that the isometry testing of tensor semi-canonical forms 
reduces to deciding whether the two skew-symmetric matrix tuples have a block diagonal isometry (Lemma~\ref{lem:restricted_tuple_isometry}). 
Making use of the properties of the skew matrix tuples constructed, 
we further show that the problem of deciding whether the two skew-symmetric tuples have a block diagonal isometry reduces to the skew-symmetric tuple isometry problem and the matrix tuple equivalence problem (Lemma~\ref{lem:test_semi_canonical_form}).




\subsection{Matrix tuple for skew-symmetric matrix space tensor semi-canonical forms}
We present our skew-symmetric matrix tuple construction for a skew-symmetric matrix space tensor semi-canonical form in this subsection.

Before formally defining our matrix tuple, we first give a high level overview of our construction.
Given a semi-canonical form $\semic(\mathbf{G})$ of a skew-symmetric matrix space tensor  $\mathbf{G}$ with parameters $\alpha_\matrixspaceX, \beta_\matrixspaceX, \alpha_\matrixspaceY$, and $\beta_\matrixspaceY$,
we let \[n'\coloneqq\alpha_\matrixspaceY+ \beta_\matrixspaceY \text{ and }m'\coloneqq\alpha_\matrixspaceX+ \beta_\matrixspaceX.\]

In the skew matrix tuple of $\semic(\mathbf{G})$, denoted as $\FF_{\semic(\mathbf{G})}$, all the matrices are in $\SS(3 + n + m', \mathbb{F}_p)$. 
The fourth row to the $(3+n)$-th row of matrices in $\FF_{\semic(\mathbf{G})}$ correspond to the rows of matrices in the matrix space $\matrixspaceX_{\mathbf{G}}$.
The last $m'$ rows of matrices in $\FF_{\semic(\mathbf{G})}$ correspond to the first $m'$ rows of matrices in $\matrixspaceY_\mathbf{G}$ (or equivalently $\matrixspaceZ_\mathbf{G}$). 
The first three rows of matrices in $\FF_{\semic(\mathbf{G})}$ are auxiliary rows used to fix the correspondence between the $(4 + \alpha_\matrixspaceY)$ to $(3 +n)$-th rows of matrices in $\FF_{\semic(\mathbf{G})}$ and the $(\alpha_\matrixspaceY + 1)$-th row to $n$-th row of matrices in $\matrixspaceX_{\mathbf{G}}$, as well as the
correspondence between the $(4 + n + \alpha_\matrixspaceX)$-th row to the $(3 + n + m')$-th row of matrices in $\FF_{\semic(\mathbf{G})}$ and the $(\alpha_\matrixspaceX + 1)$-th row to $m'$-th row of matrices in $\matrixspaceY_{\mathbf{G}}$. 
\begin{figure}[h]
\begin{center}
\includegraphics[width = \textwidth]{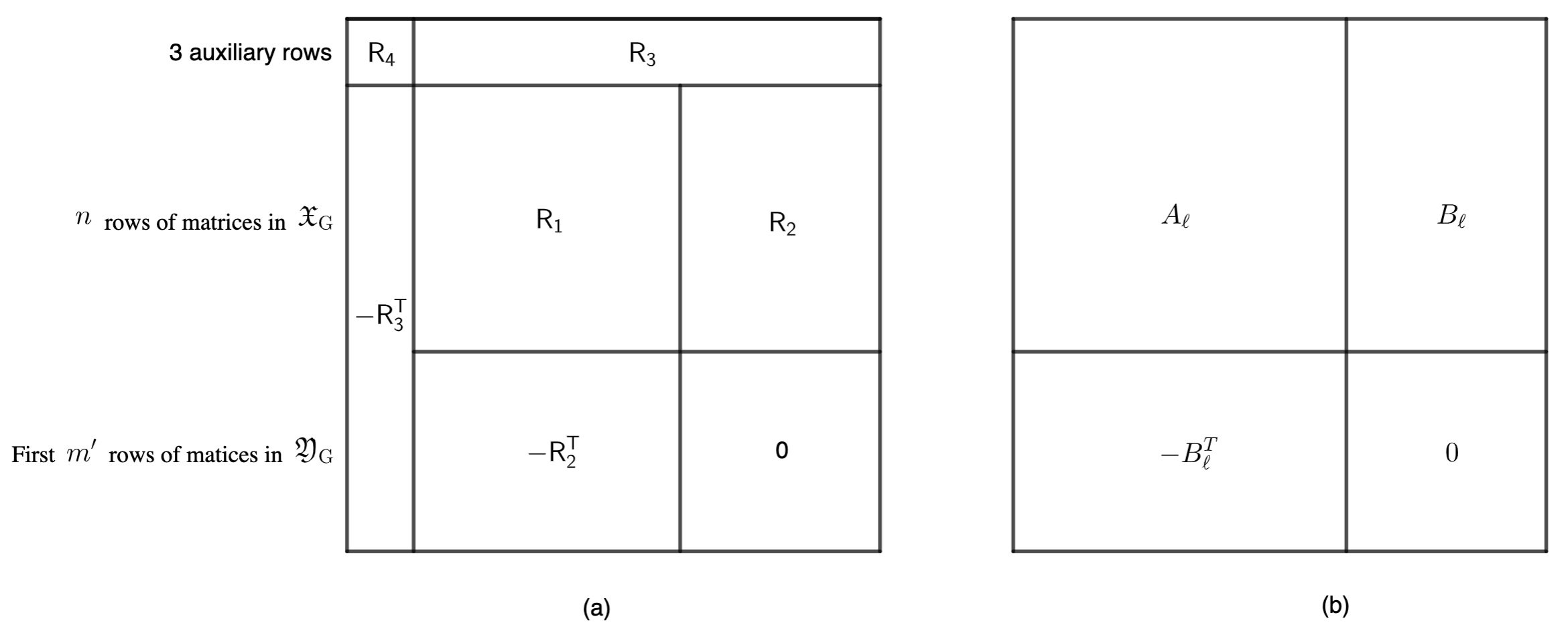}
\end{center}
\vspace{-.4cm}\caption{(a) The matrices in $\FF_{\semic(\mathbf{G})}$. (b) $A_\ell$ and $B_\ell$.}\label{fig:fig_semi_canonical_iso}
\end{figure}

See Figure~\ref{fig:fig_semi_canonical_iso}(a) for an illustration. 
The submatrices on $R_1$ for all the matrices in $\FF_{\semic(\mathbf{G})}$ are used to encode the matrices of the kernel (the second step in the construction of $\FF_{\semic(\mathbf{G})}$), as well as the skew-symmetric matrices in the surface for $\matrixspaceX_{\mathbf{G}}$, i.e., $\matrixspaceX_{\mathbf{G}, m' + 1}, \dots, \matrixspaceX_{\mathbf{G}, m}$ (the second, third, and fourth step in the construction of $\FF_{\semic(\mathbf{G})}$).
The submatrices on $R_2$ for all the matrices in $\FF_{\semic(\mathbf{G})}$ are used to encode the matrices in the surface of $\matrixspaceY_{\mathbf{G}}$, excluding the intersection with the surface of $\matrixspaceX_{\mathbf{G}}$ (the sixth and seventh step in the construction of $\FF_{\semic(\mathbf{G})}$).
Consequently, 
submatrices on the last $m'$ rows and columns from the $4$-th to the $(3 + n)$-th ($-R_2^T$ in Figure~\ref{fig:fig_semi_canonical_iso}(a)) for all the matrices in $\FF_{\semic(\mathbf{G})}$ 
encode the matrices in the surface of $\matrixspaceZ_{\mathbf{G}}$, excluding the intersection with the surface of $\matrixspaceX_{\mathbf{G}}$ . 
The submatrices on $R_3$ for all the matrices in $\FF_{\semic(\mathbf{G})}$ are used to fix the correspondence between some rows of matrices in $\FF_{\semic(\mathbf{G})}$ and some rows in the matrices of $\matrixspaceX_\mathbf{G}$ and $\matrixspaceY_\mathbf{G}$ (the fifth and eighth steps in the construction of $\FF_{\semic(\mathbf{G})}$). 
The submatrices on $R_4$ for all the matrices in $\FF_{\semic(\mathbf{G})}$ are used to fix the first three rows (the first step in the construction of $\FF_{\semic(\mathbf{G})}$). 
The submatrices on the last $m'$ rows and the last $m'$ columns for all the matrices in $\FF_{\semic(\mathbf{G})}$ are always zero matrices. 

The construction of the skew-symmetric matrix tuple for $\semic(\mathbf{G})$ is defined in Figure~\ref{fig:ff_construction}. The matrix tuple contains $t$ matrices 
for some $t = \mathrm{poly}(n, m)$. 

\begin{figure}[h!]
\begin{framed} 
\noindent{} \textbf{$\FF_{\semic(\mathbf{G})} = (F_1, \dots, F_t)$ Construction} 

\vspace{.2cm}\noindent (For each matrix, the undefined entries are zeros.)


\begin{enumerate}
\item Let $t_1 = 3$. $F_1(1, 2) = 1, F_1(2, 1) = -1$, $F_2(1, 3) = 1, F_2(3, 1) = -1$, $F_3(2, 3) = 1$,  and $F_3(3, 2) = -1$. 
\item Let $t_2 = t_1 + m - \alpha_\matrixspaceX$. For each $t_1 + 1 \leq \ell \leq t_2$, the submatrix $F_\ell[4, n' + 3;4, n'+3]$ equals $\matrixspaceX_{\semic(\mathbf{G}), \alpha_\matrixspaceX+ (\ell -t_1)}[1, n';1,n']$. 
\item Let $t_3 = t_2 + m - m'$. For each $t_2 + 1 \leq \ell \leq t_3$, the submatrix $F_\ell[n'+4, n +3; n' + 4, n + 3]$ equals $\matrixspaceX_{\semic(\mathbf{G}), m' + (\ell -t_2)}[n' + 1, n;n' + 1, n]$. 

\item Let $t_4 = t_3 + m - m'$. For each $t_3 + 1 \leq \ell \leq t_4$, the submatrix $F_\ell[4, n' + 3; n' + 4, n + 3]$ equals $\matrixspaceX_{\semic(\mathbf{G}), m' + (\ell -t_3)}[1, n';n' + 1, n]$, and the submatrix of $F_\ell[n' + 4, n + 3; 4, n' + 3]$ equals $\matrixspaceX_{\semic(\mathbf{G}), m' + (\ell -t_3)}[n' + 1, n;1, n']$. 


\item Let $t_5 = t_4 + 2(n - \alpha_\matrixspaceY)$. For each $1\leq \ell \leq (n - \alpha_\matrixspaceY)$, $F_{t_4 + 2\ell - 1}(1, 3 + \alpha_\matrixspaceY + \ell) = 1, F_{t_4 + 2\ell - 1}(3 + \alpha_\matrixspaceY + \ell, 1) = -1$, 
$F_{t_4 + 2\ell}(2, 3 + \alpha_\matrixspaceY + \ell) = 1$, and $F_{t_4 + 2\ell}(3 + \alpha_\matrixspaceY + \ell, 2) = -1$. 


\item Let $t_6 = t_5 + n - \alpha_\matrixspaceY$. For each $t_5 + 1 \leq \ell \leq t_6$, 
the submatrix $F_\ell[3 + n + 1, 3 + n + m';4, n' + 3]$ equals $\matrixspaceY_{\semic(\mathbf{G}), \ell - t_5 + \alpha_\matrixspaceY}[1, m'; 1, n']$, and
the submatrix $F_\ell[4, n' + 3; 3 + n + 1, 3 + n + m']$ equals $-(\matrixspaceY_{\semic(\mathbf{G}), \ell - t_5 + \alpha_\matrixspaceY}[1, m'; 1, n'])^T$.

\item Let $t_7 = t_6 + n - n'$. For each $t_6 + 1 \leq \ell \leq t_7$, 
$F_\ell[3 + n + 1, 3 + n + m'; n' + 4, n + 3]$ equals $\matrixspaceY_{\semic(\mathbf{G}), \ell - t_6 + n'}[1, m'; n' + 1, n]$. 
$F_\ell[n' + 4, n + 3; 3 + n + 1, 3 + n + m']$ equals $-(\matrixspaceY_{\semic(\mathbf{G}), \ell - t_6 + n'}[n'+1, n; 1, m'])^T$.

\item Let $t = t_7 + 2 \beta_{\matrixspaceX}$. For each $1\leq \ell \leq \beta_{\matrixspaceX}$, $F_{t_7 + 2\ell - 1}(1, 3 + n + \alpha_{\matrixspaceX} + \ell) = 1$,  $F_{t_7 + \ell}(3 + n + \alpha_{\matrixspaceX} + \ell, 1) = -1$, $F_{t_7 + 2\ell}(2, 3 + n + \alpha_{\matrixspaceX} + \ell) = 1$,  and $F_{t_7 + 2\ell}(3 + n + \alpha_{\matrixspaceX} + \ell, 2) = -1$. 

\end{enumerate}
\end{framed}
\caption{$\FF_{\semic(\mathbf{G})}$ construction}\label{fig:ff_construction}
\end{figure}

As illustrated in Figure~\ref{fig:fig_semi_canonical_iso}(b), 
for a matrix $F_\ell \in \FF_{\semic(\mathbf{G})}$, 
we use $A_\ell$ to denote the submatrix on the first $3 + n$ rows and the first $3 + n$ columns (i.e., $F_\ell[1, 3 + n; 1, 3+n]$),
and use $B_\ell$ to denote the submatrix on the first $3 + n$ rows and the last $m'$ columns (i.e., $F_\ell[1, 3 + n; 4 + n, 3 + n + m']$).
By the skew-symmetric condition, the submatrix on the last $m'$ rows and the first $3 + n$ columns is $-B_\ell^T$.

All the matrices in $\FF_{\semic(\mathbf{G})}$ have two types. 
If $B_\ell$ is a zero matrix, then 
$F_\ell$ is a type 1 matrix. 
If $A_\ell$ is a zero matrix, then 
$F_\ell$ is a type 2 matrix. 
By the construction of $\FF_{\semic(\mathbf{G})}$, $F_\ell$ is either a type 1 matrix or a type 2 matrix.




We prove some useful properties for our construction of $\FF_{\semic(\mathbf{G})}$.

\begin{lemma}\label{lem:property_F_single}
We have the following properties for $\FF_{\semic(\mathbf{G})} = (F_1, \dots, F_t)$:
\begin{enumerate}
\item $t$ is upper bounded by a polynomial of $n$ and $m$.
\item For $1 \leq \ell \leq t_5$, $F_\ell$ is a type 1 matrix, where $t_5$ is defined in the construction of $\FF_{\semic(\mathbf{G})}$.
\item For $t_5 + 1 \leq \ell \leq t$, $F_\ell$ is a type 2 matrix.
\item For every  non-zero row vector $v \in \mathbb{F}_p^{3 + n + m'}$ such that $v[k] = 0$ for $k = 1, 2, 3$, 
there is an $1 \leq \ell\leq t$ such that $v F_{\ell}$ is a non-zero vector, and the first three rows of $F_{\ell}$ are all zero.
\item The linear span of the rows of $B_\ell$ for all the $1 \leq \ell \leq t$ is a row vector space of dimension $m'$.
\end{enumerate}
\end{lemma}
\begin{proof}
The first three properties are by the construction of $\FF_{\semic(\mathbf{G})}$.
For the fourth property, let $x \in \mathbb{F}_p^{3 + n + m'}$ be a non-zero row vector such that 
$x[k] = 0$ for all the $k = 1, 2, 3$, and $k > 3 + n$.
By the construction of $\FF_{\semic(\mathbf{G})}$, the row vectors of $x F_{\ell}$ for all the $4 \leq \ell \leq t_4$ and $t_5 + 1 \leq \ell \leq t_7$ encode the matrix $\sum_{j = 1}^{n}x[j + 3]\cdot \matrixspaceY_{\mathbf{G}, j}$.
By Fact~\ref{fact:basic_tensor},
the matrix is a non-zero matrix. 
Hence, $x F_{\ell}$ is a non-zero vector for some $4 \leq \ell \leq t_4$ or $t_5 + 1 \leq \ell \leq t_7$.

Similarly, let $y$ be a row vector such that 
$y[k] = 0$ for all the $k \leq 3 + n$.
The row vectors of $y F_{\ell}$ for all the $t_5 + 1 \leq \ell \leq t_7$ encode the matrix $\sum_{i = 1}^{m'} y[i + 3 + n] \cdot \matrixspaceX_{\mathbf{G}, i}$.
By Fact~\ref{fact:basic_tensor}, $\sum_{i = 1}^{m'} y[i + 3 + n] \matrixspaceX_{\mathbf{G}, i}$ is a non-zero matrix,  
$y F_{\ell}$ is a non-zero vector for some $t_5 + 1 \leq \ell \leq t_7$.

For arbitrary row vectors $x$ and $y$ defined above, 
by the second and third properties,
at most one of $(xF_\ell)[k]$ and $(yF_\ell)[k]$ is non-zero for any $1 \leq \ell \leq t$ and $1 \leq k \leq 3 + n +m'$.
Hence, for any row vector $v \in \mathbb{F}_p^{3 + n + m'}$ such that $v[k] = 0$ for $k = 1, 2, 3$, 
there is a non-zero vector $v F_\ell$ for some $4 \leq \ell \leq t_4$ or $t_5 + 1 \leq \ell \leq t_7$. 
Since the first three rows of $F_\ell$ are zero rows  for all the $4 \leq \ell \leq t_4$ and $t_5 + 1 \leq \ell \leq t_7$, the fourth property holds.

For the last property, if linear span by the rows of $B_\ell$ for all the $1 \leq \ell \leq t$ is of dimension smaller than $m'$, then there is a non-zero vector 
$(t_1, \dots, t_{m'}) \in \mathbb{F}_p^{m'}$ such that 
$\sum_{i=1}^{m'} t_i \cdot \matrixspaceX_{\semic(\mathbf{G}), i}$ is a zero matrix, which contradicts to Fact~\ref{fact:basic_tensor}. 
Hence, the last property holds.
\end{proof}

\subsection{Reduction to the restricted skew-symmetric matrix tuple isometry}

In this section, we show that two semi-canonical forms $\semic(\mathbf{G})$ and $\semic(\mathbf{H})$ are isometric if and only if there is a block diagonal matrix $S$ 
such that $S\FF_{\semic(\mathbf{G})} S^T = \FF_{\semic(\mathbf{H})}$. 

\begin{lemma}\label{lem:restricted_tuple_isometry}
Let $\semic(\mathbf{G})$ and $\semic(\mathbf{H})$ be the semi-canonical forms of two tensors with the same parameters $\alpha_\matrixspaceX, \beta_\matrixspaceX, \alpha_\matrixspaceY$, and $\beta_\matrixspaceY$.
Then $\semic(\mathbf{G})$ and $\semic(\mathbf{H})$  are isometric if and only if 
there is a matrix $S$ of form 
\begin{equation}\label{equ:restricted_tuple_isometry}S = \left(\begin{array}{cc} Q & 0 \\ 0 & W\end{array}\right)\end{equation} 
such that $S \cdot \FF_{\semic(\mathbf{G})} \cdot S^T = \FF_{\semic(\mathbf{H})}$, where $Q$ is a $(3+n) \times (3+n)$ matrix and $W$ is an $m' \times m'$ matrix.
\end{lemma}

We first prove some useful properties for the case that there is a matrix $S$ such that $S \cdot \FF_{\semic(\mathbf{G})} \cdot S^T = \FF_{\semic(\mathbf{H})}$.

\begin{lemma}\label{lem:property_restricted_isometry}
Let $\semic(\mathbf{G})$ and $\semic(\mathbf{H})$ be semi-canonical forms for two tensors with the same parameters $\alpha_\matrixspaceX, \beta_\matrixspaceX, \alpha_\matrixspaceY$, and $\beta_\matrixspaceY$. 
Let $\FF_{\semic(\mathbf{G})} = (F_1, \dots, F_t)$ and $\FF_{\semic(\mathbf{H})} = (F_1', \dots, F_t')$ be the skew-symmetric matrix tuples for $\semic(\mathbf{G})$ and $\semic(\mathbf{H})$ respectively. 
If there is a matrix $S \in M(3 + n + m', \mathbb{F}_p)$ such that $S \cdot \FF_{\semic(\mathbf{G})} \cdot S^T = \FF_{\semic(\mathbf{H})}$, then 
$S$ satisfies the following properties:
\begin{enumerate}
\item $S$ is a full rank matrix.
\item For any $1 \leq \ell \leq t$, if $F_\ell$ is a type 1 matrix, then $F_\ell'$ is a type 1 matrix. 
If $F_\ell$ is a type 2 matrix, then $F_\ell'$ is a type 2 matrix. 
\item 
Let $\Phi$ be the set 
\[\{1, 2, 3\} \cup \{3 + \alpha_\matrixspaceY + 1, \dots, 3 + n\} \cup \{4 + n + \alpha_\matrixspaceX, \dots, 3 + n + m'\}.\]

There is a $\gamma \in \mathbb{F}_p$ satisfying $\gamma^2 = 1$ such that the following conditions hold:
\begin{enumerate}
\item $S[k, k] = \gamma$
for any $k \in \Phi$.
\item $S[i, k ] = 0$ for any $k \in \Phi$ and $i \neq k$.
\end{enumerate}
\item $S[1, 3;4, 3 + n + m']$ is a zero matrix.
\item For each row $v$ of $S[3 + n + 1, 3 + n + m';1, 3 + n]$, 
$v A_\ell$ is a zero row vector for all the $1 \leq \ell \leq t$. 

\item $S[4, 3 + n, 4, 3+n]$ is of form 
\begin{equation}\label{equ:s_form_pre}S[4, 3 + n, 4, 3+n] =  \left(\begin{array}{ccc} 
A & 0 & 0 \\ 
B & \gamma\cdot I_{\beta_{\matrixspaceY}}& 0 \\ 
C & 0 & \gamma \cdot I_{n - n'}\end{array}\right)\end{equation}
for some $A \in M(\alpha_{\matrixspaceY}, \alpha_{\matrixspaceY}, \mathbb{F}_p)$, 
$B \in M(\beta_{\matrixspaceY}, \alpha_{\matrixspaceY}, \mathbb{F}_p)$, 
and $C \in M(n - n', \alpha_{\matrixspaceY}, \mathbb{F}_p)$ satisfying the following conditions
\begin{enumerate}
\item  $C \cdot (\matrixspaceX_{\semic(\mathbf{G}), i}[1, \alpha_{\matrixspaceY}; 1, n'])$ is a zero matrix for each $1 \leq i \leq m$.
\item  $C \cdot (\matrixspaceX_{\semic(\mathbf{G}), i}[1, \alpha_{\matrixspaceY}; n' + 1, n])$ is a zero matrix for each $1 \leq i \leq m'$.
\item  $C \cdot (\matrixspaceX_{\semic(\mathbf{G}), i}[1, \alpha_{\matrixspaceY}; n' + 1, n]) + 
(\matrixspaceX_{\semic(\mathbf{G}), i}[1, \alpha_{\matrixspaceY}; n' + 1, n])^T \cdot C^T$ is a zero matrix for each $m'+1 \leq i \leq m$.
\end{enumerate}
\end{enumerate}
\end{lemma}
\begin{proof}
To prove the first property, we show that for each row vector $v \in \mathbb{F}_p^{3 + n + m'}$, there is a matrix $F_\ell \in \FF_{\semic(\mathbf{G})}$ such that
$v F_\ell$ is a non-zero vector.
By the construction of $F_1, F_2$ and $F_3$, if at least one of $v[1]$, $v[2]$ and $v[3]$ is not equal to zero, 	then at least one of $vF_1$, $vF_2$ and $vF_3$ is a non-zero vector. 
For the case of $v[1] = v[2] = v[3] = 0$, by the third property of Lemma~\ref{lem:property_F_single}, if $v$ is a non-zero vector, 
$v F_\ell$ is a non-zero vector for some $F_\ell \in \FF_{\semic(\mathbf{G})}$. 
Hence, the first property of the lemma holds.

The second property is obtained by  $S \cdot \FF_{\semic(\mathbf{G})} \cdot S^T = \FF_{\semic(\mathbf{H})}$.

For the third property of the lemma,
by the definition of $\FF_{\semic(\mathbf{G})}$, if $S[1, 2] \neq 0$, then the first row of $S F_3$ is not a zero matrix. Since $S$ is invertible, the first row of $S F_3 S^T$ is a non-zero row.
However, the first row of $F_3'$ is a zero row, contradiction. Hence, $S[1, 2] = 0$. Similarly, we have \[S[1, 3] = S[2, 1] = S[2, 3] = S[3,1] = S[3,2] = 0.\]
Since $F_1$ are all zero from the $4$-th row to the $(3 + n + m')$-th row, we have
\begin{align*} \left(SF_1S^T\right)[1, 3; 1,3] = & \left(\begin{array}{ccc} S[1, 1] & 0 & 0 \\ 0 & S[2,2] & 0 \\ 0 & 0 & S[3,3] \end{array}\right) \left(\begin{array}{ccc} 0 & 1 & 0 \\ -1 & 0 & 0 \\ 0 & 0 & 0 \end{array} \right)  \left(\begin{array}{ccc} S[1, 1] & 0 & 0 \\ 0 & S[2,2] & 0 \\ 0 & 0 & S[3,3] \end{array}\right)\\
= & \left(\begin{array}{ccc} 0 & S[1,1] \cdot S[2,2]& 0 \\ - S[1,1]\cdot S[2,2] & 0 & 0 \\ 0 & 0 & 0 \end{array}\right) \\
= & F_1'[1, 3;1, 3]. \end{align*}
Hence, $S[1,1] \cdot S[2,2] = 1$.
Similarly, we have $S[1, 1] \cdot S[3, 3] = S[2,2] \cdot S[3, 3] = 1$. 
Because $S[1, 1], S[2, 2], S[3,3] \in \mathbb{F}_p$ for some prime $p$, 
$S[1, 1] = S[2, 2] = S[3, 3]$, and thus 
$S[1, 1]^2 = 1$. 
So we have $S[1, 1] = \gamma$ for some $\gamma^2 = 1$, and consequently
\[S[1, 3;1,3] = \left(\begin{array}{ccc} \gamma & 0 & 0 \\ 0 & \gamma & 0 \\ 0 & 0 & \gamma \end{array}\right).\]

We prove $S[j, k] = 0$ for all the $j > 3$ and $k \in \{1, 2, 3\}$ by contraction. 
If $S[j, k] \neq 0$ for some $j > 3$ and $k \in \{1, 2, 3\}$, 
the $j$-th row of $S F_{i} S^T$ is not a zero row for some $i \in \{1, 2, 3\}$, which contradicts to the fact that the $j$-th row of $F_i'$ is a zero row for all the $1 \leq i \leq 3$ and $j > 3$.
Hence, $S[j, k] = 0$ for all the $j > 3$ and $k \in \{1, 2, 3\}$.

Now we prove the third property for $\Phi \setminus \{1, 2, 3\}$.
Let $k$ be an arbitrary number in $\Phi \setminus \{1, 2, 3\}$. 
Since there is an $\ell \in \{1, \dots, t\}$ such that $F_\ell[1, k] = F_\ell'[1, k] = 1$, $F_\ell[k, 1] = F_\ell'[k, 1] = -1$, and all the other entries of $F_\ell$ and $F_\ell'$ are zero.
Hence, for all the $j \neq 1$ and $j \neq k$, $S[j, k] = 0$.
Since $F_{\ell + 1}[2, k] = F_{\ell + 1}'[2, k] = 1$, $F_{\ell + 1}[k, 2] = F_{\ell + 1}'[k, 2] = -1$, and all the other entries of $F_{\ell + 1}$ and $F_{\ell + 1}'$ are zero,
$S[1, k] = 0$.
Thus, we have 
\begin{align*}(S\cdot F_\ell \cdot S^T)[1, k] = & \sum_{i = 1}^{3+n+m'}  \sum_{j = 1}^{3+n+m'}  S[1, i]\cdot F_\ell [i, j] \cdot S^T[j, k] \\
= & S[1, 1] \cdot S^T[k, k] - S[1, k] \cdot S^T[1, k] \\ = & S[1, 1] \cdot S^T[k, k]\\ = & 1,\end{align*}
where the third inequality uses the fact that $S[1,k] = 0$ for all the $k > 3$.
Thus, $S[k, k] = \gamma$.
Then the third property of the lemma holds.

For the fourth property, 
if $S[1, 3; 4, 3 + n + m']$ is a non-zero matrix, then 
by Lemma~\ref{lem:property_F_single}, 
there exists an $\ell \in \{1, \dots, t\}$ such that the first three rows of  $F_\ell$ and $F_\ell'$ are zero rows,
but $S F_\ell S^T$ has a non-zero row in one of the first three rows. This contradicts to $S F_\ell S^T = F_\ell'$.
Hence, the fourth property holds.

For the fifth property, since $S \cdot F_\ell \cdot S^T = F_\ell'$ for each type 1 matrix $F_\ell$ and $F_\ell'$, and $S^T$ is an invertible matrix, 
the last $m'$ rows of $S\cdot F_\ell$  are zero rows.
Since the last $m'$ rows of $F_\ell$ are zero rows, for each row $v$ of $S[3 + n + 1, 3 + n + m';1, 3 + n]$, 
$v A_\ell$ is a zero row vector.

Now we prove the last property of the lemma. Equation~(\ref{equ:s_form_pre}) is obtained by the third property of the current lemma.
For each $t_1 + 1 \leq \ell \leq t_2$,
$F_\ell$ and $F_\ell'$ have non-zero entries only in the submatrices $F_\ell[4, 3 +n'; 4, 3 +n']$
and $F_\ell'[4, 3 +n'; 4, 3 +n']$, respectively.
Hence, $C \cdot F_\ell [4, 3 +\alpha_\matrixspaceY; 4, 3 + n']$ must be a zero matrix because $S F_\ell S^T[4 + n', 3 + n ; 1, 3 + n + m']$ is a zero matrix. 
Since $F_\ell[4, 3 + \alpha_\matrixspaceY;4, n'+3]$ equals \[\matrixspaceX_{\semic(\mathbf{G}), \alpha_\matrixspaceX+ (\ell -t_1)}[1, \alpha_\matrixspaceY;1,n']\] for each $t_1 + 1 \leq \ell \leq t_2$,
the property 6(a) holds. 

For each $t_5 + 1 \leq \ell \leq t_6$,
the rows between the $4$-th row and the $(3 + n')$-th row of $F_\ell$ and $F_\ell'$ are non-zero only in the last $m'$ columns, 
and the rows between the $(4+n')$-th row and the $(3 + n)$-th row of $F_\ell$ and $F_\ell'$ are zero rows.
Since $S F_\ell S^T = F_\ell'$,
$C \cdot S[4, 3 + \alpha_\matrixspaceY; 4+n, 3 + n + m']$ is a zero matrix. 
Since 
$F_\ell[4, n' + 3; 4 + n, 3 + n + m']$ equals $-(\matrixspaceY_{\semic(\mathbf{G}), \ell - t_5 + \alpha_\matrixspaceY}[1, m'; 1, n'])^T$, 
we have that 
$C \cdot (\matrixspaceY_{\semic(\mathbf{G}), q}[1, m'; 1, \alpha_\matrixspaceX])^T$ is a zero matrix for all the $n' + 1\leq q \leq n$, and thus 
$C \cdot (\matrixspaceY_{\semic(\mathbf{G}), q}[r, r; 1, \alpha_\matrixspaceX])^T$ is a zero column vector for each $n' + 1 \leq q \leq n$ and $1 \leq r \leq m'$.
Since $\matrixspaceX_{\semic(\mathbf{G}), i}[j, k] = \matrixspaceY_{\semic(\mathbf{G}), j}[i, k]$, 
$C \cdot (\matrixspaceX_{\semic(\mathbf{G}), r}[q, q; 1, \alpha_\matrixspaceY])^T$ is a zero column vector for each $n' + 1 \leq q \leq n$ and $1 \leq r \leq m'$.
Using the fact that $\matrixspaceX_{\semic(\mathbf{G})}$ is a skew-symmetric matrix space,
$C \cdot \matrixspaceX_{\semic(\mathbf{G}), r}[1, \alpha_\matrixspaceY; q, q]$ is a zero column vector for each $n' + 1 \leq q \leq n$ and $1 \leq r \leq m'$.
Then the property 6(b) holds.

To prove the property 6(c), we consider $F_\ell$ and $F_\ell'$ for $t_3 + 1 \leq \ell \leq t_4$. 
Since $F_\ell$ is non-zero only in the submatrices $F_\ell[4, n' + 3; n' + 4, n + 3]$ and  $F_\ell[n' + 4, n + 3; 4, n' + 3]$  for $t_3 + 1 \leq \ell \leq t_4$, 
by Equation~(\ref{equ:s_form_pre}),
we have 
\[(S F_\ell)[4 + n', 3 + n; 4 + n', 3 + n] = C \cdot F_\ell[4, 3 + \alpha_\matrixspaceY;4 + n', n + 3]\]
and 
\[(S F_\ell)[4 + n', 3 + n; 4, 3 + n'] = F_\ell[4 + n', 3 + n; 4, 3 + n'].\]
All the other entries of $(S F_\ell)[4 + n', 3 + n; 1, 3 + n + m'] $ are zero.
Hence, by Equation~(\ref{equ:s_form_pre}),
\begin{align*}& (S F_\ell S^T)[4 + n', 3 + n; 4 + n', 3 + n] \\ = & C \cdot F_\ell[4, 3 + \alpha_\matrixspaceY;4 + n', n + 3] + (F_\ell[4, 3 + \alpha_\matrixspaceY;4 + n', n + 3])^T \cdot C^T \\
=  & F_\ell'[4 + n', 3 + n; 4 + n', 3 + n]
\end{align*}
is a zero matrix by the construction of $\FF_{\semic(\mathbf{G})}$.
Since for each $t_3 + 1 \leq \ell \leq t_4$, $F_\ell[4, 3 + \alpha_\matrixspaceY;4 + n', n + 3]$ 
corresponds to $\matrixspaceX_{\semic(\mathbf{G}), \ell - t_3 + m'}[1 + \alpha_\matrixspaceY; n' + 1, n]$.
The property 4(c) holds.
\end{proof}

\begin{proof}[Proof of Lemma~\ref{lem:restricted_tuple_isometry}]
If two tensor semi-canonical forms are isometric, then by Definition~\ref{def:skew_tensor_semi-canonical_form_iso},
there exist two matrices $M$ and $N$ satisfying Equation~(\ref{equ:skew_tensor_semi-canonical_form_iso})
such that $\trans_{N, M}(\semic(\mathbf{G})) = \semic(\mathbf{H})$.
By Lemma~\ref{lem:property_semi_canonical_form_tensor} and Equation~(\ref{equ:skew_tensor_semi-canonical_form_iso}), 
$\semic(\mathbf{G})[i, j, k] = \semic(\mathbf{H})[i, j, k]$  for any $1 \leq i \leq m'$ and $1 \leq j, k \leq n'$. 
In addition, let
\[S' = \left( \begin{array}{ccc} I_3 & 0 & 0 \\ 0 & N & 0 \\ 0 & 0 & M' \end{array}\right),\]
where $M'$ is obtained by removing the last $m - m'$ rows and the last $m - m'$ columns of $M$.
By our construction of $\FF_{\semic(\mathbf{G})}$ and  $\FF_{\semic(\mathbf{H})}$, we have $S' \cdot \FF_{\semic(\mathbf{G})} \cdot S'^T = \FF_{\semic(\mathbf{H})}$.

Now we show that $\semic(\mathbf{G})$ and $\semic(\mathbf{H})$ are isometric if there is a matrix $S$ satisfying the form of Equation~(\ref{equ:restricted_tuple_isometry})
such that $S \cdot \FF_{\semic(\mathbf{G})} \cdot S^T = \FF_{\semic(\mathbf{H})}$.
By Lemma~\ref{lem:property_restricted_isometry}, we have 
\[Q = \gamma \cdot \left(\begin{array}{cccc} I_3 & 0 & 0 & 0 \\ 0 & A & 0 & 0 \\ 0 &  B & I_{\beta_{\matrixspaceY}}& 0 \\ 0 & C & 0 & I_{n - n'}\end{array}\right)\]
for some some $\gamma \in \mathbb{F}_p$ satisfying $\gamma^2 = 1$, $A \in M(\alpha_{\matrixspaceY}, \alpha_{\matrixspaceY}, \mathbb{F}_p)$, 
$B \in M(\beta_{\matrixspaceY}, \alpha_{\matrixspaceY}, \mathbb{F}_p)$,  $C \in M(n - n', \alpha_{\matrixspaceY}, \mathbb{F}_p)$,
and 
\[W = \gamma \cdot \left(\begin{array}{cc} D & 0 \\ E & I_{\beta_{\matrixspaceX}} \end{array}\right)\]
for some $D \in M(\alpha_{\matrixspaceX}, \alpha_{\matrixspaceX}, \mathbb{F}_p)$, 
$E \in M(\beta_{\matrixspaceX}, \alpha_{\matrixspaceX}, \mathbb{F}_p)$.
Let 
\[N =  \left(\begin{array}{ccc} A & 0 & 0 \\ B & I_{\beta_{\matrixspaceY}}& 0 \\ 0 & 0 & I_{n - n'}\end{array}\right)\text{ and }
M = \left(\begin{array}{ccc} D & 0  & 0 \\ E & I_{\beta_{\matrixspaceX}} &  0 \\ 0 & 0 & I_{m - m'}\end{array}\right).
\]
In the rest of this proof, we show that $\trans_{N, M}(\semic(\mathbf{G})) = \semic(\mathbf{H})$.
Let \[N' =  \left(\begin{array}{ccc} A & 0 & 0 \\ B & I_{\beta_{\matrixspaceX}}& 0 \\ C & 0 & I_{n - n'}\end{array}\right).\]
Since $N' - N$ has non-zero entries only in the submatrix of $(N' - N)[n' + 1, n ; 1, \alpha_\matrixspaceY]$, 
by the last property of Lemma~\ref{lem:property_restricted_isometry}, the construction of $\FF_{\semic(\mathbf{G})}$ and  $\FF_{\semic(\mathbf{H})}$, and the condition that $S \cdot \FF_{\semic(\mathbf{G})} \cdot S^T = \FF_{\semic(\mathbf{H})}$,
we have 
\begin{align*}N \matrixspaceX_{\semic(\mathbf{G}), i} N^T = & N' \matrixspaceX_{\semic(\mathbf{G}), i} N'^T 
+ N' \matrixspaceX_{\semic(\mathbf{G}), i} (N - N')^T  \\
& + (N - N') \matrixspaceX_{\semic(\mathbf{G}), i} N'^T 
+ (N - N') \matrixspaceX_{\semic(\mathbf{G}), i} (N - N')^T  \\
= & N' \matrixspaceX_{\semic(\mathbf{G}), i} N'^T  \\
= & \matrixspaceX_{\semic(\mathbf{H}), i}
\end{align*}
for all the $m' \leq i \leq m$.
By the definition of $M$, we have
$\trans_{N, M}(\semic(\mathbf{G}))[i, j, k] = \semic(\mathbf{H})[i, j, k]$ for all the $m' + 1 \leq i \leq m, 1 \leq j, k \leq n$.

By Lemma~\ref{lem:property_semi_canonical_form_tensor} and the construction of $\FF_{\semic(\mathbf{G})}$ and  $\FF_{\semic(\mathbf{H})}$,  for all the $1 \leq i \leq m'$ and $1 \leq j, k \leq n'$,
\[\trans_{N, M}(\semic(\mathbf{G}))[i, j, k] = \trans_{N', M}(\semic(\mathbf{G}))[i, j, k] = \semic(\mathbf{H})[i, j, k].\] 
Furthermore, by the construction of $\FF_{\semic(\mathbf{G})}$, we have \[\matrixspaceX_{\trans_{N', M}(\semic(\mathbf{G})), i}[j, k] = \matrixspaceX_{\semic(\mathbf{H}), i}[j, k] \] for all the $1 \leq i \leq m'$, $n' + 1 \leq j \leq n$ and $1 \leq k \leq n$.
By the last property of Lemma~\ref{lem:property_restricted_isometry}, 
we have 
\[\matrixspaceX_{\trans_{N, M}(\semic(\mathbf{G})), i}[j, k] = \matrixspaceX_{\trans_{N', M}(\semic(\mathbf{G})), i}[j, k]\] for all the $1 \leq i \leq m'$, $n' + 1 \leq j \leq n$ and $1 \leq k \leq n$.
Together with the skew-symmetric condition of matrices in $\matrixspaceX_{\semic(\mathbf{H})}$, 
we have $\trans_{N, M}(\semic(\mathbf{G}))[i, j, k] = \semic(\mathbf{H})[i, j, k]$ for all the $1 \leq i \leq m', 1 \leq j, k \leq n$.
Hence, $\trans_{N, M}(\semic(\mathbf{G})) = \semic(\mathbf{H})$. 
\end{proof}

\subsection{Isometry testing of tensor semi-canonical forms}
We present the algorithm for deciding whether the semi-canonical forms of two skew-symmetric matrix spaces are isometric.

Suppose we run the algorithm for skew-symmetric matrix tuple isometry on $\FF_{\semic(\mathbf{G})}$ and $\FF_{\semic(\mathbf{H})} $.
If the algorithm returns no, then by Lemma~\ref{lem:restricted_tuple_isometry}, the two semi-canonical forms are not isometric. 
If the algorithm returns yes and a block diagonal $S$, then by Lemma~\ref{lem:restricted_tuple_isometry}, 
the two semi-canonical forms are isometric. 
The difficult case is that 
the algorithm returns yes and a matrix $S$ not satisfying Equation~(\ref{equ:restricted_tuple_isometry}). 
For this case, we neither certify that the two semi-canonical forms are isometric by Lemma~\ref{lem:restricted_tuple_isometry} nor 
 rule out the possibility that the two semi-canonical forms are not isometric.

We characterize the matrix $S$ in the difficult case by Lemma~\ref{lem:reduction_easy_case} and Lemma~\ref{lem:s_transform_final}. We further show that the isometry between the two semi-canonical forms can be determined by running the matrix tuple equivalence algorithm for two matrix tuples constructed based on $\semic(\mathbf{G}), \semic(\mathbf{H})$, and $S$.

\begin{lemma}\label{lem:reduction_easy_case}
Let $\semic(\mathbf{G})$ and $\semic(\mathbf{H})$ be semi-canonical forms of two tensors in $\mathbb{F}_p^{m\times n\times n}$ with the same parameters for some prime $p > 2$ and integers $n, m$. 
Let $\FF_{\semic(\mathbf{G})} = (F_1, \dots, F_t)$ and $\FF_{\semic(\mathbf{H})} = (F_1', \dots, F_t')$ be the skew-symmetric matrix tuples for $\semic(\mathbf{G})$ and $\semic(\mathbf{H})$ respectively. 
Suppose there is a matrix $S$ such that $S\cdot \FF_{\semic(\mathbf{G})} \cdot S^T =  \FF_{\semic(\mathbf{H})}$. 
Denote $S$ as 
\[S = \left( \begin{array}{cc} Q & R \\ V & W \end{array} \right),\]
where $Q, R, V$, and $W$ are matrices of dimensions $(n + 3)\times (n + 3)$, $(n + 3) \times m'$, $m' \times (n + 3)$, and $m' \times m'$, respectively.
If at least one of $Q$ and $W$ is full rank, then there are
$(n + 3)\times (n + 3)$ matrix $Q'$ and $m'\times m'$ matrix $W'$ such that \[\left( \begin{array}{cc} Q' & 0 \\ 0 & W' \end{array} \right) \FF_{\semic(\mathbf{G})}  \left( \begin{array}{cc} Q'^T & 0 \\ 0 & W'^T \end{array} \right) =  \FF_{\semic(\mathbf{H})}. \] 
Furthermore, $Q'$ and $W'$ can be computed in time  $\mathrm{poly}(n, m, p)$.
\end{lemma}

\begin{proof}
If $Q$ is full rank, let \begin{equation}\label{equ:case1_s_1}S' = \left(\begin{array}{cc} Q & 0 \\ 0 & W - V Q^{-1}R\end{array}\right).\end{equation}
Let $F_\ell$ be a matrix in $\FF_{\semic(\mathbf{G})}$ for some $1 \leq \ell \leq t$.
If $F_\ell$ is a type 1 matrix, 
we have 
\begin{equation}\label{equ:case1}\begin{split}S F_\ell S^T =  &  \left( \begin{array}{cc}  Q & R \\ V &  W \end{array}\right) \left(\begin{array}{cc}A_\ell & 0 \\ 0 & 0 \end{array} \right) \left( \begin{array}{cc}  Q^T & V^T \\ R^T &  W^T \end{array}\right) \\
= & \left(\begin{array}{cc}Q A_\ell Q^T & Q A_\ell V^T \\ V A_\ell Q^T & VA_\ell V^T \end{array} \right)  \\
= & F_\ell' \\ = & \left(\begin{array}{cc} Q A_\ell Q^T & 0 \\ 0 & 0\end{array}\right), \end{split}\end{equation}
where the last equality uses the fact that $F_\ell'$ is a type 1 matrix by Lemma~\ref{lem:property_restricted_isometry}.
We also have 
\begin{align*}S' F_\ell S'^T =  & \left( \begin{array}{cc}  Q & 0 \\ 0 &  W - V Q^{-1}R\end{array}\right) \left(\begin{array}{cc}A_\ell & 0 \\ 0 & 0 \end{array} \right) \left( \begin{array}{cc}  Q^T & 0 \\ 0 &  W^T - R^T (Q^{-1})^T V^T\end{array}\right) \\
=  & \left( \begin{array}{cc}  QA_\ell Q^T & 0 \\ 0 &  0\end{array}\right) \\ = & F_\ell'.\end{align*}
Hence, $S'F_\ell S'^T = S F_\ell S^T$ for all the type 1 $F_\ell$.

If $F_\ell$ is a type 2 matrix, 
since 
$S F_\ell S^T = F_\ell'$ is also a type 2 matrix,
we have 
\begin{equation}\label{equ:case1_1}\begin{split}S F_\ell S^T = & S \left( \begin{array}{cc}  0 & B_\ell \\ -B_\ell^T & 0  \end{array}\right) S^T\\
= &  \left( \begin{array}{cc}  -R B_\ell^T Q^T + QB_\ell R^T = 0 & -RB_\ell^T V^T + QB_\ell W^T  \\ -WB_\ell^TQ^T + VB_\ell R^T & -WB_\ell^T V^T + VB_\ell W^T = 0 \end{array} \right).\end{split}\end{equation}
Since
$RB_\ell^T = R B_\ell^T Q^T (Q^T)^{-1} = QB_\ell R^T (Q^T)^{-1}$, we have 
\[S F_\ell S^T = \left( \begin{array}{cc}  0 & QB_\ell(W^T - R^T (Q^T)^{-1} V^T)  \\ (-W + VQ^{-1}R )^T B_\ell^T Q^T  & 0 \end{array} \right).\]
Hence,
\begin{align*}S' F_\ell S'^T =& \left(  \begin{array}{cc}   Q & 0 \\ 0 &  W - V Q^{-1}R\end{array}\right) \left( \begin{array}{cc} 0 & B_\ell \\ -B_\ell^T & 0 \end{array}\right) \left( \begin{array}{cc}  Q^T & 0 \\ 0 &  W^T - R^T (Q^{-1})^T V^T\end{array}\right) \\ = & S F_\ell S^T \\ = & F_\ell'. \end{align*}
Then the lemma holds if $Q$ is full rank. 

Now we consider the case that $W$ is full rank. Let 
\begin{equation}\label{equ:case1_s_2}S' = \left( \begin{array}{cc}  Q - RW^{-1}V & 0 \\ 0 &  W\end{array}\right) .\end{equation}
We have for each type 1 matrix $F_\ell \in \FF_{\semic(\mathbf{G})}$, 
\begin{align*}S' F_\ell S'^T = & \left( \begin{array}{cc}  Q - RW^{-1}V & 0 \\ 0 &  W\end{array}\right) \left( \begin{array}{cc}  A_\ell & 0 \\ 0 & 0 \end{array} \right)\left( \begin{array}{cc}  Q^T - V^T(W^{-1})^TR^T & 0 \\ 0 &  W^T\end{array}\right)  \\ 
= & \left( \begin{array}{cc} (Q - RW^{-1}V) A_\ell (Q - RW^{-1}V)^T & 0 \\ 0 & 0 \end{array}\right) \\
= & \left( \begin{array}{cc} Q A_\ell Q^T & 0 \\ 0 & 0 \end{array}\right) \\
= & S F_\ell S^T \\ = & F_\ell', \end{align*}
where the third equality uses the fact that $V A_\ell Q^T = 0$, $Q A_\ell V^T = 0$ and $V A_\ell V^T = 0$ by Equation~(\ref{equ:case1}). 

For each type 2 matrix $F_\ell$,
by Equation~(\ref{equ:case1_1}), we have 
\[VB_\ell = V B_\ell W^T (W^T)^{-1} = W B_\ell^T V^T(W^T)^{-1},\] and thus 
\[S F_\ell S^T = \left( \begin{array}{cc}  0 & (R W^{-1} V- Q )B_\ell W^T \\  W B_\ell^T (-V^T(W^T)^{-1}R^T+Q^T) & 0 \end{array} \right).\]
Hence,
\begin{align*}S' F_\ell S'^T = & \left(  \begin{array}{cc}   Q - RW^{-1}V  & 0 \\ 0 &  W\end{array}\right) \left( \begin{array}{cc} 0 & B_\ell \\ -B_\ell^T & 0 \end{array}\right) \left( \begin{array}{cc}  Q - RW^{-1}V & 0 \\ 0 &  W\end{array}\right)^T \\ = & S F_\ell S^T \\ = & F_\ell'. \end{align*}
Then the lemma holds if $W$ is full rank. 
By Equation~(\ref{equ:case1_s_1}) and Equation~(\ref{equ:case1_s_2}), $Q'$ and $W'$ can be computed in $\mathrm{poly}(n, m, p)$ time.
\end{proof}

\begin{lemma}\label{lem:upper_triangle}
Let $\semic(\mathbf{G})$ and $\semic(\mathbf{H})$ be semi-canonical forms of two tensors with the same parameters. 
Let $\FF_{\semic(\mathbf{G})} = (F_1, \dots, F_t)$ and $\FF_{\semic(\mathbf{H})} = (F_1', \dots, F_t')$ be the skew-symmetric matrix tuples for $\semic(\mathbf{G})$ and $\semic(\mathbf{H})$ respectively. 
If the following two conditions hold:
\begin{enumerate}
\item There is an invertible matrix $S$ such that $S\cdot \FF_{\semic(\mathbf{G})} \cdot S^T =  \FF_{\semic(\mathbf{H})}$.
\item There is a matrix 
\[P = \left(\begin{array}{cc} I_{3+n} & U \\ 0 & I_{m'} \end{array}\right)\]
for some $U \in M(3+n, m', \mathbb{F}_p)$
such that for each type 2 matrix $F_\ell \in \FF_{\semic(\mathbf{G})} $,
$PS F_\ell S^T P^T [1, 3+n;1, 3+n]$ is a zero matrix. 
\end{enumerate}
Then $PS \cdot \FF_{\semic(\mathbf{G})} \cdot S^T P^T = \FF_{\semic(\mathbf{H})}$.
\end{lemma}
\begin{proof}
For each type 1 matrix $F_\ell \in \mathbf{F}_{\semic(\mathbf{G})} $, since $SF_\ell S^T = F_\ell'$, the submatrix
$(SF_\ell S^T)[4 + n, 3 + n + m';1, 3 + n + m']$ is a zero matrix because $F_\ell'$ is also a type 1 matrix, we have 
\begin{align*} & PS F_\ell S^T P^T \\= &(PS F_\ell S^T) P^T \\ = &
\left( \left(\begin{array}{cc} I_{3+n} & U \\ 0 & I_{m'} \end{array}\right) SF_\ell S^T\right) P^T  \\
= & \left( \left(\begin{array}{cc} I_{3+n} & 0 \\ 0 & I_{m'} \end{array}\right) SF_\ell S^T\right) P^T  \\
= & S F_\ell S^T \left(\begin{array}{cc} I_{3+n} & 0 \\ U^T & I_{m'} \end{array}\right) \\
= & S F_\ell S^T \\ 
= & F_\ell',\end{align*}
where the fifth equality is obtained by the skew symmetric condition of $S F_\ell S^T $.
For each type 2 matrix $F_\ell \in \mathbf{F}_{\semic(\mathbf{G})} $, we have 
\begin{align*}  & PS F_\ell S^T P^T \\ = & \left(\begin{array}{cc} I_{3 + n} & U \\ 0 & I_{m'}\end{array}\right) S F_\ell S^T \left( \begin{array}{cc} I_{3 + n} & U\\ 0 & I_{m'}\end{array}\right)^T \\
= & S F_\ell S^T + \left(\begin{array}{cc} 0 & U \\ 0 & 0\end{array}\right) S F_\ell S^T + S F_\ell S^T\left(\begin{array}{cc} 0 & 0 \\ U^T & 0\end{array}\right)  +   \left(\begin{array}{cc} 0 & U \\ 0 & 0\end{array}\right) S F_\ell S^T  \left(\begin{array}{cc} 0 & 0 \\ U^T & 0\end{array}\right).
\end{align*}
Let $V$ be the matrix
\[\left(\begin{array}{cc} 0 & U \\ 0 & 0\end{array}\right) S F_\ell S^T + S F_\ell S^T\left(\begin{array}{cc} 0 & 0 \\ U^T & 0\end{array}\right)
+  \left(\begin{array}{cc} 0 & U \\ 0 & 0\end{array}\right) S F_\ell S^T  \left(\begin{array}{cc} 0 & 0 \\ U^T & 0\end{array}\right). \]
Since the submatrix $F_\ell[4 + n, 3 + n + m'; 4 + n, 3 + n + m']$  is  a zero matrix, $V$
is a matrix such that the last $m'$ rows are all zero, and the last $m'$ columns are all zero.
On the other hand, 
since both $SF_\ell S^T[1, 3 + n; 1, 3 + n]$ and 
$(PS F_\ell SP^T) [1, 3 + n; 1, 3 + n]$ are a zero matrices,
$V[1, 3+n;1, 3+n]$ is a zero matrix. Hence, $V$ is a zero matrix. 
So we have 
$PS F_\ell S^T P^T = S F_\ell S^T = F_\ell'$. 
\end{proof}

\begin{lemma}\label{lem:s_transform_final}
Let $\semic(\mathbf{G})$ and $\semic(\mathbf{H})$ be semi-canonical forms of two skew-symmetric matrix space tensors with the same parameters $\alpha_\matrixspaceX, \beta_\matrixspaceX, \alpha_\matrixspaceY$, and $\beta_\matrixspaceY$.
If there is a matrix $S$ such that $S\cdot \FF_{\semic(\mathbf{G})} \cdot S^T =  \FF_{\semic(\mathbf{H})}$. 
Denote $S$ as 
\[S = \left( \begin{array}{cc} Q & R \\ V & W \end{array} \right),\]
where $Q, R, V$, and $W$ are of dimensions $(n + 3)\times (n + 3)$, $(n + 3) \times m'$, $m' \times (n + 3)$, and $m' \times m'$, respectively.
If both $Q$ and $W$ are not full rank, then 
there is a matrix $J \in \GL(3+n, \mathbb{F}_p)$, a matrix $K\in \GL(m', \mathbb{F}_p)$, a positive integer $q$, and a matrix $S'$ satisfying the following conditions:
\begin{enumerate}
\item $S'$ can be represented as 
\[\left(\begin{array}{cc} Q' & 0 \\ 0 & R' \\ 0 & W' \\ V' & 0 \end{array} \right),\]
where $Q'$ is of dimension $q \times (3 + n)$, $R'$ is of dimension $(3 + n - q) \times m'$, $W'$ is of dimension $(m' - (3 + n - q)) \times m'$, and $V'$ is of dimension $(3 + n - q) \times (3 + n)$.
\item For each type 1 matrix $F_\ell$ in $\FF_{\semic(\mathbf{G})}$,
\[ S' F_\ell S'^T =  \left( \begin{array}{cc} Q' A_\ell Q'^T & 0 \\ 0 & 0 \end{array} \right) = \left( \begin{array}{cc} J & 0 \\ 0 & K \end{array} \right) F_\ell' \left( \begin{array}{cc} J^T& 0 \\ 0 & K^T \end{array} \right).\]
\item For each type 2 matrix $F_\ell$ in $\FF_{\semic(\mathbf{G})}$,
\[ S' F_\ell S'^T   = 
 \left( \begin{array}{cc} 0 & D_\ell \\ -D_\ell^T & 0 \end{array} \right) = \left( \begin{array}{cc} J & 0 \\ 0 & K \end{array} \right) F_\ell' \left( \begin{array}{cc} J^T& 0 \\ 0 & K^T \end{array} \right)\]
 for some $D_\ell = \left( \begin{array}{cc} D_\ell' & 0 \\ 0 & D_\ell''\end{array} \right)$
such that $D_\ell'$ is of dimension $q \times (m' - (3 + n - q))$ and  $D_\ell''$ is of dimension $(3 + n - q) \times (3 + n -q)$.
\end{enumerate}

\end{lemma}

\begin{proof}
Denote 
\[\tau_S \coloneqq \mathrm{dim}\left(\langle\{v QB_\ell R^T : v \in \mathbb{F}_p^{3+n}, 1 \leq \ell \leq t\} \rangle \right).\]
By Equation~(\ref{equ:case1_1}), $QB_\ell R^T$ is a skew-symmetric matrix for all the $1 \leq \ell \leq t$.
Hence, there is a matrix $J_0' \in \mathrm{GL}(3 + n, \mathbb{F}_p)$ such that 
for any $\tau_S + 1 \leq 1 \leq i \leq 3 + n$, 
the $i$-th row of $J_0 QB_\ell R^T$ is a zero row for all the $1 \leq \ell \leq t$.
Thus, 
the multiplication of $J_0$ and the submatrix on the first $3 + n$ rows of $S$ can be represented as 
\begin{equation}\label{equ:Q_R}
J_0' \cdot \left( \begin{array}{cc} Q & R \end{array} \right) = \left( \begin{array}{cc} Q_1 & R_1 \\ Q_0 & R_0 \end{array} \right),
\end{equation}
where $Q_1$ is of dimension $\tau_S \times (3+n)$, $R_1$ is of dimension $\tau_S \times m'$, $Q_0$ is of dimension $(3 + n - \tau_S) \times (3+n)$, and $R_0$ is of dimension $(3 + n - \tau_S) \times m'$, such that 
$J_0' Q B_\ell R^T$
is non-zero only in the first $\tau_S$ rows for all the $1 \leq \ell \leq t$.
By Equation~(\ref{equ:case1_1}), $J_0' Q B_\ell R^T J_0'^T$
is non-zero only in the submatrix on the first $\tau_S$ rows and the first $\tau_S$ columns.
Hence, $Q_0 B_\ell R^T$ and $Q B_\ell R_0^T$ are zero matrices for all the $B_\ell$.
Furthermore, each non-zero linear combination of the rows of $R_0$ is not a linear combination of the rows of $R_1$. Otherwise, it contradicts the definition of $\tau_S$.
Thus \begin{equation}\label{equ:case2_t1}\rank{R} = \rank{R_0} + \rank{R_1}.\end{equation}
On the other hand, by Lemma~\ref{lem:property_restricted_isometry}, $S$ is an invertible matrix. So we have \begin{equation}\label{equ:r_0}\rank{R_0} \geq n + 3 - \rank{Q}.\end{equation}

Furthermore, we can decompose $V$ and $W$ as follows:
\begin{enumerate}
\item 
$V = V_0 + Z_{Q, V} Q$ for some $Z_{Q, V} \in M(m', 3 + n, \mathbb{F}_p)$ such that $\rank{V_0} =  3+ n - \rank{Q}$.
\item 
$W = W_0 + Z_{R, W} R$ for some $Z_{R, W} \in M(m', m', \mathbb{F}_p)$ such that $\rank{W_0} = m' - \rank{R}$.
\end{enumerate}
Consider all the type 2 matrices $F_\ell$. 
For each type 2 matrix $F_\ell$, we have 
\begin{align*}
-RB_\ell^T V^T + QB_\ell W^T 
= & -RB_\ell ^T(V_0 + Z_{Q, V} Q)^T +  QB_\ell  (W_0 + Z_{R, W} R)^T \\ 
=  & QB_\ell  R^T ( - Z_{Q, V}^T + Z_{R, W}^T)  - RB_\ell ^T V_0^T + Q B_\ell  W_0^T \\
=  & QB_\ell  R_1^T ( - Z_{Q, V}^T + Z_{R, W}^T)  - RB_\ell ^T V_0^T + Q B_\ell  W_0^T, 
\end{align*}
where the second equality is by Equation~(\ref{equ:case1_1}), and
the third equality uses the fact that  $Q B_\ell R_0$ is a zero matrix for all the $B_\ell $. 

Since $S$ is an invertible matrix, the span of row vectors of $-RB_\ell ^T V^T + QB_\ell W^T$ for all the $B_\ell $ is of dimension $m'$ by Lemma~\ref{lem:property_F_single}.
Hence, we have 
\begin{align*} & \rank{R_1} + \rank{V_0} +\rank{W_0} \\ = & (\rank{R} - \rank{R_0}) + (3 + n - \rank{Q}) + (m' - \rank{R}) \\ \geq & m'\end{align*}
which implies $3 + n - \rank{Q} \geq \rank{R_0}$.
By Equation~(\ref{equ:r_0}), we have $\rank{R_0} = 3 + n - \rank{Q}$.
By Equation~(\ref{equ:Q_R}) and the fact that $S$ is invertible, 
there is a matrix $J_0 \in \mathrm{GL}(n + 3, \mathbb{F}_p)$ such that 
\begin{equation}\label{equ:s_form} J_0 \cdot\left( \begin{array}{cc} Q & R \end{array} \right)  
= \left(\begin{array}{cc} Q_1 & R_1 \\ Q_0 & 0 \\ 0 & R_0 \end{array}\right)\end{equation}
where $Q_1$ is of dimension $\tau_S \times (3 + n)$, $R_1$ is of dimension $\tau_S \times m'$, $Q_0$ is of dimension $(\rank{Q} - \tau_S) \times (3+n)$, and $R_0$ is of dimension $(\rank{R} - \tau_S) \times m'$.




Notice that the intersection of any two spaces among the space spanned by the row vectors of $V_0$, the space spanned by the row vectors of $W_0$, and the space spanned by the row vectors of $R_1^T ( - Z_{Q, V}^T + Z_{R, W}^T)$ only contains the zero row vector.
There is a matrix $K_0 \in \mathrm{GL}(m', \mathbb{F}_p)$ such that 
the multiplication of $K_0$ and the submatrix on $V$ and $W$ (i.e., $S[4 + n, 3 + n + m',  1, 3 + n + m']$) can be written as (by slightly abusing the notations of $W_0$ and $V_0$)
\begin{equation}\label{equ:lower_s_form} \begin{split}& K_0 \cdot \left(S[4 + n, 3 + n + m',  1, 3 + n + m'] \right) \\ = & \left(\begin{array}{cc}  Z_{Q_1, V} Q_1 + Z_{Q_0, V} Q_0 &  Z_{R_1, W} R_1 + Z_{R_0, W} R_0 \\ Z_{Q_1, V}' Q_1 + Z_{Q_0, V}' Q_0   & W_0 \\ V_0  &  Z_{R_1, V}' R_1 + Z_{R_0, V}' R_0\end{array}  \right) 
\end{split}
\end{equation}
for some $\tau_S \times \tau_S$ dimensional $Z_{Q_1, V}, Z_{R_1, W}$, 
$\tau_S\times \rank{Q_0}$ dimensional $Z_{Q_0, V}$, 
$\tau_S\times \rank{R_0}$ dimensional $Z_{R_0, W}$, 
$(m' - \rank{R})\times \tau_S$ dimensional $Z_{Q_1, V}'$,
$(m' - \rank{R}) \times \rank{Q_0}$ dimensional $Z_{Q_0, V}'$, 
$(3 + n - \rank{Q})  \times \tau_S$ dimensional $Z_{R_1, W}'$, 
and $(3 + n - \rank{Q}) \times \rank{R_0}$ dimensional $Z_{R_0, W}'$. 
In the rest of this proof, we consider three cases. 

 Case 1. $\tau_S = 0$. Since $Q_1$ and $R_1$ do not exist by the condition of $\tau_S = 0$,
\[K_0 \cdot (S[4 + n, 3 + n + m',  1, 3 + n + m']) = \left(\begin{array}{cc}   Z_{Q_0, V}' Q_0   & W_0   \\ V_0 &  Z_{R_0, W}' R_0  \end{array}  \right).\]
Since \[-K_0WB_\ell^TV^TK_0^T + K_0 V B_\ell W^T K_0^T = K_0 (-WB_\ell^TV^T + V B_\ell W^T ) K_0^T\] is a zero matrix, 
\[ - Z_{Q_0, V}' Q_0  B_\ell R_0^T Z_{R_0, W}'^T  +W_0B_\ell^T V_0^T = W_0B_\ell^T V_0^T \]
is a zero matrix based on the fact that 
$Q B_\ell R^T = 0$ for all the $B_\ell$. 
Let $q$ be the number of rows of $Q_0$. Let  
\[S' = \left(\begin{array}{cc} Q_0 & 0 \\ 0 & R_0 \\ 0 & W_0    \\ V_0  & 0 \end{array}\right),\]
$J = J_0$, and $K = K_0$. 
Let $Q' = Q_0, R' = R_0, V' = V_0$, and $W' = W_0$.

By the fifth property of Lemma~\ref{lem:property_restricted_isometry}, each row vector $v$ that is a row of $V$
satisfies $v A_\ell = 0$ for each type 1 matrix $F_\ell$, and thus the second condition of the current lemma holds. 

For each type 2 matrix $F_\ell$, by the fact that $Q_0 B_\ell R_0^T$ is a zero matrix,
$S' F_\ell S'^T$ is also a type 2 matrix. 
Furthermore, we have 
\begin{align*}& (S' F_\ell S'^T)[1, 3 + n;4+n, 3 + n + m'] \\ = & \left( \begin{array}{c} Q_0 \\ 0 \end{array} \right) B_\ell \left(\begin{array}{c} W_0 \\ 0 \end{array} \right)^T - \left( \begin{array}{c} 0 \\ R_0 \end{array} \right) B_\ell^T \left( \begin{array}{c} 0 \\ V_0  \end{array} \right)^T \\
= & J \left(S F_\ell S^T [1, 3 + n;4+n, 3 + n + m'] \right)K^T. \end{align*}
Thus, the third condition of the current lemma also holds. Thus, the current lemma holds for this case.

Case 2. $\tau_S > 0$ and at least one of $Z_{Q_1, V}$ and $Z_{R_1, W}$ is full rank. 
We show that if $Z_{Q_1, V}$ is full rank, then the current lemma holds. The case that $Z_{R_1, W}$ is full rank is similar.
Let $S^\dagger$ be the matrix of 
\[
S^\dagger = \left(\begin{array}{cc} Q^\dagger & R^\dagger \\K_0\cdot V & K_0 \cdot W \end{array} \right)\]
where \[Q^\dagger = \left(\begin{array}{c} 0 \\ Q_0 \\ 0 \end{array}\right) \text{ and } R^\dagger = \left(\begin{array}{c} R_1' \\ 0  \\ R_0 \end{array}\right) \]
with 
$R_1' = R_1 - (Z_{Q_1, V})^{-1} (Z_{R_1, W} R_1 + Z_{R_0, W}R_0)$. 
Since $Q_0 B_\ell R^T$ is a zero matrix for every type 2 matrix $F_\ell$, 
$Q^\dagger B_\ell (R^\dagger)^T$ is a zero matrix for every type 2 matrix $F_\ell$. And thus,
$(S^\dagger F_\ell (S^\dagger)^T) [1, 3 + n; 1, 3 + n]$ is a zero matrix for every type 2 matrix $F_\ell$. 

On the other hand, 
notice that 
\[
S^\dagger = \left(\begin{array}{cc} X & U \\ 0 & I_{m'}\end{array}\right) \left(\begin{array}{cc} J_0 & 0 \\ 0 & K_0 \end{array}\right) S\] 
for some $X\in \GL(3+n, \mathbb{F}_p)$ and $U \in M(3+ n, m', \mathbb{F}_p)$. 
We have 
\[\left(\begin{array}{cc} X & U \\ 0 & I_{m'}\end{array}\right) \left(\begin{array}{cc} J_0& 0 \\ 0 & K_0 \end{array}\right) 
= \left(\begin{array}{cc} XJ_0 & UK_0 \\ 0 & K_0\end{array}\right)
=
 \left(\begin{array}{cc} XJ_0& 0 \\ 0 & K_0 \end{array}\right) \left(\begin{array}{cc} I_{3 + n} & J_0^{-1} X^{-1}U K_0 \\ 0 & I_{m'}\end{array}\right)\]
 using the fact that $J_0$ is an invertible matrix.
 Let 
\[
S' = \left(\begin{array}{cc} (XJ_0)^{-1} & 0 \\ 0 & K_0^{-1}  \end{array} \right)S^\dagger  = \left(\begin{array}{cc} (XJ_0)^{-1} Q^\dagger &  (XJ_0)^{-1} R^\dagger \\ V & W  \end{array} \right)\text{ and }
P =  \left(\begin{array}{cc} I_{3 + n} & J^{-1}U K \\ 0 & I_{m'} \end{array} \right).\]
We have $S' = PS$. 
Using the fact that $(S^\dagger F_\ell (S^\dagger)^T) [1, 3 + n; 1, 3 + n]$ is a zero matrix for every type 2 matrix $F_\ell$,
$(S' F_\ell S'^T) [1, 3 + n; 1, 3 + n]$ is a zero matrix for every type 2 matrix $F_\ell$.
Hence, $PS F_\ell S^TP^T [1, 3 + n; 1, 3 + n]$ is a zero matrix for every type 2 matrix $F_\ell$. 
By Lemma~\ref{lem:upper_triangle}, $ S' F_\ell S'^T  =  S F_\ell S^T $ for all the $F_\ell$ in $\FF_{\semic(\mathbf{G})}$.
In addition, since $Q^\dagger B_\ell (R^\dagger)^T$ is a zero matrix for each $B_\ell$, 
we have 
$\tau_{S'} = 0$, where 
\[\tau_{S'} \coloneqq \mathrm{dim}\left(\left\langle\left\{v (XJ_0)^{-1}Q^\dagger B_\ell ((XJ_0)^{-1}R^\dagger)^T : v \in \mathbb{F}_p^{3+n}, 1 \leq \ell \leq t\right\} \right\rangle \right).\]
By Case 1, the current lemma holds for this case.

Case 3. Both $Z_{Q_1, V}$ and $Z_{R_1, W}$ are not full rank. 
Let $S''$ be the matrix of 
\[
\left(\begin{array}{cc} Q'' & R'' \\ V & W \end{array} \right)\]
where \[Q'' = \left(\begin{array}{c} Q_1'' \\ Q_0 \\ 0 \end{array}\right) \text{ and } R'' = \left(\begin{array}{c} R_1'' \\ 0  \\ R_0 \end{array}\right) \]
with $Q_1'' = Z_{R_1, W} Q_1$ and $R_1'' = (2Z_{R_1, W} - Z_{Q_1, V}) R_1$. 
We prove some useful properties of $S''$
\begin{enumerate}
\item[(a).] $S''$ is a full rank matrix, and there is a full rank matrix $P'' =  \left(\begin{array}{cc} Z'' & U'' \\ 0 & I_{m'} \end{array} \right)$ such that $S'' = P''S$.
\item[(b).] $S'' F_\ell S''^T[1, 3+n;1, 3+n]$ is a zero matrix for each type 2 matrix $F_\ell$. 
\item[(c).] $\tau_{S''} < \tau_S$, where 
\[\tau_{S'} \coloneqq \mathrm{dim}\left(\left\langle\left\{v Q'' B_\ell (R'')^T : v \in \mathbb{F}_p^{3+n}, 1 \leq \ell \leq t\right\} \right\rangle \right).\]
\end{enumerate}
To prove the property (a), 
we first show that $Z_{Q_1, V} - Z_{R_1, W}$ is a full rank matrix, i.e., \[\rank{Z_{Q_1, V} - Z_{R_1, W}} = \tau_S.\]
Suppose it is not. Then there is a row vector $v \in \mathbb{F}_p^{\tau_S}$ such that $v (Z_{Q_1, V}- Z_{R_1, W})$ is a zero row vector, which means that 
there is a non-zero linear combination of the first $3+n$ rows of $S$ equals  
a non-zero linear combination of the last $m'$ rows of $S$, which contradicts to the fact that $S$ is a full rank matrix. 
Hence, $Z_{Q_1, V} - Z_{R_1, W}$ is a full rank matrix. 
Since there is an invertible matrix $Z'$ such that 
\[Z'\cdot S[1, 3+n; 1, 3+m] = \left(\begin{array}{cc} (Z_{R_1, W} - Z_{Q_1, V} )Q_1 & (Z_{R_1, W} - Z_{Q_1, V} ) R_1 \\ Q_0 & 0 \\ 0 & R_0 \end{array}\right),\]
the property (a) holds.


To prove the property (b), we show that for each type 2  matrix $F_\ell$,
$Q'' B_\ell R''^T - R'' B_\ell^T Q''^T$ is a zero matrix. We have
\begin{align*}Q_1'' B_\ell R_1''^T = & Z_{R_1, W} Q_1 B_\ell R_1^T (2 Z_{R_1, W} - Z_{Q_1, V})^T \\
= & 2 Z_{R_1, W} Q_1 B_\ell R_1^T Z_{R_1, W} ^T - Z_{R_1, W} Q_1 B_\ell R_1^T Z_{Q_1, V}^T
\end{align*}
$Z_{R_1, W} Q_1 B_\ell R_1^T Z_{R_1, W} ^T$ is a skew-symmetric matrix by the fact that $Q_1 B_\ell R_1^T$ is a skew-symmetric matrix for all the $1 \leq \ell \leq t$.
$Z_{R_1, W} Q_1 B_\ell R_1^T Z_{Q_1, V}^T$ is also a skew-symmetric matrix by the fact that $W B_\ell V^T$ is a zero matrix and Equation~(\ref{equ:lower_s_form}).
Hence, $Q'' B_\ell R''^T - R'' B_\ell^T Q''^T$ is a zero matrix. 


For the property (c), since $Z_{R_1, W}$ is not full rank, we have $\tau_{S''} < \tau_S$.


Let \[S' = \left(\begin{array}{cc} (Z'')^{-1}  & 0 \\ 0 & I_{m'}\end{array}\right) S'' =  \left(\begin{array}{cc} (Z'')^{-1}  & 0 \\ 0 & I_{m'}\end{array}\right)\left(\begin{array}{cc}  Z'' & U'' \\ 0 & I_{m'}\end{array}\right)S= \left(\begin{array}{cc}  I_{3+n} & (Z'')^{-1}U'' \\ 0 & I_{m'}\end{array}\right)S.\]
Based on the properties of $S''$,
$S'' F_\ell S''^T[1, 3+n;1, 3+n]$ is a zero matrix for each type 2 matrix $F_\ell$, and $\tau_{S'} = \tau_{S''} < \tau_S$.
By Lemma~\ref{lem:upper_triangle}, we have $S' F_\ell S'^T = F_\ell'$ for each $F_\ell$. 
Repeating the process for at most $3 + n$ times, 
we obtain a matrix of either Case 1 or Case 2. 
Then the current lemma follows. 
\end{proof}




The following lemma was proved in the proof of Lemma 2.2 in \cite{futorny2019wildness}.
\begin{lemma}\label{lem:block_matrix_tuple}
Let $\AA = (A_1, \dots A_k)$ and $\BB = (B_1, \dots, B_k)$ be two matrix tuples in $M(m, n, \mathbb{F}_{p})^k$. 
Suppose there are $1 \leq q < m$ and $1 \leq r < n$ such that each $A_i$ equals   
\[\left(\begin{array}{cc} A_i' & 0
\\ 0
& A_i'' \end{array}\right)\]
for some $A_i'\in M(q, r, \mathbb{F}_p)$ and $A_i'' \in M(m - q, n - r, \mathbb{F}_p)$,
and each $B_i$ equals 
\[\left(\begin{array}{cc} A_i' & 0
\\ 0
& B_i'' \end{array}\right)\]
for some $B_i'' \in M(m - q, n - r, \mathbb{F}_p)$. 
There are $P \in \mathrm{GL}(m, \mathbb{F}_p)$ and $Q \in \mathrm{GL}(n, \mathbb{F}_p)$ such that $PA_i Q = B_i$ for each $1 \leq i\leq k$
if and only if there are 
$P'' \in \mathrm{GL}(m - q, \mathbb{F}_p)$ and $Q'' \in \mathrm{GL}(n - r, \mathbb{F}_p)$ such that
$P'' A_i'' Q'' = B_i''$ for all the $1 \leq i \leq k$.
\end{lemma}

Now we give our algorithm for isometry testing of semi-canonical forms of two skew-symmetric matrix space tensors.

\begin{framed}
\noindent \textbf{Isometry Testing of Tensor Semi-Canonical Forms Algorithm}

\noindent \textbf{Input:} Semi-canonical forms $\semic(\mathbf{G})$ and $\semic(\mathbf{H})$ of two skew-symmetric matrix space tensors.

\noindent \textbf{Output:} Yes or no.

\begin{enumerate}
\item Return no if the parameters of the two semi-canonical forms are different.
Otherwise, let $\alpha_\matrixspaceX, \beta_\matrixspaceX, \alpha_\matrixspaceY$, and $\beta_\matrixspaceY$ be the parameters of the two semi-canonical forms. 
\item Construct $\FF_{\semic(\mathbf{G})} = (F_1, \dots, F_t)$ and $\FF_{\semic(\mathbf{H})} = (F_1', \dots, F_t')$.
\item Run the skew-symmetric matrix tuple isometry algorithm on $\FF_{\semic(\mathbf{G})}$ and $\FF_{\semic(\mathbf{H})}$. If the algorithm returns no, then return no. Otherwise, the algorithm returns a matrix $S$ of form \[\left(\begin{array}{cc} Q & R \\ V & W \end{array} \right)\]
for some $P, Q, R$, and $S$ of dimensions $(n + 3)\times (n + 3)$, $(n + 3) \times m'$, $m' \times (n + 3)$, and $m' \times m'$ respectively.

\item If either $Q$ or $W$ is a full rank matrix, return yes. 
\item Construct matrices $J \in \GL(3+n, \mathbb{F}_p), K \in \GL(m', \mathbb{F}_p)$ and matrix $S' \in \GL(3+n + m', \mathbb{F}_p)$ of form \[S' = \left(\begin{array}{cc} Q' & 0 \\ 0 & R' \\ 0 & W' \\ V' & 0\end{array}\right)\] 
for some positive integer $q$,  $Q'$ of dimension $q \times (3 + n)$, $V'$ of dimension $(3 + n - q) \times (3 + n)$, 
$R'$ of dimension $(3 + n - q) \times m'$, and $W'$ of dimension $(m' - 3 - n + q) \times m'$
such that 
Lemma~\ref{lem:s_transform_final} is satisfied.
\item Return the output of the matrix tuple equivalence algorithm with matrix tuples  
$(V' B_1 R'^T, \dots, V' B_t R'^T)$ and 
$(R' B_1^T V'^T, \dots, R' B_t^T V'^T)$.
\end{enumerate}

\end{framed}

\begin{lemma}\label{lem:test_semi_canonical_form}
There is an algorithm for the isometry testing of the semi-canonical forms of two skew-symmetric matrix space tensors in $\mathbb{F}_p^{m \times n \times n}$ for some prime $p > 2$ and positive integers $n, m$
with running time $\mathrm{poly}(n, m, p)$. 
\end{lemma}
\begin{proof}
We first prove the correctness of the algorithm. 
By Definition~\ref{def:skew_tensor_semi-canonical_form_iso} and Lemma~\ref{lem:restricted_tuple_isometry}, 
the two semi-canonical forms are isometric if and only if the following two conditions hold
\begin{enumerate}
\item[(a)] The two semi-canonical forms have the same parameters  $\alpha_\matrixspaceX, \beta_\matrixspaceX, \alpha_\matrixspaceY$, and  $\beta_\matrixspaceY$.
\item[(b)] 
There is a matrix $S_0$ of form 
\[\left(\begin{array}{cc} Q_0 & 0 \\ 0 & W_0\end{array}\right)\] such that $S_0 \cdot \FF_{\semic(\mathbf{G})} \cdot S_0^T = \FF_{\semic(\mathbf{H})}$, where $Q_0$ is a $(3+n) \times (3+n)$ matrix and $W_0$ is an $m' \times m'$ matrix.
\end{enumerate}

If the two input semi-canonical forms are isometric, then the first step of the algorithm does not return no by Definition~\ref{def:skew_tensor_semi-canonical_form_iso}. 
Hence, step 4 of the algorithm returns yes and a matrix $S$ of form \[\left(\begin{array}{cc} Q & R \\ V & W \end{array} \right).\]
If at least one of $Q$ and $W$ is full rank, then the algorithm returns yes. 
Otherwise, 
 step 6 of the algorithm constructs matrices $J, K$, and $S'$ satisfying Lemma~\ref{lem:s_transform_final}. 
Let $Q''$ be the $(3+n) \times (3+n)$ matrix such that
\[Q''[1, q;1, 3+ n] = Q' \text{ and }Q''[q+1, 3+n; 1, 3+n] = V',\]
and $W''$ be the $m' \times m'$ matrix such that 
\[W''[1, m' - 3 - n + q;1, m'] = W' \text{ and } W''[m' - 2 - n + q, m'; 1, m'] = R'.\]
By Lemma~\ref{lem:s_transform_final}, for each type 2 matrix $F_\ell$, we have 
\[Q'' B_\ell W''^T = \left(\begin{array}{cc} Q' B_\ell W'^T & 0 \\ 0 & V' B_\ell R'^T \end{array}\right),\]
and 
\[J B_\ell' K^T =  \left(\begin{array}{cc} Q' B_\ell W'^T & 0 \\ 0 & R' B_\ell^T V'^T \end{array}\right).\]
Since one necessary condition for two semi-canonical forms being isometric is that 
there are $Q^\dagger \in \GL(3 + n, \mathbb{F}_p)$ and $W^\dagger \in \GL(m', \mathbb{F}_p)$ such that 
$Q^\dagger B_\ell (W^\dagger)^T = B_\ell'$ for all the type 2 matrices $F_\ell$,
by Lemma~\ref{lem:block_matrix_tuple},
the matrix tuples $(V' B_1 R'^T, \dots, V' B_t R'^T)$  and  $(R' B_1^T V'^T, \dots, R' B_t V'^T)$ are isometric. 
Hence, the algorithm returns yes.

If the input two semi-canonical forms are not isometric, then at least one of condition (a) and condition (b) does not hold.
If condition (a) does not hold, then the algorithm returns no at step 1. 
If condition (b) does not hold and the algorithm does not return no at step 3, 
then by Theorem~\ref{thm:skew_tuple_isometry_intro}, step 3 returns a matrix $S$ of form \[\left(\begin{array}{cc} Q & R \\ V & W \end{array} \right)\]
such that $S \FF_{\semic(\mathbf{G})} S^T = \FF_{\semic(\mathbf{H})}$. 
Both $Q$ and $W$ are not full rank because otherwise, by Lemma~\ref{lem:reduction_easy_case}, Lemma~\ref{lem:restricted_tuple_isometry} does not hold.
Hence, step 5 of the algorithm constructs matrices $J, K$, and $S'$ satisfying 
Lemma~\ref{lem:s_transform_final}. 
Let $Q''$ be the $(3+n) \times (3+n)$ matrix such that $Q''[1, q;1, 3+ n] = Q'$ and $Q''[q+1, 3+n; 1, 3+n] = V'$.
Let $W''$ be the $m' \times m'$ matrix such that $W''[1, m' - 3 - n + q;1, m'] = W'$ and $W''[m' - 2 - n + q, m'; 1, m'] = R'$.
By Lemma~\ref{lem:s_transform_final}, for each type 2 matrix $F_\ell$, we have 
\[Q'' B_\ell W''^T = \left(\begin{array}{cc} Q' B_\ell W'^T & 0 \\ 0 & V' B_\ell R'^T \end{array}\right),\]
and 
\[J B_\ell' K^T =  \left(\begin{array}{cc} Q' B_\ell W'^T & 0 \\ 0 & R' B_\ell^T V'^T \end{array}\right).\]
The matrix tuples $(V' B_1 R'^T, \dots, V' B_t R'^T)$  and  $(R' B_1^T V'^T, \dots, R' B_t V'^T)$ are not isometric because otherwise, it contradicts Lemma~\ref{lem:restricted_tuple_isometry}.
By Theorem~\ref{thm:tuple_isometry}
the algorithm returns no. 

The running time of the algorithm is obtained by Theorem~\ref{thm:skew_tuple_isometry_intro}, Theorem~\ref{thm:tuple_isometry}, and Lemma~\ref{lem:s_transform_final}. 
\end{proof}

\section{Proof of  Theorem~\ref{thm:main_group} and Theorem~\ref{thm:main_matrix_space}}\label{sec:final}
We first present our algorithm for the isometry testing of skew-symmetric matrix spaces. 

\begin{framed}
\noindent \textbf{Isometry Testing of Skew-Symmetric Matrix Spaces Algorithm}

\noindent \textbf{Input:} Linear bases for two skew-symmetric matrix spaces $\matrixspaceG, \matrixspaceH \leq SS(n, \mathbb{F}_p)$, both of dimension $m$, for some prime $p > 2$ and positive integers $n, m$.

\noindent \textbf{Output:} Yes or no.

\begin{enumerate}
\item Construct skew-symmetric matrix space tensors $\mathbf{G}$ and $\mathbf{H}$ for  $\matrixspaceG$ and  $\matrixspaceH$,  respectively.
\item Return the output of the algorithm for the isometry testing of skew-symmetric matrix space tensors with $\mathbf{G}$
 and $\mathbf{H}$.
\end{enumerate}
\end{framed}

\begin{proof}[Proof of Theorem~\ref{thm:main_matrix_space}]
The correctness of the algorithm is by Lemma~\ref{lem:matrixspace_tensor_iso_equivalent} and Lemma~\ref{lem:tensor_iso_algo}. 
Now we bound the running time. By Definition~\ref{def:tensor_from_space},  $\matrixspaceG$ and  $\matrixspaceH$ can be constructed in time $\mathrm{poly}(n, m, p)$. 
The running time of the second step of the algorithm is obtained by Lemma~\ref{lem:tensor_iso_algo} and Lemma~\ref{lem:test_semi_canonical_form}.
\end{proof}

\begin{framed}
\noindent \textbf{Isomorphism Testing for $p$-Groups of  Class 2 and Exponent $p$ Algorithm}

\noindent \textbf{Input:} Two $p$-groups $G$ and $H$ of class 2 and exponent $p$ for some prime $p > 2$. 

\noindent \textbf{Output:} Yes or no.

\begin{enumerate}
\item If the orders of the two groups are different, then return no. Otherwise, let $n$ denote the order of group $G$, and $k$ be $\log_p(n)$.

\item If $k \leq (\log_2( p))^5$, then 
\begin{enumerate}
\item Enumerate all the possible $g_1, \dots, g_k \in G$  and $h_1, \dots, h_k \in H$ such that $\{g_1, \dots, g_k\}$ is a generating set of $G$ and $\{h_1, \dots, h_k\}$ is a generating set of $H$. 
\item For each enumeration, if $f(g_i) = h_i$ gives an isomorphism from $G$ to $H$, then return yes.
\item Return no.
\end{enumerate}

\item Construct skew-symmetric matrix spaces  $\matrixspaceG$ and $\matrixspaceH$ for groups $G$ and $H$, respectively, via Baer's correspondence.

\item Return the output of the algorithm for the isometry testing of skew-symmetric matrix spaces with $\matrixspaceG$ and $\matrixspaceH$.

\end{enumerate}
\end{framed}

\begin{proof}[Proof of Theorem~\ref{thm:main_group}]
The correctness of the algorithm is obtained by Theorem~\ref{thm:baer_correspondance} and Theorem~\ref{thm:main_matrix_space}. 
Now we bound the running time of the algorithm. 
If $k \leq (\log_2( p))^5$, then the running time of the algorithm is $p^{O(k^2)}$ because there are $p^k$ elements in each group, and there are $2k$ elements need to enumerate, $k$ elements for $G$ and $k$ elements for $H$.
Since
\[k = k^{5/6} \cdot k^{1/6}  \leq  k^{5/6} \cdot (\log_2(p))^{5/6} = (k\cdot \log_2(p))^{5/6},\]
 we have $p^{O(k^2)} \leq p^{O(k \cdot  (k\cdot \log_2(p))^{5/6})} = n^{O((\log n)^{5/6})}$.
Hence, the running time for this case is $ n^{O((\log n)^{5/6})}$.

If $k > (\log_2( p))^5$, then by Theorem~\ref{thm:main_matrix_space}, the running time of the algorithm 
is $p^{O(k^{1.8} \cdot \log_2 (p))}$.  
Since 
\[k^{0.8} \cdot \log_2(p) = k^{0.8} \cdot (\log_2 (p))^{1/6} \cdot (\log_2(p))^{5/6} < k^{0.8} \cdot k^{1/30} \cdot (\log_2(p))^{5/6} = (k\cdot \log_2(p))^{5/6},\]
 we have $p^{O(k^{1.8} \cdot \log_2 (p))}  \leq p^{O(k \cdot  (k\cdot \log_2(p))^{5/6})} = n^{O((\log n)^{5/6})}$.
Hence, the running time for this case is also $ n^{O((\log n)^{5/6})}$.
\end{proof}

\bibliography{isomorphism}
\bibliographystyle{alpha}

\end{document}